\RequirePackage{fix-cm}
\documentclass[smallcondensed]{svjour3}     
\smartqed  

\usepackage{graphicx}

\usepackage[utf8]{inputenc} 
\usepackage[T1]{fontenc}    
\usepackage{hyperref}       
\usepackage{url}            
\usepackage{booktabs}       
\usepackage{amsfonts}       
\usepackage{nicefrac}       
\usepackage{microtype}      

\usepackage{enumitem}
\usepackage{algorithm} 
\usepackage{algorithmic}
\usepackage{subfigure}
\usepackage{tabularx}
\usepackage[breakable,skins]{tcolorbox} 

\usepackage{fancyvrb}
\usepackage{hyperref}
\usepackage{amssymb} 
\usepackage{amsmath}
\usepackage{soul} 
\usepackage{mathtools}
\usepackage{bbm}
\usepackage{xspace}
\usepackage{multirow}
\usepackage{xr}

\usepackage{tikz}
\usetikzlibrary{positioning}

\def \ZN2Z{\mathbb{Z}_{N^2}^*}

\def \Z{\mathbb{Z}}

\newcommand{\error}{E}
\newcommand{\err}[1]{\error_{#1}}
\newcommand{\comschpp}{\Theta}
\newcommand{\G}{\mathbb{G}}
\newcommand{\rel}{\mathcal{R}}
\newcommand{\Prov}{\mathcal{P}}
\newcommand{\Verif}{\mathcal{V}}
\newcommand{\wit}{w}
\newcommand{\stat}{a}
\newcommand{\pOrd}{q}
\newcommand{\gexp}{\textsf{GEX}}
\newcommand{\prot}{\Pi}
\newcommand{\vv}[1]{\bar{#1}}

\newcommand{\gdprec}{{\psi}\xspace}
\newcommand{\enc}[1]{\langle #1 \rangle}
\newcommand{\erf}{\hbox{erf}\xspace}
\newcommand{\erfc}{\hbox{erfc}\xspace}
\newcommand{\Lerfc}{L_\mathrm{erfc}}
\newcommand{\xminerfc}{x_{\mathrm{erfc}}^{min}}
\newcommand{\Lerf}{L_{\mathrm{erf}}}

\newcommand{\Eerfround}{E^{round}_{\mathrm{erf}}}

\newcommand{\Eerfcround}{E^{round}_{\mathrm{erfc}}}
\newcommand{\errerfc}[1]{F_{#1}} 

\newcommand{\comb}[2]{\left(\begin{array}{c}{#1}\\{#2}\end{array}\right)}
\newcommand{\userset}{U}
\newcommand{\avgx}{X^{avg}}

\newcommand{\getsR}{\gets_R}

\newcommand{\gopa}{\textsc{Gopa}\xspace}

\newcommand{\edgeset}{E}

\newcommand{\honestuserset}{\userset^{H}}
\newcommand{\honestedgeset}{\edgeset^{H}}

\newcommand{\onlineuserset}{\userset^{O}}

\newcommand{\honestfraction}{\rho}

\newcommand{\Sigmaginv}{\Sigma^{(-g)}}

\newcommand{\rndG}{G^H} 
\newcommand{\rndroot}{v_1} 
\newcommand{\rndEG}{E^H} 
\newcommand{\rndn}{n_H} 
\newcommand{\rndk}{k_H} 
\newcommand{\rndT}{T} 
\newcommand{\rndET}{E_T} 
\newcommand{\rndq}{q} 
\newcommand{\rndF}{F} 
\newcommand{\rndEF}{E_F} 
\newcommand{\rndD}{\Delta} 
\newcommand{\rnddelta}{\delta} 
\newcommand{\rnddeltaa}{\delta_F} 
\newcommand{\rnddeltab}{\delta_B} 
\newcommand{\rndphi}{\phi} 
\newcommand{\rndgamma}{\gamma} 
\newcommand{\rndfailF}{\text{\textsc{Fail}}_F}
\newcommand{\rndlev}{z}
\newcommand{\rndclev}{Z}

\newcommand{\rndlll}[2]{{#2}}

\newcommand{\rndm}{m}
\newcommand{\rndp}{p}
\newcommand{\rndzeta}{\zeta}
\newcommand{\rndfailB}{\text{\textsc{Fail}}_B}

\newcommand{\rnddcnt}{l}
\newcommand{\rnddsum}{\rndm}
\newcommand{\rndd}{d}
  \newcommand{\rndGbis}{{G^\prime}}
  \newcommand{\rndpbis}{{\rndp^\prime}}

\newcommand{\groupelsize}{W} 

\spnewtheorem{assumption}{Assumption}{\bf}{\it}
\spnewtheorem{fact}{Fact}{\bf}{\it}

\DeclareMathOperator*{\argmin}{arg\,min}

\definecolor{sapphire}{rgb}{0.03, 0.15, 0.4}

\newtcolorbox{cbox}[3][]
{
	colframe = #2!25,
	colback  = #2!10,
	coltitle = #2!20!black,  
	title    = {#3},
	#1,
	breakable,
	enhanced
}

\newcommand{\twopiconst}{C_1}

\journalname{Machine Learning Journal}
\begin{document}
	
	\title{An Accurate, Scalable and Verifiable Protocol for Federated
		Differentially Private Averaging
	}
	
	\titlerunning{An Accurate, Scalable and Verifiable Protocol for Federated
		DP Averaging}        
	
	\author{C\'esar Sabater  \and
		Aur\'elien Bellet \and 
		Jan Ramon 
	}
	
	
	\institute{C. Sabater  \at
		INRIA Lille - Nord Europe, 
		59650 Villeneuve d’Ascq, France \\
		\email{cesar.sabater@inria.fr}           
		\and
		A. Bellet  \at
		INRIA Lille - Nord Europe, 
		59650 Villeneuve d’Ascq, France \\
		\email{aurelien.bellet@inria.fr} 
		\and
		J. Ramon \at 
		INRIA Lille - Nord Europe, 
		59650 Villeneuve d’Ascq, France \\ 
		\email{jan.ramon@inria.fr} 
	}
	
	\date{Received: date / Accepted: date}

	\maketitle





\begin{abstract}
Learning from data owned by several parties, as in federated learning, raises
challenges regarding the privacy guarantees provided to participants and the
correctness of the computation in the presence of malicious parties. We tackle
these challenges in the context of distributed averaging, an essential
building block of federated learning algorithms. Our first contribution is a
scalable protocol in which participants exchange correlated Gaussian noise
along the edges of a graph, complemented by independent noise added by
each party. We analyze the differential privacy guarantees of our protocol and
the impact of the graph topology under colluding malicious parties, showing
that we can nearly match the utility of the trusted curator model even when
each honest party communicates with only a logarithmic number of other parties
chosen at random. This is in contrast with protocols in the local model of
privacy (with lower utility) or based on secure aggregation (where all pairs
of users need to exchange messages). Our second contribution enables users to
prove the correctness of their computations without compromising the
efficiency and privacy guarantees of the protocol. Our construction
relies on standard cryptographic primitives like commitment schemes and zero knowledge proofs.
\end{abstract}


\section{Introduction}
\label{sec:intro}

Individuals are producing ever growing amounts of personal data, which in turn
fuel innovative services based on machine
learning (ML). The classic \emph{centralized} paradigm consists in collecting,
storing and analyzing this data on a (supposedly trusted) central server or in
the cloud,  
which poses well documented privacy risks for the users.
With the increase of public awareness and regulations, we are
witnessing a shift towards a more \emph{decentralized} paradigm where personal
data remains on each user's device, as can be seen from the growing popularity
of federated learning
\cite{kairouz2019advances}. In this setting, users typically do not trust the
central server (if any), or each other, which introduces new issues regarding
privacy and security. First, the information shared by users during the
decentralized training protocol can reveal a lot
about their private data (see \cite{Melis2019,Shokri2019,inverting} for inference attacks on federated
learning).
Formal guarantees such as differential privacy (DP) \cite{Dwork2006a} are
needed to provably mitigate this and convince users to participate.
Second, \emph{malicious} users may send incorrect results to bias the
learned model in arbitrary ways \cite{Hayes2018,Bhagoji2019,Bagdasaryan2020}.
Ensuring the correctness of the computation is crucial to persuade service
providers to move to a more decentralized and privacy-friendly setting.

In this work, we tackle these challenges in the context of private distributed
averaging. In this canonical problem, the objective is
to privately compute an estimate of the average of values owned by many users who
do not want to disclose them.
Beyond simple data analytics, distributed averaging is of high
relevance to modern ML. Indeed, it is the key primitive
used to aggregate
user updates in gradient-based distributed and federated learning algorithms 
\cite{minibatch_local_sgd,local_sgd,mcmahan2016communication,Jayaraman2018,cp-sgd,kairouz2019advances}. It also
allows to train ML models whose
sufficient statistics are averages (e.g., linear models and decision trees).
Distributed averaging with differential privacy guarantees has thus
attracted a lot of interest in recent years.
In the strong model of local differential privacy (LDP)
\cite{Kasiviswanathan2008,d13,kairouz2015secure,Kairouz2016a,discrete_dis_local,trilemma_neurips20}, each user randomizes its input locally before
sending it to an untrusted aggregator.
Unfortunately, the best possible error for the estimated average with $n$
users is of the order $O(\sqrt{n})$ larger than in
the centralized model of DP where a trusted curator
aggregates data in the clear and perturbs the output \cite{Chan2012}. To
fill this gap,
some work has explored relaxations of LDP that make it possible to match the
utility of the trusted curator model. This is achieved through the use of
cryptographic primitives such as 
 secure aggregation
\cite{Dwork2006ourselves,ChanSS12,Shi2011,Bonawitz2017a,Jayaraman2018} and
secure shuffling \cite{amp_shuffling,Cheu2019,Balle2020,Ghazi2020ICML}. 
Many of these solutions however assume that all users truthfully follow the
protocol (they are \emph{honest-but-curious}) and/or give significant
power (ability to reveal sensitive data)
to a small number of servers. Furthermore, their practical implementation poses
important challenges when the number of parties is large: for instance, the
popular secure aggregation approach of \cite{Bonawitz2017a} requires all $O
(n^2)$
pairs of users to exchange messages.

In this context, our contribution is fourfold.
\begin{itemize}
  \item First, we propose
\gopa, a novel decentralized differentially private averaging protocol that
relies on users exchanging (directly or through a
server) some pairwise-correlated Gaussian noise terms
along the edges of a graph so as to mask
their private values without affecting the global average. This ultimately
canceling noise is complemented by the addition of independent 
(non-canceling) Gaussian noise by each user.
\item Second, we analyze the differential privacy guarantees of \gopa. 
  Remarkably, we establish that our approach can achieve nearly the same
  privacy-utility trade-off as a trusted curator who would average the values
  of honest users, provided that the graph of honest-but-curious users is
  connected and the pairwise-correlated noise variance is large enough. In
  particular, for $n_H$ honest-but-curious users and any fixed DP guarantee,
  the
  variance of the estimated average is only $n_H/(n_H-1)$ times larger than with a trusted curator, a factor which goes to $1$ as $n_H$ grows.  We further show that if the graph is well-connected, the pairwise-correlated noise variance can be significantly reduced.
\item Third,
  to ensure both scalability and robustness to malicious users, we propose a
  randomized procedure in which each user communicates with only a 
  \emph{logarithmic} number of other users while still matching the
  privacy-utility trade-off of the trusted curator. Our analysis is novel and
  requires to leverage and adapt results
  from random graph theory on embedding spanning trees in random graphs.
  Additionally, we show our protocol is robust to a fraction of users dropping out.
\item Finally, we propose
  a procedure to make \gopa verifiable by untrusted external parties, i.e., to
  enable users to prove the correctness of their computations
  without compromising the efficiency or the privacy guarantees of the protocol.
  To the best of our knowledge, we are the first to propose such a procedure. 
  It offers a strong preventive countermeasure against various attacks such as
  data poisoning or protocol deviations aimed at reducing utility.  
Our construction relies on commitment schemes and zero knowledge proofs (ZKPs), which are very popular in auditable electronic payment systems and cryptocurrencies. These cryptographic primitives
scale well both in communication and computational requirements and are
perfectly suitable in our untrusted decentralized setting. We use classic ZKPs
to design a procedure for the generation of noise with verifiable
distribution, and ultimately to prove the correctness of the final computation 
(or detect malicious users who did not follow the protocol). Crucially, the privacy guarantees of the protocol are not compromised by this procedure,
while the integrity of the computation relies on a standard discrete logarithm assumption. In the end, we argue that our protocol offers correctness guarantees that are essentially equivalent to the case where a trusted curator would hold the private data of users.
\end{itemize}

The paper is organized as follows.
Section~\ref{sec:setting} introduces the problem setting.
We discuss the related work in more details in Section~\ref{sec:related}.
The \gopa protocol is introduced in Section~\ref{sec:gopa} and we analyze its
differential privacy guarantees in Section~\ref{sec:privacy}. We present our procedure to ensure correctness against malicious behavior in
Section~\ref{sec:verif}, and summarize computational and communication costs
in Section~\ref{sec:complexity}.
Finally, we present some experimental results in
Section~\ref{sec:exp} and conclude with future lines of research in Section~\ref{sec:conclu}.


\section{Notations and Setting}
\label{sec:setting}

We consider a set $\userset=\{1,\dots,n\}$ of $n\geq 3$ users (parties). Each
user $u\in
 U$ holds a private value $X_u$, which can be thought of as being computed
 from the private dataset of user $u$. We assume that $X_u$ lies in a bounded
 interval of $\mathbb{R}$ (without loss of generality, we assume $X_u\in[0,
 1]$). The extension to the vector case is straightforward. We denote by $X$
 the column vector $X=[X_1,\dots,X_n]^\top\in[0,1]^n$ of private values.
Unless otherwise noted, all vectors are column vectors.
Users communicate over a network represented by a connected
undirected graph $G=(U,E)$, where $\{u,v\}\in E$ indicates that users $u$ and
$v$ are neighbors in $G$ and can exchange secure messages.
For a given user $u$, we denote by $N(u)=\{v : \{u,v\}\in E\}$ the set of its neighbors.
We note that in settings where users can only communicate with a server,
the latter can act
as a relay that forwards (encrypted and authenticated) messages between users,
as done in secure aggregation \cite{Bonawitz2017a}.

The users aim to
collaboratively estimate the average value $\avgx = \frac{1}{n}\sum_{u=1}^n
X_u$
without revealing their individual private values. Such a protocol can be
readily used to
privately execute distributed ML algorithms that interact
with data through averages over values computed locally by the
participants, but do not actually need to see the individual values. We
give two concrete examples below.

\begin{example}[Linear regression]
\label{ex:linearreg}
Let $\lambda\geq 0$ be a public parameter. Each user $u$ holds a private
feature vector $\phi_u=[\phi_u^1,\dots,\phi_u^d]\in\mathbb{R}^d$ and a private
label $y_u\in\mathbb{R}$.
The goal is to solve a ridge regression task, i.e. find
$\theta^*\in\argmin_\theta \frac{1}{n}\sum_{u\in U}
(\phi_u^\top\theta - y_u)^2 + \lambda\|\theta\|^2$. $\theta^*$ can be
computed
in closed form from the quantities $\frac{1}
{n}\sum_{u\in U} \phi_u^iy_u$ and $\frac{1}{n}\sum_{u\in U}\phi_u^i\phi_u^j$
for all $i,j\in
\{1,\dots,d\}$.
\end{example}

\begin{example}[Federated ML]
\label{ex:federatedml} 
In federated learning \cite{kairouz2019advances} and 
distributed empirical risk minimization, each user $u$ holds a
private dataset $\mathcal{D}_u$ and the goal is to find $\theta^*$ such that
$\theta^*\in\argmin_\theta \frac{1}{n}\sum_{u\in U}f(\theta; \mathcal{D}_u)$
where $f$ is some loss function.
Popular algorithms 
\cite{minibatch_local_sgd,local_sgd,mcmahan2016communication,Jayaraman2018,cp-sgd}
all follow the same high-level procedure: at round $t$, each user $u$ computes
a local update $\theta_u^t$ based on $
\mathcal{D}_u$ and the current global model $\theta^{t-1}$, and the updated
global model is computed as $\theta^t=\frac{1}{n}\sum_u \theta_u^t$.
\end{example}

\noindent\emph{Threat model.}
We consider two commonly adopted adversary models formalized
by \cite{goldreich1998secure} and used in the design of
many secure protocols.
A \emph{honest-but-curious} (\emph{honest} for short) user will follow
the protocol specification, but may use all the information obtained
during the execution to infer information about other users. A honest user may
accidentally drop out at any point of the execution (in a way that is
independent of the private values $X$).
On the other hand, a \emph{malicious user} may deviate from the protocol
execution (e.g, sending incorrect values or dropping out on
purpose). Malicious users can collude, and thus will be seen as a single
malicious party (the \emph{adversary}) who has access to all information
collected by malicious users. Our privacy guarantees will
hold under
the assumption that honest users communicate through secure channels, while
the correctness of our protocol will be guaranteed under some form of the
Discrete Logarithm Assumption (DLA), a
standard assumption in cryptography.

For a given execution of the protocol, we denote by $U^O$ the set of the users
who remained online until the end (i.e., did not drop out). Users in $U^O$ are
either honest
or malicious: we denote
by $U^H\subseteq U^O$ those who are honest, by
$\rndn=|\honestuserset|$ their number and by
$\honestfraction=\rndn/n$ their proportion with respect to the total number of
users.
We also denote by
$\rndG = (\honestuserset, \honestedgeset)$ the subgraph of $G$
induced by the set of honest users $\honestuserset$, i.e., $\honestedgeset = \{\{u,v\} \in E : u,v
\in \honestuserset\}$. The properties of $G$ and $\rndG$ will play a key
role in the privacy and scalability guarantees of our protocol.\\

\noindent\emph{Privacy definition.} 
Our goal is to design a protocol that satisfies differential privacy (DP)
\cite{Dwork2006a}, which has become a gold standard in
private information release.

\begin{definition}[Differential privacy]
\label{def:dp}
Let $\varepsilon>0,\delta\geq 0$. A (randomized) protocol $\mathcal{A}$ is $
(\varepsilon,\delta)$-differentially private if for all neighboring datasets
$X=[X_1,\dots,X_n]$ and $X'=[X'_1,\dots,X_n]$ differing only in
a single data point, and for
all sets of possible outputs $\mathcal{O}$, we have:
\begin{equation}
\label{eq:dp}
\Pr(\mathcal{A} (X) \in \mathcal{O}) \leq e^
{\varepsilon}\Pr(\mathcal{A} ( X') \in \mathcal{O}) + \delta.
\end{equation}
\end{definition}


\section{Related Work}
\label{sec:related}

\begin{table}
\centering
  \begin{tabular}{lllcl}
    \hline    
    Approach & Com. per party & MSE                 & Verif & Risks \\\hline
    Central DP \cite{Dwork2014a} & $O(1)$ & $O(1/n^2)$ & No & Trusted curator
    \\
    Local DP \cite{Kasiviswanathan2008} & $O(1)$ & $O(1/n)$ & No & 
    \\
    Verifiable secret sharing \cite{Dwork2006ourselves} & $O(n)$ & $O(1/n^2)$
    & Yes & \\
    Secure agg. \cite{Bonawitz2017a} + DP 
    \cite{pmlr-v139-kairouz21a,skellam} & $O
    (n)$ & $O(1/n^2)$
    &
    No & Honest users\\
    CAPE \cite{imtiaz2021distributed} & $O(n)$ & $O(1/n^2)$ &  No &  \\
    Shuffling  \cite{Balle2020} & $O(1 + \log(1/\delta))$ & $O(1/n^2)$ & No &
    Trusted shuffler \\
    \hline
    \textbf{GOPA (this work)}    &  $O(\log n)$ & $O(1/n^2)$ &
    Yes &
    \\
    \hline
  \end{tabular}
  \caption{\label{tab:relwork}Comparison of GOPA with previous DP averaging
  approaches with their communication cost per party, mean squared error (MSE),
  verifiability (Verif) and additional risks.}
  \end{table}

In this section we review the most important work related to ours.  A set of
key approaches together with their main features are summarized
in Table \ref{tab:relwork}.
  
Distributed averaging is a key subroutine in distributed and federated
learning 
\cite{minibatch_local_sgd,local_sgd,mcmahan2016communication,Jayaraman2018,cp-sgd,kairouz2019advances}. Therefore, any improvement in the privacy-utility-communication trade-off for averaging implies gains for many ML approaches downstream.

Local differential privacy (LDP) 
\cite{Kasiviswanathan2008,d13,kairouz2015secure,Kairouz2016a,discrete_dis_local} requires users to locally randomize their input before they send it to an untrusted
aggregator. This very strong model of privacy
comes at a significant cost in utility:
the best possible mean squared is of order $1/n^2$ in
the trusted curator model while it is of order $1/n$ in LDP
\cite{Chan2012,trilemma_neurips20}. This limits the usefulness of the
local model to industrial settings where the
number of participants is huge \cite{rappor15,telemetry17}.
Our approach belongs to the recent line of work which attempts to relax the
LDP model so as to improve utility without relying on a
trusted curator (or similarly on a small fixed set of parties).

Previous work considered the use of
cryptographic primitives like secure aggregation protocols, which can be used to compute the (exact) average of private values \cite{Dwork2006ourselves,Shi2011,Bonawitz2017a,ChanSS12}. While secure aggregation allows in principle to recover the
utility of the trusted curator model, it suffers three main drawbacks.
Firstly, existing protocols require $\Omega(n)$ communication per party, which
is hardly
feasible beyond a few hundred or thousand users. In contrast, we propose a protocol which requires only $O(\log n)$
communication.\footnote{We note that, independently and in parallel to our
work, \cite{bell_secure_2020} recently proposed a
secure aggregation protocol with $O(\log n)$ communication at the cost of relaxing the functionality under colluding/malicious users.}
Secondly, combining secure aggregation with DP is nontrivial as the noise must be added
in a distributed fashion and in the discrete domain. Existing complete
systems
\cite{pmlr-v139-kairouz21a,skellam} assume an ideal secure aggregation
functionality which does not reflect the impact of colluding/malicious
users. In these more challenging settings, it is not clear how to add the
necessary noise for DP and what the resulting privacy/utility trade-offs would
be. Alternatively, \cite{Jayaraman2018} adds the noise within the secure
protocol
but relies on two non-colluding servers. Thirdly, most of the above schemes
are not verifiable. One exception is
the verifiable secret sharing approach of \cite{Dwork2006ourselves}, which
again induces $\Omega(n)$ communication. 
Finally, we note that secure aggregation typically uses uniformly distributed
pairwise masks, hence a single
residual term completely destroys the utility. In contrast, we use Gaussian
pairwise masks that have zero mean and bounded variance, which provides more
robustness but requires the more involved privacy analysis we present in
Section~\ref{sec:privacy}.

Recently, the shuffle model of privacy
\cite{Cheu2019,amp_shuffling,Hartmann2019,Balle2020,Ghazi2020ICML},
where inputs are passed to a trusted/secure shuffler that obfuscates the
source of
the messages, has been studied theoretically as an intermediate point
between the local and trusted curator models.
For differentially private averaging, the shuffle model allows to match the
utility of the trusted curator setting \cite{Balle2020}.
However, practical implementations of secure shuffling are not discussed in
these works.
Existing solutions typically rely on multiple layers of routing servers 
\cite{Dingledine2004} with high communication overhead and non-collusion
assumptions.
Anonymous communication is also potentially at odds
with the identification of malicious parties. To the best of our
knowledge, all protocols for averaging in the shuffle model assume
honest-but-curious parties. 

The protocol proposed in \cite{imtiaz2021distributed} uses correlated
Gaussian noise to achieve trusted curator utility for averaging, but
the dependence structure of the noise must be only at the global level (i.e.,
noise terms sum to zero over all users). Generating such noise actually
requires a call to a secure
aggregation primitive, which incurs $\Omega(n)$ communication per party as
discussed above. In contrast, our pairwise-canceling noise terms can be
generated with only $O(\log n)$ communication. Furthermore, \cite{imtiaz2021distributed}
assume honest parties, while our protocol
is robust to malicious participants.

In summary, an original aspect of our work is to match the
privacy-utility trade-off of the trusted curator model
at a relatively low cost without requiring to trust a fixed small set of
parties.  By spreading trust over sufficiently many parties, we ensure that
even in the unlikely case where many parties
collude they will not be able to infer much sensitive
information, reducing the incentive to collude.
We are not aware of other differential privacy work sharing this feature.
Overall, our protocol provides a unique combination of three
important properties: (a) utility of same order as trusted curator setting, 
(b)
logarithmic
communication per user, and (c) robustness to malicious users.


\section{Proposed Protocol}
\label{sec:gopa}

In this section we describe our protocol called \gopa (GOssip noise
for Private Averaging).
The high-level idea of \gopa is to have each user $u$ mask its private value
by
adding two different types of noise. The first type is a sum of
pairwise-correlated noise terms $\Delta_{u,v}$ over the set of neighbors $v\in
N(u)$ such that each $\Delta_{u,v}$ cancels out with the $\Delta_{v,u}$ of user $v$ in the final
result. The second type of noise is an independent term $\eta_u$ which does
not cancel out. At the end of the protocol, each user has
generated a noisy version $\hat{X}_u$ of its private value $X_u$, which takes
the following form:
\begin{equation}
\label{eq:noisy}
\hat{X}_u = X_u + \textstyle\sum_{v\in N(u)} \Delta_{u,v} + \eta_u.
\end{equation}
Algorithm \ref{alg:rndphase} presents the
detailed steps.
Neighboring nodes $\{u,v\}\in E$
contact each other to draw a real number
from the Gaussian distribution $\mathcal{N}(0,\sigma_\Delta^2)$, that $u$
adds to its private value and $v$ subtracts.
Intuitively, each
user thereby distributes noise masking its private value across its neighbors
so that even if some of them are malicious
and collude, the remaining noise values will be enough to provide the desired
privacy guarantees.
The idea is reminiscent of uniformly random pairwise masks in secure
aggregation \cite{Bonawitz2017a} but we use Gaussian noise and restrict
exchanges to the edges of the graph
instead of requiring messages between all pairs of users.
As in gossip algorithms \cite{random_gossip}, the pairwise
exchanges can be performed asynchronously and in parallel.
Additionally, every user $u \in U$ adds an independent noise term $\eta_u\sim
\mathcal{N}{(0, \sigma_\eta^2)}$ to its private value.
This noise will ensure that the final estimate of the average satisfies
differential
privacy (see Section~\ref{sec:privacy}).
The pairwise and independent noise variances $\sigma_\Delta^2$ and
$\sigma_\eta^2$ are public parameters of the protocol.

\begin{algorithm}[t]
\floatname{algorithm}{Algorithm}
  \caption{\gopa protocol}
     \text{\textbf{Input:} $G=(\userset, E)$, $(X_u)_
     {u\in \userset}$, $\sigma_\Delta^2, \sigma_\eta^2 \in 
     \mathbb{R}^+$} 
  \begin{algorithmic}[1]
    \FORALL{neighbor pairs $\{u,v\}\in E$ s.t. $u<v$}
    \STATE{$u$ and $v$ draw a random $y \sim\mathcal{N}(0,
    \sigma_\Delta^2)$ and set $\Delta_{u,v}\gets y$, $\Delta_{v,u}\gets -y$}
    \ENDFOR
    \FORALL{users $u\in\userset$}
    \STATE{$u$ draws a random $\eta_u \sim \mathcal{N}(0, \sigma_\eta^2)$
    and reveals noisy value $\hat{X}_u\gets  X_u+\sum_{v\in N(u)} \Delta_
    {u,v}+\eta_u$}
    \ENDFOR
  \end{algorithmic}
  \label{alg:rndphase}
\end{algorithm}

\noindent\emph{Utility of \gopa.}
The protocol generates a set of noisy values $\hat{X}=[\hat{X}_1,\dots,\hat{X}_n]^\top$ which
are then publicly released. They can be sent to
an untrusted aggregator, or averaged in a decentralized way via gossiping 
\cite{random_gossip}. In any case,
the estimated average is given by
$\hat{X}^{avg} = \frac{1}{n}\sum_{u\in U} \hat{X}_u = \avgx + 
\frac{1}{n}\sum_{u \in U}
\eta_u,$ which
has expected value $\avgx$ and variance $\sigma_\eta^2/n$.
Recall that the local model of DP, where each user releases a locally
perturbed input without communicating with other users, would require
$\sigma_\eta^2=O(1)$. In contrast, we would like the total amount of
independent noise to be of order $O(1/n_H)$ as needed to protect the
average of honest users with the standard Gaussian mechanism in the
trusted curator model of DP \cite{Dwork2014a}.
We will show in Section~\ref{sec:privacy} that we can achieve this
privacy-utility trade-off by choosing an
appropriate variance $\sigma_\Delta^2$ for
our pairwise
noise terms.\\

\noindent\emph{Dealing with dropout.} A
user $u\notin\onlineuserset$ 
who drops out during the execution of the protocol does not actually publish
any noisy value (i.e., $\hat{X}_u$ is empty). The estimated average is thus
computed
by averaging only over the noisy values of users in $\onlineuserset$.
Additionally, any
residual noise term that a user $u\notin\onlineuserset$ may have exchanged
with a user $v\in\onlineuserset$ before dropping out can be ``rolled back''
by having $v$ reveal
$\Delta_{u,v}$ so it can be subtracted from the result (we will ensure this does not threaten privacy by having sufficiently many neighbors,
see Section~\ref{sec:graphs}). We can thus obtain an
estimate of $\frac{1}{|\onlineuserset|}\sum_{u\in \onlineuserset} X_u$
with variance $\sigma_\eta^2/|\onlineuserset|$. Note that even
if some residual noise terms are not rolled back, e.g. to avoid
extra communication, the estimate remains unbiased (with a
larger variance that depends on $\sigma_\Delta^2$). This is a rather unique
feature of \gopa which comes from the use of Gaussian noise rather than the
uniformly random noise used in
secure aggregation 
\cite{Bonawitz2017a}. We
discuss strategies to handle users dropping out in more details in
Appendix~\ref{app:crypto.drop-out}.


\section{Privacy Guarantees}
\label{sec:privacy}

Our goal is to prove differential privacy guarantees for \gopa.
First, we develop in Section~\ref{sec:priv.general} a general result providing
privacy guarantees as a function of the structure of the communication graph
$G^H$, i.e., the subgraph of $G$ induced by $U^H$.  Then, in Sections 
\ref{sec:priv.worst-dp} and \ref{sec:priv.completeGraph}, we  study the
special cases of the path graph and the complete graph respectively, showing
they are the worst and best cases in terms of privacy. Yet, we show
that as long as $G^H$ is connected and the variance $\sigma_\Delta^2$ for the pairwise (canceling) noise is large enough \gopa can (nearly) match the privacy-utility trade-off of the trusted curator setting.
In Section~\ref{sec:priv.random}, we propose a randomized procedure to
construct the graph $G$ and show that it strikes a good balance between
privacy and communication costs.
In each section, we first discuss the result and its consequences, and then
present the proof.
In Section~\ref{sec:priv.summary}, we summarize our results and provide
further discussion.

\subsection{Effect of the Communication Structure on Privacy}
\label{sec:priv.general}

The strength of the privacy guarantee we can prove depends on the
communication graph $G^H$ over honest users. Intuitively, this is because the
more terms $\Delta_{u,v}$ a given honest user $u$ exchanges with other honest
users $v$, the more he/she spreads his/her secret over others and the more
difficult it becomes to estimate the private value $X_u$.
We first introduce in Section~\ref{sec:priv.prelim} a number of preliminary
concepts.
Next, in Section~\ref{sec:generic-dp}, we prove an abstract
result, Theorem~\ref{thm:diffpriv.etadelta}, which gives DP guarantees for
\gopa that depend on the choice of a labeling $t$ of the graph $G^H$.

In Section \ref{sec:generic-dp.discuss} we discuss a number of implications of
Theorem~\ref{thm:diffpriv.etadelta} which provide some insight into the
dependency between the structure of $G^H$ and the privacy of \gopa{}, and will turn out helpful in the proofs of Theorems \ref{thm:diffprivacy}, \ref{thm:diffprivacy-complete} and \ref{thm:diffprivacy-random}.

\subsubsection{Preliminary Concepts}
\label{sec:priv.prelim}

Recall that each user $u\in\onlineuserset$ who does not drop out generates $\hat{X}_u$ from its private value $X_u$ by adding pairwise noise terms $\bar{\Delta}_u=\sum_{v\in N(u)} \Delta_{u,v}$ (with $\Delta_{u,v}+\Delta_{v,u}=0$) as well as independent noise $\eta_u$. All random variables $\Delta_{u,v}$ (with $u<v$) and $\eta_u$ are independent.
We thus have the system of linear equations
\[{\hat{X}} = X+\bar{\Delta}+\eta,\]
where $\bar{\Delta}=(\bar{\Delta}_u)_{u\in \onlineuserset}$ and $\eta=(\eta_u)_{u\in \onlineuserset}$.

We now define the knowledge acquired by the adversary (colluding malicious users) during a given execution of the protocol. It consists of the following:
\begin{enumerate}[label=\roman*.]
\item the noisy value $\hat{X}_u$ of all users $u\in\onlineuserset$ who did not drop out,
\item the private value $X_u$ and the noise $\eta_u$ of the malicious users, and
\item all $\Delta_{u,v}$'s for which $u$ or $v$ is malicious.
\end{enumerate}
We also assume that the adversary knows the full network graph $G$ and all the pairwise noise terms exchanged by dropped out users (since they can be rolled back, as explained in Section~\ref{sec:gopa}).
The only unknowns are thus the private value $X_u$ and independent noise $\eta_u$ of each honest user $u\in U^H$, as well as the $\Delta_{u,v}$ values exchanged between pairs of honest users $\{u,v\}\in E^H$.

Letting $N^H(u)=\{v : \{u,v\}\in E^H\}$, from the above knowledge the adversary can subtract $\sum_{v \in N(u) \setminus N^H(u)} \Delta_{u,v}$ from $\hat{X}_u$ to obtain
\[\hat{X}_u^H = X_u + \sum_{u \in N^H (u)} \Delta_{u,v} + \eta_u\]
for every honest $u \in U^H$. The view of the adversary can thus be summarized by the vector $\hat{X}^H = (\hat{X}^H_u)_{u \in U^H}$ and the correlation between its elements.
Let ${\hat{X}}^H_u = {\hat{X}}_u - \sum_{v\in N(u)\setminus N^H(u)} \Delta_{u,v}$.
Let $X^H=(X_u)_{u\in U^H}$ be the vector of private values restricted to the honest users and similarly $\eta^H=(\eta_u)_{u\in U^H}$.
\newcommand{\dirEH}{{{\vec{E}}^H}}
Let the directed graph $(U^H,\dirEH)$ be an arbitrary orientation of the undirected graph $G^H=(U^H,E^H)$, i.e., for every edge $\{u,v\}\in E^H$, the set $\dirEH$ either contains the arc $(u,v)$ or the arc $(v,u)$.
  For every arc $(u,v)\in \dirEH$, let $\Delta_{(u,v)} = \Delta_{u,v} = -\Delta_{v,u}$.  Let $\Delta^H=(\Delta_e^H)_{e\in \dirEH}$ be a vector of pairwise noise values indexed by arcs from $\dirEH$.
 Let $K\in\mathbb{R}^{U^H\times \dirEH}$ denote the oriented incidence matrix of the graph $G^H$, i.e., for $(u,v)\in \dirEH$ and $w\in U^H\setminus \{u,v\}$ there holds $K_{u,(u,v)}=-1$, $K_{v,(u,v)}=1$ and $K_{w,(u,v)}=0$.
 In this way, we can rewrite the system of linear equations as
 \begin{equation}{\hat{X}}^H = X^H+K\Delta^H+\eta^H. \end{equation}

 Now, adapting differential privacy (Definition~\ref{eq:dp}) to our setting,
 for any input
$X$ and any possible
outcome $\hat{X}$, we need to compare the probability of the
outcome being equal to $\hat{X}$ when a (non-malicious) user $v_1\in
U$ participates in
the computation with private value $X_{v_1}^A$ to the probability of obtaining the
same outcome when the value of $v_1$ is exchanged with an arbitrary value
$X_{v_1}^B\in [0,1]$.
Since honest users drop out independently of $X$ and do not
reveal anything about their private value when they drop out, in our analysis
we will fix an execution of the protocol where some set
$\honestuserset$ of $n_H$
honest users have remained online until the end of the protocol.
For notational simplicity, we denote by $X^A$ the vector of private
values $(X_u)_{u \in U^H}$ of these honest users in which a user $v_1$ has
value $X_
{v_1}^A$,
and by $X^B$ the vector where $v_1$ has value $X_{v_1}^B$. $X^A$ and
$X^B$ differ in only in the $v_1$-th coordinate, and their maximum difference is $1$.

All noise variables are zero mean, so the expectation and covariance matrix of $\hat{X}^H$ are respectively given by:
\newcommand{\SigmaX}{\hbox{var}\left({\hat{X}}^H\right)}
\[
  \mathbb{E}\left[{\hat{X}}^H\right] = X^H,\quad
  \SigmaX = \sigma_\eta^2 I_{U^H} + \sigma_\Delta^2 L, \]
where $I_{U^H}\in\mathbb{R}^{\rndn\times\rndn}$ is the identity matrix and $L=KK^\top$ is the graph Laplacian matrix of $G^H$.

Now consider the real vector space $Z= \mathbb{R}^{n_H}\times \mathbb{R}^{|E^H|}$ of all possible values pairs  $(\eta^H,\Delta^H)$ of noise vectors of honest users.
For the sake of readability, in the remainder of this section we will often drop the superscript $H$ and write $(\eta,\Delta)$ when it is clear from the context that we work in the space $Z$.

\newcommand{\Sigmag}{{\Sigma^{(g)}}}
Let \[\Sigmag = \left[\begin{array}{cc}\sigma_\eta^2 I_{U^H} & 0 \\ 0 & \sigma_\Delta^2 I_{E^H}\end{array}\right],\]
and let $\Sigmaginv=\left(\Sigmag\right)^{-1}$,
we then have a joint probability distribution of independent Gaussians:
\[
P((\eta,\Delta)) = \twopiconst\exp\left(-\frac{1}{2}(\eta,\Delta)^\top\Sigmaginv(\eta,\Delta)\right),
\]
where $\twopiconst=(2\pi)^{-(\rndn+|E^H|)/2}|\Sigmag|^{-1/2}$.

Consider the following subspaces of $Z$:
\begin{eqnarray*}
Z^A&=&\{(\eta,\Delta)\in Z \mid \eta+K \Delta={\hat{X}}^H-X^A\}, \\
Z^B&=&\{(\eta,\Delta)\in Z \mid \eta+K \Delta={\hat{X}}^H-X^B\}. 
\end{eqnarray*}

Assume that the (only) vertex for which $X^A$ and $X^B$ differ is $v_1$.
Recall that without loss of generality, private values
are in the interval $[0,1]$. 
Hence, if we set $X^A_{v_1}-X^B_{v_1}=1$ then $X^A$ and $X^B$ are maximally apart and also the difference between $P({\hat{X}}\mid X^A)$ and $P({\hat{X}}|X^B)$ will be maximal.

Now choose any vector $t=(t_\eta,t_\Delta)\in Z$
with $t_\eta = (t_u)_{u\in U^H}$ and $t_\Delta = (t_e)_{e\in E^H}$
such that $t_\eta+Kt_\Delta =
X^A-X^B$. It follows that $Z^B = Z^A + t$, i.e., $Y\in Z^A$ if and only if $Y+t \in Z^B$.
This only imposes one linear constraint on the vector $t$: later
we will choose $t$ more precisely in a way which is most convenient to prove our privacy claims.
  In particular, for any $\hat{X}^H$ we have that $X^A + \eta + K\Delta = \hat{X}^H$ if and only if 
  $X^B + (\eta + t_\eta) +  K(\Delta + t_\Delta) = \hat{X}^H$.
  The key idea is that appropriately choosing $t$ allows us to map any noise $(\eta,\Delta)$ which results in observing ${\hat{X}}^H$ given dataset $X^A$ on a similarly likely noise $(\eta+t_\eta,\Delta+t_\Delta)$ which results in observing ${\hat{X}}^H$ given the adjacent dataset $X^B$.

We illustrate the meaning of $t$ using the example of Figure \ref{fig:dp_t}. 
 Consider the graph $G^H$ of honest users shown in Figure \ref{fig:dp_t.gr} and
databases $X^A$ and $X^B$ as defined above. The difference of neighboring databases $X^A - X^B$  is shown in   Figure \ref{fig:dp_t.a}. 
Figure \ref{fig:dp_t.b} illustrates a possible assignment of  $t$, where  $t_\eta = (\frac{1}{9}, \dots, \frac{1}{9})$ and $t_\Delta = (\frac{5}{9}, \frac{3}{9},  \frac{3}{9}, \frac{1}{9}, \dots, \frac{1}{9}, 0,  0)$.

  One can see that $t=(t_\eta,t_\Delta)$ can be interpreted as a flow on $G^H$
  where $t_\Delta$ represents the values flowing through edges, and $t_\eta$
  represents the extent to which vertices are sources or sinks. The
  requirement $t_\eta+ K t_\Delta = X^A- X^B$ means that for a given user $u$
  we have
  $\sum_{(v,u) \in \dirEH} t_{(v,u)} - \sum_{(u,v)\in\dirEH} t_{(u,v)} = X^A_u
  - X^B_u-t_u$, which can be interpreted as the property of a flow, i.e., the value of incoming edges minus the value of outgoing edges equals the extent to which the vertex is a source or sink (here $-t_u$ except for $v_1$ where it is $1-t_{v_1}$).  We will use this flow $t$ to first distribute $X^A-X^B$ as equally as possible over all users, and secondly to avoid huge flows through edges.  For example, in Figure \ref{fig:dp_t.b}, we choose $t_\eta$ in such a way that the difference $X^A_{v_1} - X^B_{v_1} = 1$ is spread over a difference of $1/9$ at each node, and $t_\Delta$ is chosen to let a value $1/9$ flow from each of the vertices in $U^H \setminus \{v_1\}$ to $v_1$, making the flow consistent.

In the following we will first prove generic privacy guarantees for any (fixed) $t$,  which will be used to map elements of $Z^A$ and $Z^B$ and bound the overall probability differences of outcomes of adjacent datasets. Next, we will
instantiate $t$ for different graphs $G^H$ and obtain concrete guarantees.

\newcount\mycount
\begin{figure}  
\centering
\subfigure[Graph $G^H$]{
\begin{tikzpicture}[inner sep=0pt,  minimum size=3pt]
	
	\node[fill, circle] (v_2) at (360/5 *2 :1.5cm) {};
	\node[fill, circle] (v_3) at (360/5 *3 :1.5cm) {};
	\node[fill, circle] (v_4) at (310 :1.8cm) {};
	\node[fill, circle] (v_5) at (360/5 *5 :1.5cm) {};
	 
	\foreach \phi in {1,...,4}{
		\node[fill, circle] (u_\phi) at (360/4 * \phi:0.7cm) {};
	}; 
	\node [fill, circle, label={[label distance=0.05cm]355:$v_1$}] (v_1) at (360/5:1.5cm) {};
	\draw[-] (v_1) -- (v_2) ;
	\draw[-] (v_2) -- (u_1) ;
	\draw[-] (v_2) -- (v_3) ;
	\draw[-] (v_2) -- (u_2) ;
	\draw[-] (v_3) -- (u_3) ;
	\draw[-] (u_3) -- (v_4) ;
	\draw[-] (u_3) -- (u_4) ;
	\draw[-] (u_1) -- (u_3) ;
	\draw[-] (u_4) -- (v_5) ;
	\draw[-] (v_1) -- (u_4) ;

\end{tikzpicture}
\label{fig:dp_t.gr}
} 
\subfigure[$X^A - X^B$]{ 
\begin{tikzpicture}[inner sep=0pt,  minimum size=3pt]
\node[fill, circle, label={[label distance=0.05cm]360/5 *2:0}] (v_2) at (360/5 *2 :1.5cm) {};
\node[fill, circle, label={[label distance=0.05cm]360/5 * 3:0}] (v_3) at (360/5 *3 :1.5cm) {};
\node[fill, circle, label={[label distance=0.05cm]30:0}] (v_4) at (310 :1.8cm) {};
\node[fill, circle, label={[label distance=0.05cm]30:0}] (v_5) at (360/5 *5 :1.5cm) {};

\node[fill, circle, label={[label distance=0.05cm]360/5:0}] (u_1) at (360/4 :0.7cm) {};
\node[fill, circle, label={[label distance=0.05cm]360/5 * 2:0}] (u_2) at (360/4 * 2:0.7cm) {};
\node[fill, circle, label={[label distance=0.05cm]360/5 * 3:0}] (u_3) at (360/4 * 3:0.7cm) {};
\node[fill, circle, label={[label distance=0.05cm]360/5 * 4:0}] (u_4) at (360/4 * 4:0.7cm) {};
\node [fill, circle,label={[label distance=0.05cm]355:1}] (v_1) at (360/5:1.5cm) {};
\draw[-] (v_1) -- (v_2) [dashed] ;
\draw[-] (v_2) -- (u_1) [dashed] ;
\draw[-] (v_2) -- (v_3) [dashed];
\draw[-] (v_2) -- (u_2) [dashed];
\draw[-] (v_3) -- (u_3) [dashed];
\draw[-] (u_3) -- (v_4) [dashed];
\draw[-] (u_3) -- (u_4) [dashed];
\draw[-] (u_1) -- (u_3) [dashed];
\draw[-] (u_4) -- (v_5) [dashed];
\draw[-] (v_1) -- (u_4) [dashed];
\end{tikzpicture}
\label{fig:dp_t.a} 
}
\subfigure[Labeling $t$] {
\begin{tikzpicture}[inner sep=0pt,  minimum size=3pt]

	\node[fill, circle, label={[label distance=0.01cm]360/5 :$\frac{1}{9}$}] (v_1) at (360/5 :1.5cm) {};
	\node[fill, circle, label={[label distance=0.01cm]360/5 *2:$\frac{1}{9}$}] (v_2) at (360/5 *2 :1.5cm) {};
	\node[fill, circle, label={[label distance=0.01cm]180:$\frac{1}{9}$}] (v_3) at (360/5 *3 :1.5cm) {};
	\node[fill, circle, label={[label distance=0.01cm]30:$\frac{1}{9}$}] (v_4) at (310:1.8cm) {};
	\node[fill, circle, label={[label distance=0.01cm]360/5 * 5:$\frac{1}{9}$}] (v_5) at (360/5 *5 :1.5cm) {};

	\node[fill, circle, label={[label distance=0.02cm]360/5:$\frac{1}{9}$}] (u_1) at (360/4 :0.7cm) {};
	\node[fill, circle, label={[label distance=0.02cm]340:$\frac{1}{9}$}] (u_2) at (360/4 * 2:0.7cm) {};
	\node[fill, circle, label={[label distance=0.02cm]270:$\frac{1}{9}$}] (u_3) at (360/4 * 3:0.7cm) {};
	\node[fill, circle, label={[label distance=0.02cm]360/5 * 4:$\frac{1}{9}$}] (u_4) at (360/4 * 4:0.7cm) {};
	
	\draw[<-] (v_1) -- (v_2) node [midway, above] {$\frac{3}{9}$} ;
	\draw[-] (v_2) -- (u_1) [dashed]  node [midway, below] {0} ;
	\draw[<-] (v_2) -- (v_3) node [midway, left] {$\frac{1}{9}$};
	\draw[<-] (v_2) -- (u_2) node [midway, right] {$\frac{1}{9}$};
	\draw[-] (v_3) -- (u_3) [dashed] node [midway, above] {0};
	\draw[<-] (u_3) -- (v_4) node [midway, above] {$\frac{1}{9}$};
	\draw[->] (u_3) -- (u_4) node [midway, above] {$\frac{3}{9}$};
	\draw[->] (u_1) -- (u_3) node [midway, left] {$\frac{1}{9}$};
	\draw[<-] (u_4) -- (v_5) node [midway, above] {$\frac{1}{9}$};
	\draw[<-] (v_1) -- (u_4) node [midway, right] {$\frac{5}{9}$};
\end{tikzpicture}
\label{fig:dp_t.b} 
}
\caption{An example of valid $t$ over a communication graph of honest users $G^H$. The graph $G^H$ is shown in Figure \ref{fig:dp_t.gr}, the difference of neighboring databases $X^A - X^B$ in Figure \ref{fig:dp_t.a}, and
a possible value of $t$ which evenly distributes the flow of information in Figure \ref{fig:dp_t.b}.}
\label{fig:dp_t} 
\end{figure}

\subsubsection{Abstract Differential Privacy Result}
\label{sec:generic-dp}

We start by proving differential privacy guarantees which depend on the
particular choice of labeling $t$.  Theorem~\ref{thm:diffpriv.etadelta} holds
for all possible choices of $t$, but some choices will lead to more advantageous results than others.  Later, we will apply this theorem for specific choices of $t$ for proving theorems giving privacy guarantees for communication graphs $G^H$ with specific properties.

\newcommand{\thetamax}{{\Theta_{max}}}
We first define the function $\thetamax:\mathbb{R}_+\times (0,1) \mapsto \mathbb{R}_+$ such that $\thetamax$ maps pairs $(\varepsilon,\delta)$ on the largest positive value of $\theta$ satisfying
\begin{align}
  \label{eq:etadeltabound4}
\varepsilon & \ge \theta^{1/2} + \theta/2,
  \\
  \frac{\left(\varepsilon - \theta/2\right)^2}{\theta} & \ge 2\log\left(
  \frac{2}{\delta\sqrt{2\pi}}\right).\label{eq:etadeltabound5}
  \end{align}
Note that for any $\varepsilon$ and $\delta$, any $\theta\in 
(0,\thetamax]$ satisfies Eqs \eqref{eq:etadeltabound4} and \eqref{eq:etadeltabound5}.

\newcommand{\thmDiffprivEtadeltaStatement}{  Let $\varepsilon,\delta \in 
(0,1)$. Choose a vector  $t=(t_\eta,t_\Delta)\in Z$
	with $t_\eta = (t_u)_{u\in U^H}$ and $t_\Delta = (t_e)_{e\in E^H}$
	such that $t_\eta+Kt_\Delta =
	X^A-X^B$ 
  and let $\theta = t^\top \Sigmaginv t$. Under the setting introduced above, if $\theta\le\thetamax(\varepsilon,\delta)$ then \gopa is $(\varepsilon,\delta)$-DP, i.e., \[ P({\hat{X}} \mid X^A) \le e^\varepsilon P({\hat{X}} \mid X^B)+\delta. \] }

\begin{theorem}
  \label{thm:diffpriv.etadelta}
  \thmDiffprivEtadeltaStatement
\end{theorem}

The proof of this theorem is in Appendix \ref{app:proof.thm.diffpriv.etadelta}.
Essentially, we adapt ideas from the privacy proof of the Gaussian mechanism 
\cite{Dwork2014a} to our setting.

\subsubsection{Discussion}
\label{sec:generic-dp.discuss}

Essentially, given some $\varepsilon$, Equation \eqref{eq:etadeltabound4}
provides a
lower bound for the noise (the diagonal of $\Sigmag$) to be added. 
Equation \eqref{eq:etadeltabound4} also implies that the left hand side of
Equation \eqref{eq:etadeltabound5} is larger than $1$.  Equation 
\eqref{eq:etadeltabound5} may then require the noise or $\varepsilon$ to be even higher if $2\log(2/\delta\sqrt{2\pi})\ge 1$,
i.e., $\delta \le 0.48394$.

If $\delta$ is fixed, both \eqref{eq:etadeltabound4} and 
\eqref{eq:etadeltabound5} allow for smaller $\varepsilon$ if
$\theta$ is smaller.  Let us analyze the implications of this result.  We know
that $\theta = \sigma_\eta^{-2} t_\eta^\top t_\eta + \sigma_\Delta^
{-2}t_\Delta^\top t_\Delta$. As we can make $\sigma_\Delta$ arbitrarily large
without affecting the variance of the output of the algorithm (the pairwise
noise terms canceling each other) and thus make the second term
$\sigma_\Delta^{-2}t_\Delta^\top t_\Delta$ arbitrarily small,  our first
priority to achieve a strong privacy guarantee will be to chose a $t$ making
$\sigma_\eta^{-2}t_\eta^\top t_\eta$ small. We have the following
lemma. 

\newcommand{\lemmaSumTuIsOneStatement}[1]{
  In the setting described above, for any $t$ chosen as in Theorem \ref{thm:diffpriv.etadelta}
  we~have: 
\begin{equation}
  \sum_{u\in U^H} t_u = 1. #1
\end{equation}
}

\begin{lemma}
  \label{lm:sum.tu.is.1}
  \lemmaSumTuIsOneStatement{\label{eq:sumtu1}}
\end{lemma}

Therefore, the vector $t_\eta$ satisfying Equation \eqref{eq:sumtu1} and
minimizing $t_\eta^\top t_\eta$ is the vector $\mathbbm{1}_{n_H}/n_H$, i.e.,
the vector containing $n_H$ components with value $1/n_H$.  The proofs of the
several specializations of Theorem \ref{thm:diffpriv.etadelta} we will present
will all be based on this choice for $t_\eta$. The proof of
Lemma~\ref{lm:sum.tu.is.1} can be found in
Appendix~\ref{app:dp.discussion.proofs},
along with the proof of another constraint that we derive from these
observations.
\newcommand{\lemmaCondGHConnectedStatement}{In the setting described above with $t$ as defined in Theorem \ref{thm:diffpriv.etadelta}, if $t_\eta = \mathbbm{1}_{n_H}/n_H$, then $G^H$ must be connected.}
\begin{lemma}
  \label{lm:cond.GH.connected}
  \lemmaCondGHConnectedStatement
\end{lemma}

Given a fixed privacy level and fixed variance of the output, a second
priority is
to
minimize $\sigma_\Delta$, as this may be useful when a user drops out and the
noise he/she exchanged cannot be rolled back or would take too much time to
roll back (see Appendix \ref{app:crypto.drop-out}).
For this, having more edges in $G^H$ implies that the vector $t_\Delta$ has
more components and therefore typically allows a solution to
 $t_\eta+Kt_\Delta=X^A-X^B$ with a smaller $t_\Delta^\top t_\Delta$ and hence a smaller $\sigma_\Delta$.

\subsection{Worst Case Topology}
\label{sec:priv.worst-dp}

We now specialize Theorem~\ref{thm:diffpriv.etadelta} to obtain a worst-case
result.

\newcommand{\thetap}{{\theta_{\hbox{\textsc{p}}}}}

\begin{theorem}[Privacy guarantee for worst-case graph]
\label{thm:diffprivacy}  
  Let $X^A$ and $X^B$ be two databases (i.e., graphs with private values at
the vertices) which differ only in the value of one
user.
Let $\varepsilon,\delta \in (0,1)$ and $\thetap = \frac{1}{\sigma_\eta^{2}\rndn} + \frac{\rndn}{3\sigma_\Delta^{2}}$.
If $G^H$ is connected and $\thetap \le \thetamax(\varepsilon,\delta)$,
then \gopa is $(\varepsilon,\delta)$-differentially private, i.e.,
${P({\hat{X}} \mid X^A)} \le e^\varepsilon P({\hat{X}} \mid X^B)+\delta$.
\end{theorem}
Crucially, Theorem~\ref{thm:diffprivacy} holds as soon as the subgraph $G^H$
of honest users who did not drop out is connected. Note that if $G^H$ is not
connected, we can still obtain a similar but weaker result for each connected
component separately ($\rndn$ is replaced by the size of the connected component).

In order to get a constant $\varepsilon$, inspecting the term $\thetap$ shows
that the variance $\sigma_\eta^{2}$ of the independent noise must be of order $1/\rndn$. This is in a sense optimal as it corresponds to the amount of noise required when averaging $\rndn$ values in the trusted curator model. It also matches the amount of noise needed when using secure aggregation with differential privacy in the presence of colluding users, where honest users need to add $n/\rndn$ more noise to compensate for collusion \cite{Shi2011}.

Further inspection of the conditions in Theorem~\ref{thm:diffprivacy}
also shows
that the variance $\sigma_\Delta^2$ of the pairwise noise must be large
enough. How large it must be actually depends on the structure of the graph
$G^H$. Theorem~\ref{thm:diffprivacy} describes the worst case, which is attained when every
node has as few neighbors as possible while still being connected, i.e., when
$G^H$ is a path. In this case,
Theorem~\ref{thm:diffprivacy} shows that the variance $\sigma_\Delta^2$ needs
to be of
order $\rndn$. Recall that this noise cancels out, so it does not impact
the utility of the final output. It only has a minor effect on the
communication cost (the representation space of reals needs to be large
enough to avoid overflows with high probability), and on the variance of the
final result if some residual noise terms of dropout users are not rolled
back (see Section~\ref{sec:gopa}).

\begin{proof}[of Theorem \ref{thm:diffprivacy}]
  
Let $T$ be a spanning tree of the (connected) communication graph $G^H$. Let
$E^T$ be the set of edges in $T$.
Let $t\in\mathbb{R}^{\rndn+|E^H|}$ be a vector such that:
\begin{itemize}
\item For vertices $u\in U^H$, $t_u=1/\rndn$.
\item For edges $e\in E^H \setminus E^T$, $t_e=0$.
\item Finally, for edges $e\in E^T$, we choose $t_e$ in the unique way such
that $t_\eta+K t_\Delta=(X^A-X^B)$.  
\end{itemize}
In this way, $t_\eta+K t_\Delta$ is a vector with a $1$ on the $v_1$ position and $0$ everywhere else.  We can find a unique vector $t$ using this procedure for any communication graph $G^H$ and spanning tree $T$.  It holds that
\begin{equation}
  \label{eq:tntn.is.invUH}
  t_\eta^\top t_\eta = \left(\frac{\mathbbm{1}_{\rndn}}{\rndn}\right)^\top \left(\frac{\mathbbm{1}_{\rndn}}{\rndn}\right) 
 = \frac{1}{\rndn}.
\end{equation}

Both Equations \eqref{eq:etadeltabound4} and \eqref{eq:etadeltabound5} of
Theorem~\ref{thm:diffpriv.etadelta}, 
require $t^\top \Sigmaginv t$ to be sufficiently small.  We can see $t_\Delta^\top \sigma_\Delta^{-2} t_\Delta$ is maximized (thus producing the worst case) if the spanning tree $T$ is a path $\left (v_1\,v_2\ldots v_{\rndn}\right)$, in which case $t_{\{v_i,v_{i+1}\}}=(\rndn-i)/\rndn$.  Therefore,
\begin{eqnarray}
    t_\Delta^\top t_\Delta & \le &\sum_{i=1}^{\rndn-1} \left(
  \frac{\rndn-i}{\rndn}\right)^2 
  = \frac{\rndn(\rndn-1)(2\rndn-1)/6}{\rndn^2} 
  \label{eq:tdeltatdelta.le.UHdiv2} = \frac{(\rndn-1)(2\rndn-1)}{6\rndn}.
\end{eqnarray}
Combining Equations \eqref{eq:tntn.is.invUH} and \eqref{eq:tdeltatdelta.le.UHdiv2} we get
\[
\theta = t^\top \Sigmaginv t \le \sigma_\eta^{-2}\frac{1}{\rndn} + \sigma_\Delta^{-2}\frac{\rndn(\rndn-1)(2\rndn-1)/6}{\rndn^2} 
\]
We can see that $\theta \le \thetap$ and hence $\theta\le\thetamax(\varepsilon,\delta)$  satisfies the conditions of Theorem \ref{thm:diffpriv.etadelta} and \gopa{} is $(\varepsilon,\delta)$-differentially private.
\qed
\end{proof}
In conclusion, we see that in the worst case $\sigma_\Delta^2$ should be large
(linear in $\rndn$) to keep $\varepsilon$ small, which has no direct negative
effect on the utility of the resulting ${\hat{X}}$. On the other hand,
$\sigma_\eta^2$ can be small (of the order $1/\rndn$), which means that
independent of the number of participants or the way they communicate a small
amount of independent noise is sufficient to achieve DP as long as $G^H$ is
connected.

\subsection{The Complete Graph}
\label{sec:priv.completeGraph}

 The necessary value of $\sigma_\Delta^2$ depends strongly on the
 network structure. This becomes clear in Theorem 
 \ref{thm:diffprivacy-complete}, which covers the case of the complete graph and  shows that for a fully connected $G^H$, $\sigma_\Delta^2$ can be of order $O(1/\rndn)$, which is a \emph{quadratic} reduction compared to the path case.

\begin{theorem}[Privacy guarantee for complete graph]
\label{thm:diffprivacy-complete}  
Let $\varepsilon,\delta \in (0,1)$ and let $G^H$ be the complete graph.
Let $\theta_C = \frac{1}{\sigma_\eta^{2}\rndn} + \frac{1}{\sigma_\Delta^{2}\rndn}$.
    If $\theta_C\le \Theta_{max}(\varepsilon,\delta)$,
    then 
    \gopa is $(\varepsilon,\delta)$-DP.
  \end{theorem}

 \begin{proof} 
 If the communication graph is fully connected, we can use the following values for the vector $t$:
 \begin{itemize}
 \item As earlier, for $v\in U^H$, let $t_v=1/\rndn$.
   \item For edges $\{u,v\}$ with $v_1\not\in\{u,v\}$, let $t_{\{u,v\}}=0$.
\item For $u\in U^H\setminus\{v_1\}$, let $t_{\{u,v_1\}} = 1/\rndn$.
 \end{itemize}
 Again, one can verify that $t_\eta+Kt_\Delta =X^A-X^B $ is a vector with a $1$ on the $v_1$ position and $0$ everywhere else.
In this way, again $t_\eta^\top t_\eta=1/\rndn$ but now
$t_\Delta^\top t_\Delta=(\rndn-1)/\rndn^2$ is much smaller.
We now get
\[\theta = t^\top \Sigmaginv t = \sigma_\eta^{-2}/\rndn + \sigma_\Delta^{-2}(\rndn-1)/\rndn^2 \le \theta_C \le \thetamax(\varepsilon,\delta).\]
Hence, we can apply Theorem \ref{thm:diffpriv.etadelta} and \gopa{} is $
(\varepsilon,\delta)$-differentially private.
\qed
\end{proof}
  
A practical communication graph will be between the two extremes of the path
and the complete graph, as shown in the next section.

\subsection{Practical Random Graphs}
\label{sec:priv.random}
\label{sec:graphs}

Our results so far are not fully satisfactory from the practical perspective,
when the number of users $n$ is large. Theorem~\ref{thm:diffprivacy} assumes
that we have a procedure to generate a graph $G$ such that $G^H$ is guaranteed
to be connected (despite dropouts and malicious users), and requires a large
$\sigma_\Delta^2$ of $O(\rndn)$.  Theorem~\ref{thm:diffprivacy-complete}
applies if we pick $G$ to be the complete graph, which ensures connectivity
of $G^H$ and allows smaller $O(1/\rndn)$ variance but is intractable as all
$n^2$ pairs of users need to exchange noise.

To overcome these limitations, we propose a simple randomized procedure to construct a sparse network graph $G$ such that $G^H$ will be well-connected with high probability, and prove a DP guarantee \emph{for the whole process} (random graph generation followed by \gopa), under much less noise than the worst-case.
The idea is to make each (honest) user select $k$ other users uniformly at
random among all users. Then, the edge $\{u,v\}\in E$ is created if $u$
selected $v$ or $v$ selected $u$ (or both). Such graphs are known as 
\emph{random $k$-out} or \emph{random $k$-orientable} graphs 
\cite{Bollobas2001a,Fenner1982a}. They have very good connectivity properties 
\cite{Fenner1982a,Yagan2013a} and are used in creating secure communication channels in distributed sensor networks \cite{Chan2003a}. Note that \gopa can be conveniently executed while constructing the random $k$-out graph. Recall that $\honestfraction = \rndn/n$ is the proportion of honest users.  We have the following privacy guarantees (which we prove in Appendix~\ref{app:random_graph}).

\begin{theorem}[Privacy guarantee for random $k$-out graphs]
\label{thm:diffprivacy-random}  
  Let $\varepsilon,\delta\in(0,1)$ and let $G$ be obtained by letting all 
  (honest) users randomly choose $k\leq n$ neighbors.
  Let $k$ and $\honestfraction=\rndn/n$ be such that
  $\honestfraction n \ge 81$,
  $\honestfraction k \ge 4\log(2\honestfraction n/3\delta)$, $\honestfraction k
  \ge 6\log(\honestfraction n/3)$ and $\honestfraction k \ge \frac{3}{2} + \frac{9}{4}\log(2e/\delta)$. Let
    \begin{align*} 
  \theta_R &=  \frac{1}{\rndn\sigma_\eta^{2}}+\frac{1}{\sigma_\Delta^{2}}
  \Big(\frac{1}{\lfloor (k-1)\honestfraction/3 \rfloor-1} + \frac{12+6\log
  (\rndn)}{\rndn}
  \Big)
  \end{align*} 
  If $\theta_R \le \thetamax(\varepsilon,\delta)$ then \gopa{} is $(\varepsilon,3\delta)$-differentially private.
\end{theorem}
This result has a
similar form as Theorems~\ref{thm:diffprivacy} and \ref{thm:diffprivacy-complete}
but requires $k$ to be large enough (of order $\log (\honestfraction
n)/\honestfraction$) so that $G^H$ is sufficiently connected despite dropouts and malicious users.
Crucially, $\sigma_\Delta^2$ only needs to be of order
$1/k\honestfraction$
to match the utility of the trusted curator, and each user needs to
exchange with only $2k=O(\log n)$ peers in expectation, which is much
more practical than a complete graph.

Notice that up to a constant factor this result is optimal.  Indeed, in
general, random graphs are not connected if their average degree is smaller
than logarithmic in the number of vertices.  The constant factors mainly serve
for making the result practical and (unlike asymptotic random graph theory)
applicable to moderately small communication graphs, as we illustrate in the
next section.

\subsection{Scaling the Noise}
\label{sec:priv.summary}

Using these results, we can precisely quantify the amount of independent
and pairwise noise needed to achieve a desired privacy guarantee depending on
the topology, as illustrated in the corollary below. 
\begin{corollary}
\label{thm:dp_corollary}
Let $\varepsilon,\delta' \in (0,1)$, and ${\sigma_\eta^2 = c^2
/\rndn\varepsilon^2}$, where $c^2 > 2\log(1.25/\delta')$. Given some $\kappa >
0$, let $\sigma_\Delta^2= \kappa\sigma_\eta^2$ if $G$ is complete,
$\sigma_\Delta^2= \kappa\sigma_\eta^2\rndn(\frac{1}{\lfloor (k-1)\honestfraction/3 \rfloor-1} + 
(12+6\log
  (\rndn))/\rndn)$ if it is a random $k$-out graph with $k$ and $\honestfraction$ as in
  Theorem~\ref{thm:diffprivacy-random}, and $\sigma_\Delta^2 =
  \kappa
  \sigma_\eta^2\rndn^2/3$ for an arbitrary connected $G^H$. Then \gopa{} is
$(\varepsilon, \delta)$-DP with $\delta \geq a 
(\delta'/1.25)^{\kappa/\kappa+1}$, where $a=3.75$ for the $k$-out
graph and $1.25$ otherwise.
\end{corollary}

\begin{table}[t]
\caption{Value of $\sigma_\Delta$ needed to ensure $(\varepsilon,\delta)$-DP
with trusted curator utility for
$n=10000$, $\varepsilon=0.1$,
$\delta'=1/\rndn^2$, $\delta=10\delta'$ depending on the topology, as obtained
from Corollary~\ref{thm:dp_corollary} or numerical simulation. 
} \label{tab:cor}
\vskip 0.1in
\centering
\small
\begin{tabular}{|c|c|c|}
\hline
& $\honestfraction=1$ & $\honestfraction=0.5$\\
\hline
\textbf{Complete} & 1.7 & 2.2\\
\textbf{$k$-out} (Corollary~\ref{thm:dp_corollary})& 32.4 ($k=105$) & 32.5 ($k=203$)\\
\textbf{$k$-out} (simulation) & 33.8 ($k=20$) & 33.4 ($k=40$)\\
\textbf{Worst-case} & 9392.0 & 6112.5\\
\hline
\end{tabular}
\vskip -0.1in
\end{table}

We prove Corollary~\ref{thm:dp_corollary} in Appendix \ref{app:dp-cor}. 
The value of $\sigma_\eta^2$ is set 
such that after all noisy values are aggregated, the variance of the
residual noise
matches that required by the Gaussian mechanism 
\cite{Dwork2014a}
to achieve $(\varepsilon, \delta')$-DP for an average of
$\rndn$ values in the \emph{centralized} setting. The privacy-utility
trade-off achieved by \gopa is thus the same as in the trusted curator model
up to a small constant in $\delta$, as long as the pairwise
variance $\sigma_\Delta^2$ is large enough. As expected, we see that as
$\sigma_\Delta^2\rightarrow +\infty$ (that is, as $\kappa\rightarrow+\infty$),
we have $\delta\rightarrow\delta'$ for worst case and complete graphs, or $\delta\rightarrow3\delta'$ for $k$-out graphs.
Given the desired $\delta\geq \delta'$, we can use
Corollary~\ref{thm:dp_corollary}
to determine a value for $\sigma_\Delta^2$ that is sufficient for \gopa to
achieve $(\varepsilon,\delta)$-DP. Table~\ref{tab:cor} shows a numerical illustration with $\delta$ only a factor $10$ larger than $\delta'$. For random $k$-out graphs, we report the values of $\sigma_\Delta$ and $k$ given by Theorem~\ref{thm:diffprivacy-random}, as well as smaller (yet admissible) values obtained by numerical simulation (see Appendix~\ref{app:dp_numerical_sim}). Although the conditions of
Theorem~\ref{thm:diffprivacy-random} are a bit conservative (constants can
likely be improved), they still lead to
practical values. Clearly, random $k$-out graphs provide a
useful trade-off in terms of scalability and robustness.
Note that in practice, one often does not know in advance the
exact proportion $\honestfraction$ of users who are honest and will not drop out, so a
lower bound can be used instead.

\begin{remark}
\label{rmk:hbc_privacy}
For clarity of presentation, our privacy guarantees protect against
an adversary that consists of colluding malicious users. To 
\emph{simultaneously}
protect against each single honest-but-curious user (who knows his own
independent
noise term), we can simply replace $n_H$ by $n_H'=n_H-1$ in our results. This
introduces a factor $n_H/(n_H-1)$ in the variance, which is negligible for large $n_H$.  Note that the same applies to other approaches which distribute noise-generation over data-providing users, e.g., \cite{Dwork2006ourselves}.
\end{remark}


\section{Correctness Against Malicious Users}
\label{sec:verif}

While the privacy guarantees of Section~\ref{sec:privacy} hold regardless of
the behavior of the (bounded number of) malicious users, the utility guarantees discussed in Section~\ref{sec:gopa} are not valid if malicious users tamper with the protocol.
In this section, we add to our protocol the capability of being audited to
ensure the correctness of the computations while preserving privacy
guarantees.
  In particular, while Section~\ref{sec:privacy} guarantees privacy and prevents inference attacks where attackers infer sensitive information illegally,
  tampering with the protocol can be part of poisoning attacks which aim to change the result of the computation.  We argue that we can detect all attacks to poison the output which can also be detected in the centralized setting.  As we will explain we don't assume prior knowledge about data distributions or patterns, and in such conditions neither in the centralized setting nor in our setting one can detect behavior which could be legal but may be unlikely.

We present here the main ideas.  In \ref{app:crypto} and 
\ref{app:supplementary}, we review some cryptographic background required to
understand and verify the details of our construction.\\

\noindent\emph{Objective.} 
Our goal is to (i) verify that all calculations are performed correctly even though they are encrypted, and (ii) identify any malicious behavior.
As a result, we guarantee that given the input vector $X$ a truthfully computed $\hat{X}^{avg}$ is generated which excludes any faulty contributions.

Concretely, users will be able to prove the following properties:
\begin{flalign}
    \quad\quad &\label{eq:verif1}\hat{X}_u = X_u + \textstyle\sum_{v\in N(u)}
    \Delta_{u,v} + \eta_u,& \forall u \in U,\; \\
    &\label{eq:verif2}\Delta_{u,v} = -\Delta_{v,u},&\forall \{u,v\} \in E,\; \\ 
    &\label{eq:verif3}\eta_u \sim \mathcal{N}(0,
    \sigma_\eta^2),&\forall u \in U,\; \\
    &\label{eq:verif4}X_u \text{ is a valid input,}&\forall u \in U.\;
\end{flalign}
It is easy to see that the correctness of the computation is guaranteed if
Properties
\eqref{eq:verif1}-\eqref{eq:verif4} are satisfied. Note that, as long as they are self-canceling
and not excessively large (avoiding overflows and additional costs if a user
drops out, see 
Appendix \ref{app:crypto.drop-out}), we do not need to ensure that pairwise noise terms $\Delta_{u,v}$ have been
drawn from the prescribed distribution, as these terms do not influence
the final result and only those involving honest users affect the privacy
guarantees of Section~\ref{sec:privacy}. In contrast, Properties 
\eqref{eq:verif3} and \eqref{eq:verif4} are necessary
to prevent a malicious user from biasing the outcome of the computation.
Indeed, \eqref{eq:verif3} ensures that the independent noise is
generated correctly, while \eqref{eq:verif4} ensures that input values are in the allowed domain.
Moreover, we can force users to commit to input data so that they
consistently use the same values for data over multiple computations.

We first explain the involved tools to verify computations in Section \ref{sec:verif.crypto} and 
we present our verification protocol in Section \ref{sec:verif.protocol}. 

\subsection{Tools for verifying computations.}
\label{sec:verif.crypto}

Our approach consists in publishing an encrypted log of the computation using 
\textit{cryptographic commitments} and proving that it is performed correctly without revealing any additional
information using \textit{zero knowledge proofs}.
These techniques are
popular in a number of applications such as privacy-friendly financial systems such as \cite{narula_zkledger_2018,ben_sasson_zerocash_2014}. 
We explain below different tools to robustly verify our computations. Namely,
a structure to post the encrypted log of our computations, hash functions to generate secure random numbers, commitments and zero knowledge proofs.\\ 

\noindent\emph{Public bulletin board.} We implement the publication of commitments
and proofs using a public \emph{bulletin board} so that any party can verify the validity of the protocol, avoiding the need for a trusted verification entity. Users sign their messages so they cannot deny them.
More general purpose distributed ledger technology 
\cite{Nakamoto2008a} 
could be used here, but we aim at an application-specific, light-weight and
hence more scalable solution.\\

\noindent\emph{Representations.}  We will represent numbers by elements of cyclic
groups isomorphic to $\Z_q$ for some large prime $q$.  To be able to work with
signed fixed-precision values, we encode them in $\Z_q$ by multiplying them by a constant $1/\gdprec$ and using the upper half of $\Z_q$, i.e.,  $\{ x \in \Z_q : x \ge \lceil q/2 \rceil  \}$ to represent negative values.  Unless we explicitly state otherwise, properties (such as linear relationships) we establish for the $\Z_q$ elements translate straightforwardly to the fixed-precision values they represent.  We choose the precision so that the error of approximating real numbers up to a multiple of $\gdprec$ does not impact our results. \\

\noindent\emph{Cryptographic hash functions.} We also use hash functions ${H: \Z \rightarrow  \Z_{2^T}}$ for an integer $T$ such that $2^T$ is a few orders of magnitude bigger than $q$, so that numbers uniformly distributed over $\Z_{2^T}$ modulo  $q$  are indistinguishable from numbers uniformly distributed over $\Z_q$.  
Such function is easy to evaluate, but predicting its outcome or distinguishing it from random numbers is intractable for polynomially bounded algorithms \cite{yao_theory_1982}. Practical instances of $H$ can be found in \cite{barker2007recommendation,bertoni2010sponge}. 
\\

\noindent\emph{Pedersen commitments.}
Commitments, first introduced by \cite{blum_coin_1983}, allow users to commit to values while keeping them hidden from others. After a commitment is performed, the committer cannot change its value, but can later prove properties of it or reveal it.
For our protocol we use the Pedersen commitment scheme \cite{pedersen1991non}. Pedersen commitments have as public parameters $\comschpp = (q, \G, g, h)$ where $\G$ is a cyclic multiplicative group of prime order $q$, and $g$ and $h$ are two generators of $\G$ chosen at random. A commitment is computed by applying the function $Com_\comschpp : \Z_q \times \Z_q \rightarrow \mathbb{G}$, defined as
\begin{eqnarray} 
	Com_{\comschpp} (x, r) &=& g^x \cdot h^r, 
\end{eqnarray}
where $(\cdot)$ is the product operation of $\G$, $x \in \Z_q$ is the committed value, and $r \in \Z_q$ is a random number to ensure that $Com_{\comschpp} (x, r)$ \emph{perfectly hides} $x$.  When $r$ is not relevant, we simply denote by $Com_\comschpp(x)$ a commitment of $x$ and assume $r$ is drawn appropriately.  Under the Discrete Logarithm Assumption (DLA) $Com_{\comschpp}$ is \emph{binding} as long as $g$ and $h$ are picked at random such that no one knows the discrete logarithm base $g$ of $h$. That is, no computationally efficient algorithm can find $x_1,x_2, r_1, r_2 \in \Z_q$ such that $x_1 \neq x_2$ and $Com_{\comschpp}(x_1, r_1) = Com_{\comschpp}(x_2, r_2)$. As the parameters $\comschpp$ are public, many parties can share the same scheme, but parameters must be sampled with unbiased randomness to ensure the binding property. Pedersen commitments are \emph{homomorphic}, as $Com_{\comschpp}(x+y,r+s)=Com_{\comschpp}(x,r)\cdot Com_{\comschpp}(y,s)$.  This is already enough to verify linear relations, such as the ones in Properties (\ref{eq:verif1}) and (\ref{eq:verif2}). We describe commitments in more detail in Appendix \ref{app:supp.cmt}.

It is sometimes needed to let users prove that they know the values underlying
commitments.  In our discussion we will implicitly assume proofs of knowledge \cite{cramer_zero-knowledge_1998} (see also below) are inserted where needed.\\

\noindent\emph{Zero Knowledge Proofs.}
\label{app:crypto_framework:zkps_bblocks}
To verify other than linear relationships,
we use a family of techniques called Zero Knowledge Proofs (ZKPs), first proposed in \cite{goldwasser_knowledge_1989}. In these proofs, a party called the
\emph{prover} convinces another party, the \emph{verifier}, about a statement over committed values. For our
scope and informally speaking, ZKPs\footnote{Strictly speaking, the proofs we will use are called \textit{arguments}, as the soundness property relies on the computational boundedness of the Prover $P$ through the DLA described above, but as for general reference to the family of techniques we use the term \textit{proofs}.}
\begin{itemize}
	\item  allow the prover to successfully prove true a statement (completeness), 
	\item allow the verifier to discover with arbitrarily large probability any attempt to prove a false statement 
	(soundness), 
	\item guarantee that by performing the proof, no information about the knowledge of the prover other than the proven statement is revealed (zero knowledge). 
\end{itemize}
Importantly, the zero knowledge property of our proofs does not rely on any
computational hardness assumption.\\

\newcommand{\cst}{C_{cst}}
\newcommand{\zproof}[1]{\boldsymbol{\pi}_{#1}}
\noindent\emph{$\Sigma$-protocols.} We use a family of ZKPs called \emph{$\Sigma$-protocols} and introduced in \cite{cramer_modular_1997}. They allow to prove the knowledge of committed values, relations between them that involve arithmetic circuits in $\Z_q$ \cite{cramer_zero-knowledge_1998} (i.e. functions that only contain additions and products in $\Z_q$), and the disjunction of statements involving this kind of relations \cite{cramer_proofs_1994}. Let $\cst$ be the cost of computing an arithmetic circuit $C: \Z_q^m \rightarrow \Z_q^s$, then the computational cost of proving the correct computation of $C$ requires $O(\cst)$ cryptographic computations (mainly exponentiations in $\G$) and a transfer of $O(\cst)$ elements of $\G$. 
Proving the disjunction $S_1 \lor S_2$ of two statements $S_1$ and $S_2$ costs the sum of proving $S_1$ and $S_2$ separately. For simplicity, we say that a proof has cost $c$ if the cost of generating a proof and its verification is at most 
$c$ cryptographic computations each, and the 
communication cost is at most 
$c$ elements of $\G$.
We denote by $\groupelsize$  the size in bits of an element of $\G$. 

Proving that a commitment to a number $a \in \Z_q$ is in a certain range $[0, 2^k-1]$ for some integer $k$ can be derived from circuit proofs with the following folklore protocol: commit to each bit of $b_1, \dots , b_k$ of $a$ and prove that they are indeed bits, for example by proving that $b_i (1 - b_i ) = 0$ for all $i \in \{1, \dots, k\}$, then prove that $a = \sum_{i=1}^k 2^{i-1} b_i$. This proof has a cost of $5k$. The homomorphic property of commitments allows one to easily generalize this proof to any range $[a,b] \subset \Z_p$  with a cost of $10\lceil \log_2(b-a) \rceil$. 

$\Sigma$-protocols require that the prover interacts with a honest verifier.
This is not applicable to our setting where verifiers can be malicious. We can
turn our proofs into non-interactive ones with negligible additional cost with
the strong Fiat-Shamir heuristic \cite{bernhard_how_2012}. In that way,
for a statement $S$ each user generates a proof transcript $\zproof{S}$ together with the involved commitments and publish it in the bulletin board. Any party can later verify offline the correctness of $\zproof{S}$. The transcript size  is equal to the amount of exchanged elements in the original protocol. We provide further details about ZKPs and $\Sigma$-protocols in Appendix \ref{app:supp.zkp}.

\newcommand{\comD}{\mathbf{c_D}}
\newcommand{\pseed}{r}
\subsection{Verification Protocol}
\label{sec:verif.protocol}

Our verification protocol, based on the primitives described in Section \ref{sec:verif.crypto}, consists of four phases: 
\begin{enumerate}
\item \label{enum:verif.phase.privcommit}
  \textit{Private data commit.} At the start of our protocol, we assume users
  have committed to their private data.  In particular, for every user $u$ a
  commitment is available, either directly published by $u$ or available
  through a distributed ledger or other suitable mechanism.  This
  attenuates data poisoning, as it forces users to use the same value for
  $X_u$
  in each computation where it is needed.
\item \label{enum:verif.phase.setup}
\textit{Setup.} In a setup phase at the start of our protocol, users generate Pedersen commitment parameters $\comschpp$ and private random seeds that will be used to prove Property \eqref{eq:verif3}. Details are discussed in Appendix \ref{app:crypto.setup}.
  \item \label{enum:verif.phase.verif}
\textit{Verification.} During our protocol, users can prove that execution is
performed correctly and verify logs containing such proofs by others.  If
during the protocol one detects a user has cheated he/she is added to a
cheater list.  After the protocol, one can verify that all steps were
performed correctly and that the protocol has been completed. We give
details on this key step below.
\item \label{enum:verif.phase.mitigate}
\textit{Mitigation.}  Cheaters and drop-out users (who got off-line for a too
long period of time) detected during the protocol are excluded from the
computation, and their contributions are rolled back.  Details are provided in
Appendix \ref{app:crypto.drop-out}.
\end{enumerate}

\noindent\emph{Verification phase.}
First, we use the homomorphic property of Pedersen commitments to prove Properties \eqref{eq:verif1} and \eqref{eq:verif2}. 
Note that Property \eqref{eq:verif2} involves secrets of two different users $u$ and $v$. This is not a problem as these pairwise noise terms are known by both involved users, so they can use negated randomnesses $r_{\Delta_{u,v}}=-r_{\Delta_{v,u}}$ in their commitments of $\Delta_{u,v}$ and $\Delta_{v,u}$ such that everybody can verify that $Com_{\comschpp}(\Delta_{u,v},r_{\Delta_{u,v}}) \cdot Com_{\comschpp} (\Delta_{v,u},r_{\Delta_{v,u}})=Com_{\comschpp}(0,0)$.
Users can choose how they generate pairwise Gaussian noise (e.g., by convention, the user that initiates the exchange can
generate the noise). We just require that each user holds a message on the
agreed noise terms signed by the other user before publishing commitments, so
that if one of them cheats, it can be easily discovered. 

Verifying the correct drawing of Gaussian numbers is more involved and requires private seeds $\pseed_1, \dots, \pseed_n$ generated in Phase \ref{enum:verif.phase.setup}. We explain the procedure step by step in Appendix \ref{app:supp.gauss}. The proof generates a transcript $\zproof{\eta_u}$ for each user $u$. 

To verify Property \eqref{eq:verif4}, we verify its domain and its consistency.
For the domain, we prove that $X_u \in [0,1]$ with the range proof outlined in Section \ref{sec:verif.crypto}. %
For the consistency, users $u$ publish a Pedersen commitment $\mathbf{c_{X_u}}=Com_{\comschpp}(X_u)$ and prove its consistency with private data committed to in Phase \ref{enum:verif.phase.privcommit} denoted as $\comD$.  Such proof depends on the nature of the commitment in Phase \ref{enum:verif.phase.privcommit}: if the same Pedersen commitment scheme is used nothing needs to be done, but users could also prove consistency with a record in a blockchain (which is also used for other applications) or they may need to prove more complex consistency relationships.  We denote the entire proof transcript as $\zproof{X_u}$.
As an illustration, consider ridge regression in Example~\ref{ex:linearreg}. Every user $u$ can publish commitments  $\mathbf{c_{y_u}} = Com_
{\comschpp}(y_u)$, $\mathbf{c_{\phi_u^i}} = Com_{\comschpp}(\phi_u^i)$ for $i \in \{1, \dots ,d\}$ (computed with the appropriately drawn
randomness), and additionally commit to  $\phi_u^i y_u$ and $\phi_u^i \phi_u^j$, for $i,j \in \{1, \dots, d\}$. Then it can be verified that all these commitments are
computed coherently, i.e, that the commitment of $\phi_u^i y_u$
is the product of secrets committed in $\mathbf{c_{y_u}}$ and $\mathbf{c_{\phi_u^i}}$ for $i \in
\{1, \dots, d\}$, and analogously for the commitment of $\phi_u^i
\phi_u^j$ in relation with $\mathbf{c_{\phi_u^i}}$ and $\mathbf{c_{\phi_u^j}}$, for $i,j \in \{1, \dots,  d \}$.
 
We note that if poisoned private values are used consistently after committing to them, 
this will remain undetected. However, if our verification methodology is applied in the training of many different models 
over time, it could be required that users prove consistency over values that
have been committed long time ago. 
Therefore, cheating is discouraged and these attacks are attenuated by the
impossibility to adapt corrupted 
contributions to specific computations.

Compared to the central setting with a trusted curator,
encrypting the input does not make the verification of input more problematic. 
Both in the central setting and in our setting one can perform domain tests,
ask certification of values from independent sources, and require consistency
of the inputs over multiple computations, even though in some cases both the
central curator and our setting may be unable to verify the correctness of some input.

Algorithm~\ref{alg:verif} gives a high level overview of the 4 verification steps described above. 
By composition of ZKPs, these steps allow
each user to prove the correctness of their computations and preserve completeness, soundness and zero knowledge properties, thereby leading to our security guarantees: 

\begin{theorem}[Security guarantees of \gopa]
	Under the DLA, a user $u \in U$ that passes 
	the verification protocol proves that $\hat{X}_u$ was computed correctly. Additionally, 
	$u$ does not reveal any additional information about $X_u$ by running the verification, even if the DLA does not hold.
	\label{thm:gopa_security}
\end{theorem}

To reduce the verification load, we note that it is possible to perform the
verification for only a subset of users picked at random (for
example, sampled using public unbiased randomness generated in Phase \ref{enum:verif.phase.setup})
after users have published the involved commitments.
In this case, we obtain \emph{probabilistic} security guarantees, which may be
sufficient for some applications.

We can conclude that  \gopa{} is an auditable protocol that, through
existing efficient cryptographic primitives, can offer guarantees similar to
the automated auditing which is possible for data shared with a central party.

\begin{algorithm*}[t]
	\floatname{algorithm}{Algorithm}
	\caption{Verification of \gopa}
	\begin{algorithmic}[1]
		\STATE{\emph{(1) Input. } Import any previous commitments to private data $\comD$}
		\STATE{\emph{(2) Setup. } All users jointly run Phase \ref{enum:verif.phase.setup} Setup to generate Pedersen parameters $\comschpp$ and private seeds $\pseed_1, \dots, \pseed_n$.  Each user $u$ publishes $\mathbf{c_{\pseed_u}} = Com_{\comschpp}(\pseed_u)$}
		\STATE{\emph{(3a) Verification - commits.}}
		\STATE{\hspace*{.5cm}\textbf{for all} user $u \in \userset$, \textbf{publish} as soon as available:}
		\STATE{\qquad\quad$-$ $\mathbf{c_{X_u}}= Com_{\comschpp}(X_u)$ and proof $\zproof{X_u}$ that $X_u$ is valid 
                }
		\STATE{\qquad\quad$-$ $\mathbf{c_{\eta_u}} = Com_{\comschpp}(\eta_u)$ and proof $\zproof{\eta_u}$ that $\eta_u$ is drawn from Gaussian distribution} 
		\STATE{\qquad\quad$-$ $\mathbf{c_{\Delta_{u,v}}} = Com_{\comschpp}(\Delta_{u,v})$
                }
		\STATE{\qquad\quad$-$ $\hat{X}_u$ and randomness to compute its commitment 
                }
		\STATE{\emph{(3b) Verification - checks.}} 
		\STATE{\hspace*{.5cm}\textbf{for all} $u \in \userset$ \textbf{verify} when commitments/proofs are available: }
		\STATE{\qquad\quad$-$ $(\zproof{X_u}, \mathbf{c_{X_u}}, \comD)$ is correct, } 
		\STATE{\qquad\quad$-$ $\mathbf{c_{X_u}} \cdot \left( \prod_{v \in N(u)} \mathbf{c_{\Delta_{u,v}}} \right) \cdot \mathbf{c_{\eta_u}} = Com_\comschpp(\hat{X}_u)$}\label{algl:sum}, 
		\STATE{\qquad\quad$-$  $(\zproof{\eta_u},\mathbf{c_{\pseed_u}}, \mathbf{c_{\eta_u}})$ is correct.}
		\STATE{\hspace*{.5cm}If a check is incorrect, add $u$ to cheaters list.} 
		\STATE{\qquad\quad\textbf{for all} user $v \in N(u)$ \textbf{do}:} 
			\STATE{\qquad\qquad $-$ \textbf{if} $\mathbf{c_{\Delta_{u,v}}} \cdot \mathbf{c_{\Delta_{v,u}}} \neq Com_{\comschpp}(0,0)$\textbf{:} add $u$ and/or $v$ as cheater} \label{algl:zero_sum} 
		\STATE{\emph{(4) Mitigation.}} 
		\STATE{\hspace*{.5cm}$-$Roll back contributions of drop-outs and exchange more noise if necessary} 
		\STATE{\hspace*{.5cm}$-$If a harmless amount of non-canceled pairwise noise remains,\\
                  \quad\quad\quad declare the computation successful, otherwise abort.}
	\end{algorithmic}
	\label{alg:verif}
\end{algorithm*}


\section{Computation and Communication Costs}
\label{sec:complexity}

Our cost analysis considers user-centered
costs, which is natural as most operations can be performed asynchronously
and in parallel.
The following statement summarizes our results (concrete non-asymptotic
costs are in \ref{app:complexity}).
\begin{theorem}[Complexity of \gopa]
  \label{thm:cryptocost}
  Let $\gdprec>0$ be the desired fixed precision such that the number $1$ would be represented as  $1/\gdprec$. Let $B > 0$ be such that the $\eta_u$'s are drawn
  from a Gaussian distribution approximated with $1/B$ equiprobable bins. 
  Then, each user $u$, to perform and prove its contribution, requires 
  $O(|N(u)|+\log(1/\gdprec)\log(1/B) + \log(1/B) + \log(1/\gdprec))$ computations and transferred bits. 
  The verification of its contribution requires the same cost.
\end{theorem}
  
Unlike other frameworks like fully homomorphic encryption and secure
multiparty computation, our cryptographic primitives \cite{franck_efficient_2017} scale well to large data.


\section{Experiments}
\label{sec:exp}

\noindent\emph{Private averaging.} We present some numerical simulations to study
the empirical utility of \gopa and in particular the influence of malicious
and dropped out users. We consider $n=10000$,
$\varepsilon=0.1$,
$\delta=1/n^2$ and set the values of $k$, $\sigma^2_\eta$ and $\sigma^2_\Delta$
using Corollary~\ref{thm:dp_corollary} so that \gopa satisfies $
(\varepsilon,\delta)$-DP.
Figure~\ref{fig:exp} (left) shows the utility of 
\gopa when executed with $k$-out graphs as a function of $\rho$ , which is the (lower bound on) the proportion of
users who are honest and do not drop out.
The curves in the figure are closed form formulas given by Corollary~\ref{thm:dp_corollary} (for \gopa{}) and Appendix A of \cite{Dwork2014a} (for local DP and central DP). 
As long as the value of $k$ is admissible, it does not change $\sigma_\eta$. The utility of \gopa is shown for  different values of $\kappa$. This parameter allows to obtain different trade-offs between magnitudes of $\sigma_\eta$ and $\sigma_\Delta$. While a very small  $\kappa$  degrades the utility, this impact quickly becomes negligible as $\kappa$ reaches $10$ (which also has a minor effect in $\sigma_\Delta$).  With $\kappa = 10$ and even for reasonably small $\rho$, \gopa{} already approaches 
a utility of the same order as a trusted curator that would average
the values of \emph{all} $n$ users. Further increasing $\kappa$ would not be of any use as
this will not have a significant impact in utility and will simply increase
$\sigma_\Delta$.

While the values of $\varepsilon$ and
$\delta$ obviously impact the utility, we stress the fact that they only have
a uniform scaling effect
which does not affect the relative distance between the utility of \gopa and that
of a trusted curator. Regarding the communication graph $G^H$, it greatly
influences the
communication cost and $\sigma_\Delta$, but only affects $\sigma_\eta$ via
parameter $a$ of Corollary~\ref{thm:dp_corollary} which has a negligible impact in utility.

In Appendix~\ref{app:crypto.drop-out}, we further describe the
ability of \gopa to tolerate a small number of residual
pairwise noise terms of variance $\sigma_\Delta^2$ in the final result.
We note that this feature
is rather unique to \gopa and is not possible with secure
aggregation \cite{Bonawitz2017a,bell_secure_2020}.\\

\noindent\emph{Application to federated SGD.}
We present some experiments on training a logistic regression
model in a federated learning setting.
We use a binarized version of UCI Housing dataset
with standardized features and points
normalized to unit L2 norm to ensure a gradient sensitivity bounded by $2$. We
set aside $20\%$ of the points as test set
and split the rest uniformly at random across $n=10000$ users
so that each user $u$ has a local dataset $D_u$ composed of 1 or 2 points.

We use the Federated SGD algorithm, which corresponds to FedAvg with a single
local update \cite{mcmahan2016communication}. At each iteration, each user
computes a stochastic gradient using one sample of their local dataset;
these gradients are then averaged and used to update the model parameters.
To privately average the gradients, we compare \gopa (using $k$-out graphs with $\rho=0.5$ and $\kappa = 10$) to (i)
a
\emph{trusted} aggregator that averages all $n$ gradients in the clear and adds
Gaussian noise to the result as per central DP, and (ii) local DP.
We fix the total privacy budget to $\varepsilon=1$ and
$\delta=1/(\rho n)^2$ and use advanced composition (in Section 3.5.2 of \cite{Dwork2014a}) to compute the budget
allocated to each iteration. Specifically,  we use Corollary
3.21 of
\cite{Dwork2014a} by  requiring that each  update is  $(\varepsilon_s,
\delta_s)$-DP for $\varepsilon_s = \varepsilon / 2\sqrt{2 T \ln(1/\delta_s)}$,
$\delta_s = \delta/T+1$ and $T$ equal to the total number of iterations. This
ensures $(\varepsilon, \delta)$-DP for the overall algorithm.
The step size is tuned for each approach,
selecting the value with the highest accuracy after a predefined number
$T$ of iterations. 

Figure~\ref{fig:exp} (middle) shows a typical run of the algorithm for
$T=50$ iterations. Local DP is not shown as it diverges unless the learning
rate
is overly small. On the other hand, \gopa is able to decrease the objective
function steadily, although we see some difference with the trusted
aggregator (this is
expected since $\rho=0.5$).
Figure~\ref{fig:exp} (right) shows the final test accuracy (averaged over
10 runs) for different numbers of iterations $T$. Despite
the small gap in objective function, \gopa
nearly matches the accuracy achieved by the trusted aggregator, while
local DP is unable to learn useful models.

We provide the code to reproduce the experimental results
presented in Figures \ref{fig:exp} and \ref{fig:nonroll_dropout} (see Appendix
\ref{app:crypto.drop-out}) and in Tables \ref{tab:cor} (see Section 
\ref{sec:priv.summary}) and \ref{tab:cor2} (see Appendix 
\ref{app:dp_numerical_sim}) in a public repository.\footnote{
\url{https://gitlab.inria.fr/cesabate/mlj2022-gopa}} 

\begin{figure*}[t] 
	\subfigure[Utility with respect to $\rho$]{
		\centering
		\includegraphics[width=0.31\linewidth]{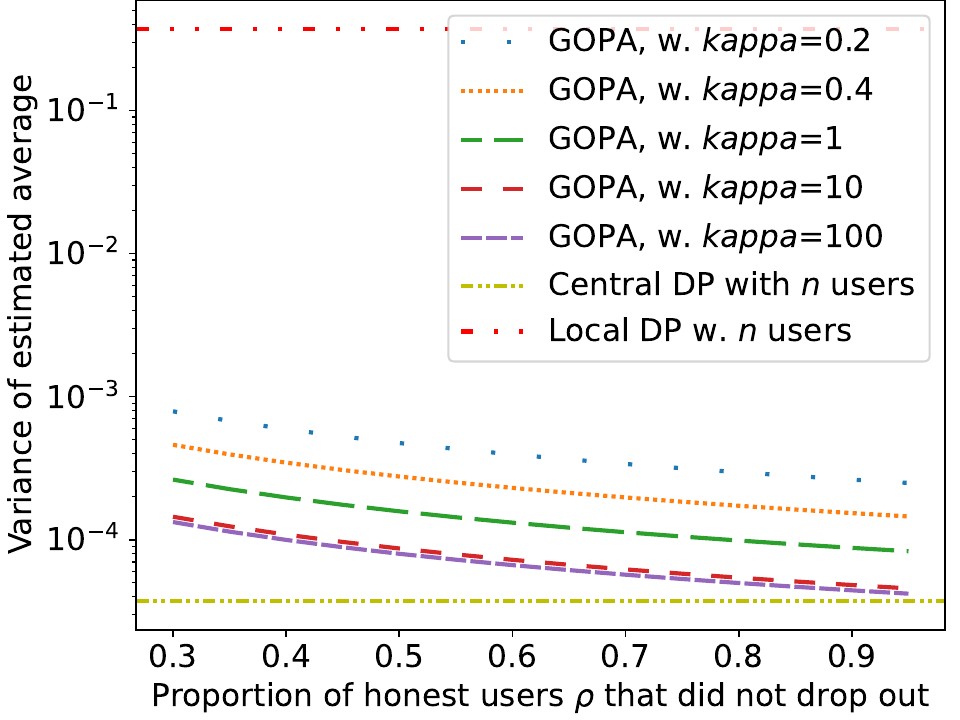}
		\label{fig:utility_rho}
	}
	\subfigure[A typical run of FedSGD]{
		\centering
		\includegraphics[width=.31\linewidth]{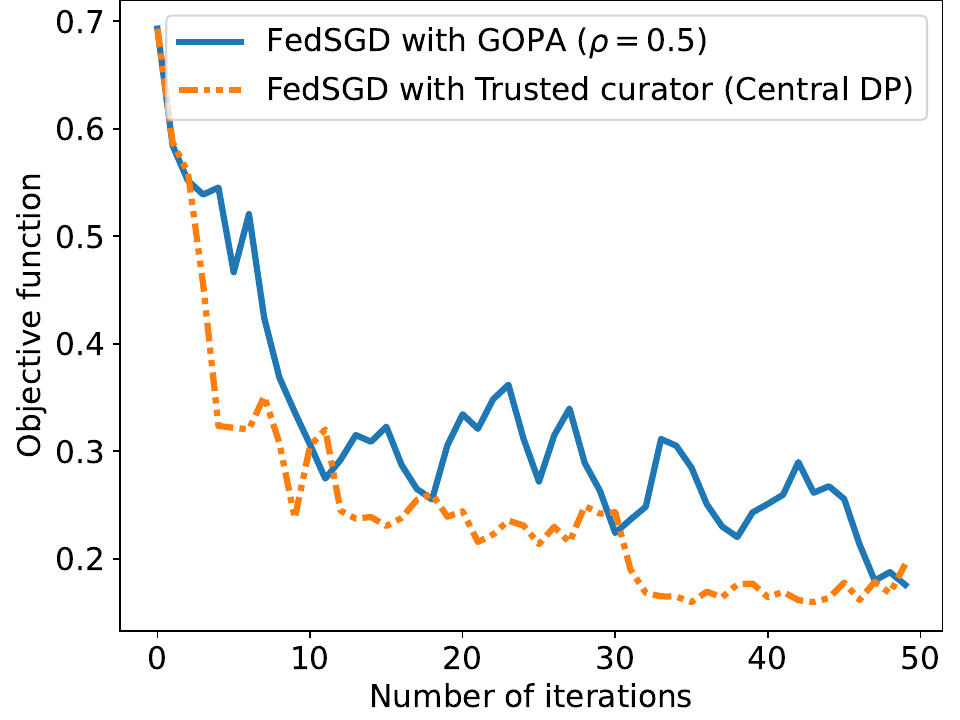}
		\label{fig:sgd_run}
	}
	\subfigure[Test accuracy of FedSGD]{
		\centering
		\includegraphics[width=.31\linewidth]{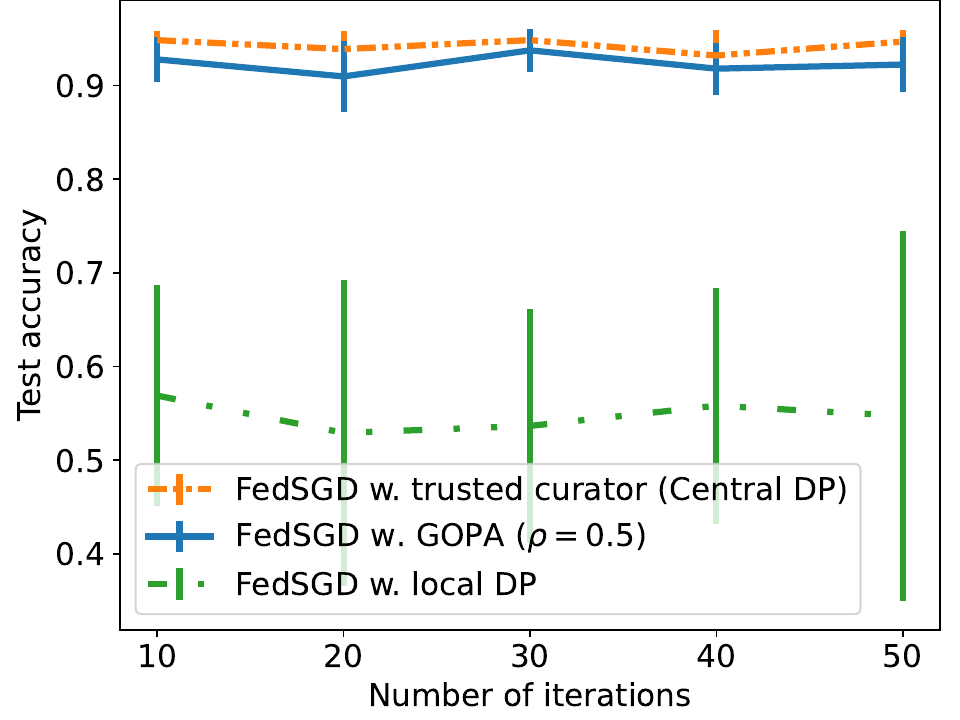}
		\label{fig:sgd_acc}
	}
	\caption{Comparing \gopa to central and local DP. \emph{Left}: Utility of
		$\gopa$ (measured by the variance of the
		estimated average) w.r.t. $\rho$.
		\emph{Middle:} Evolution of the
		objective for a typical run of FedSGD. \emph{Right:} Test accuracy
		of
		models learned with FedSGD. See text for details.}
	\label{fig:exp}
\end{figure*}


\section{Conclusion}
\label{sec:conclu}

We proposed \gopa, a protocol to privately compute averages over the
values of many users. \gopa satisfies DP, can nearly
match the utility of a trusted curator, and is robust to
malicious parties. It can be used in distributed and federated ML 
\cite{Jayaraman2018,kairouz2019advances} in place of more
costly
secure aggregation schemes.
In future work, we plan to provide efficient implementations,
to integrate our approach in complete ML systems, and to exploit scaling
to reduce the cost per average.
We think that our work is also relevant
beyond averaging, e.g. in the context of
robust aggregation for distributed SGD \cite{Blanchard2017} and for
computing pairwise statistics \cite{Bell2020a}.

%
%



%
%

\section*{Declarations}
This work was partially supported by ANR project ANR-20-CE23-0013 'PMR',
ANR-20-CE23-0015 'PRIDE', and the 'Chair TIP' project funded by ANR, I-SITE,
MEL, ULille and INRIA.
We thank Pierre Dellenbach and Alexandre Huat for the fruitful discussions.
There are no conflicts of interest.
No ethical approval was needed.
As there were no participants no consent was needed to participate nor to publish.
Code and data can be accessed by following the links in the text.
The authors made approximately equal contributions.

\bibliographystyle{spmpsci}      
\bibliography{main}   


\newpage
\onecolumn
\appendix
\renewcommand{\thesection}{Appendix~\Alph{section}}
\renewcommand{\thesubsection}{\Alph{section}.\arabic{subsection}}


\section{Proofs of Differential Privacy Guarantees}
\label{app:dp-proof}

In this appendix, we provide derivations for our differential privacy
guarantees.

\subsection{Proof of Theorem \ref{thm:diffpriv.etadelta}}
\label{app:proof.thm.diffpriv.etadelta}

\textbf{Theorem \ref{thm:diffpriv.etadelta}.}
\thmDiffprivEtadeltaStatement

\begin{proof}
  We adapt ideas from \cite{Dwork2014a} to our setting.
  First, we show that it is sufficient to prove that

\begin{equation}
  P((\eta,\Delta)) \le P((\eta,\Delta)+t) e^\varepsilon + \delta.
  \label{eq:propboundeps}
\end{equation}
In particular,  if Eq \eqref{eq:propboundeps} holds we have that 
\begin{align*}  
P(\hat{X} \mid X^A ) &= \int_{(\eta, \Delta) \in Z^A} P((\eta,\Delta)) \text{d}\eta \text{d} \Delta \\
&\leq \int_{(\eta, \Delta) \in Z^A} \left(P((\eta,\Delta) +t) e^\varepsilon + \delta \right) \text{d} \eta  \text{d} \Delta  \\
&= \int_{(\eta, \Delta) -t \in Z^A } \left(P((\eta,\Delta)) e^\varepsilon + \delta \right) \text{d} \eta  \text{d} \Delta \\
&= \int_{(\eta, \Delta) \in Z^B } \left(P((\eta,\Delta)) e^\varepsilon + \delta \right) \text{d} \eta  \text{d} \Delta \\
&= P(\hat{X} \mid X^B ) e^\varepsilon + \delta,
\end{align*}
which proves the required bound.
Hence, it is sufficient to prove Eq \eqref{eq:propboundeps},
  or else to prove that $P((\eta,\Delta))\le e^\varepsilon P((\eta,\Delta)+t)$ with probability at least $1-\delta$.
\newcommand{\etadelta}{\gamma}
Denoting $\etadelta=(\eta,\Delta)$ for convenience, we need to prove that with
probability $1-\delta$ it holds that $|\log(P(\etadelta)/P(\etadelta+t))| \le
e^\varepsilon$.
We have
\begin{eqnarray*}
\Big|\log \frac{P(\etadelta)}{P(\etadelta+t)}\Big|
&=& \Big|-\frac{1}{2}\etadelta^\top \Sigmaginv\etadelta + \frac{1}{2}
(\etadelta+t)^\top \Sigmaginv(\etadelta+t)\Big|\\
&=& \Big|\frac{1}{2}(2\etadelta+t)^\top\Sigmaginv t\Big|.
\end{eqnarray*}
To ensure that $|\log(P(\etadelta)/P(\etadelta+t))| \le
e^\varepsilon$ holds with probability at least $1-\delta$, since we are interested in the absolute value, we will show that
\[
P\Big(\frac{1}{2}(2\etadelta+t)^\top\Sigmaginv t \ge \varepsilon\Big) \le
\delta/2,
\]
the proof of the other direction is analogous.
This is equivalent to
\begin{equation}
  \label{eq:etadeltabound1}
P(\etadelta\Sigmaginv t \ge \varepsilon - t^\top\Sigmaginv t/2) \le \delta/2.
\end{equation}
The variance of $\etadelta\Sigmaginv t$ is
\begin{eqnarray*}
  \text{var}(\etadelta\Sigmaginv t) &=&
  \sum_v \text{var}\left(\eta_v \sigma_\eta^{-2} t_v\right) + \sum_e \text{var}\left(\Delta_e \sigma_\Delta^{-2} t_e \right) \\
  &=& \sum_v \text{var}\left(\eta_v\right) \sigma_\eta^{-4} t_v^2 + \sum_e 
  \text{var}\left(\Delta_e\right) \sigma_\Delta^{-4} t_e^2 \\
  &=& \sum_v \sigma_\eta^2 \sigma_\eta^{-4} t_v^2 + \sum_e \sigma_\Delta^2 \sigma_\Delta^{-4} t_e^2  \\
  &=& \sum_v \sigma_\eta^{-2} t_v^2 + \sum_e \sigma_\Delta^{-2} t_e^2 \\
  &=& t^\top \Sigmaginv t.
  \end{eqnarray*}
For any centered Gaussian random variable $Y$ with variance
$\sigma_Y^2$, we have that
\begin{equation}
  P(Y\ge \lambda) \le \frac{\sigma_Y}{\lambda\sqrt{2\pi}}\exp\left(-\lambda^2/2\sigma_Y^2\right).
  \label{eq:guassTailBound}
\end{equation}
Let $Y=\etadelta\Sigmaginv t$, $\sigma_Y^2=t^\top \Sigmaginv t$ and $\lambda=\varepsilon - t^\top\Sigmaginv t/2$, then satisfying
\begin{equation}
  \frac{\sigma_Y}{\lambda\sqrt{2\pi}}\exp\left(-\lambda^2/2\sigma_Y^2\right)\le \delta/2
  \label{eq:gaussTailLeDelta}
\end{equation}
implies \eqref{eq:etadeltabound1}.
Equation \eqref{eq:gaussTailLeDelta} is equivalent to
\[
\frac{\lambda}{\sigma_Y}\exp\left(\lambda^2/2\sigma_Y^2\right)\ge
2/\delta\sqrt{2\pi},
\]
or, after taking logarithms on both sides, to
\[
\log\left(\frac{\lambda}{\sigma_Y}\right)+\frac{1}{2}\left(\frac{\lambda}{\sigma_Y}\right)^2\ge
\log\left(\frac{2}{\delta\sqrt{2\pi}}\right).
\]
To make this inequality hold, we require
\begin{equation}
  \log\left(\frac{\lambda}{\sigma_Y}\right)\ge 0
  \label{eq:etadeltabound2}
\end{equation}
and 
\begin{equation}
  \frac{1}{2}\left(\frac{\lambda}{\sigma_Y}\right)^2\ge
\log\left(\frac{2}{\delta\sqrt{2\pi}}\right).
  \label{eq:etadeltabound3}
\end{equation}
Equation \eqref{eq:etadeltabound2} is equivalent to
$
\lambda\ge \sigma_Y.
$
Substituting $\lambda$ and $\sigma_Y$ we get
\[
\varepsilon - t^\top\Sigmaginv t/2 \ge 
(t^\top \Sigmaginv t)^{1/2},
\]
which is equivalent to \eqref{eq:etadeltabound4}.
Substituting $\lambda$ and $\sigma_Y$ in Equation \eqref{eq:etadeltabound3}
gives \eqref{eq:etadeltabound5}.  Hence, if Equations \eqref{eq:etadeltabound4}
 and \eqref{eq:etadeltabound5} are satisfied the desired differential privacy follows.\qed
\end{proof}

\subsection{Proofs for Section \ref{sec:generic-dp.discuss}}
\label{app:dp.discussion.proofs}

\textbf{Lemma \ref{lm:sum.tu.is.1}.}
\lemmaSumTuIsOneStatement{\tag{\ref{eq:sumtu1}}\nonumber}

\begin{proof}
  Due to the properties of the incidence matrix $K$, i.e., $\forall u,v: K_{u,\{u,v\}}=-K_{v,\{u,v\}}$, the sum of the components of the vector $K\Delta$ is zero, i.e.,
\[
\sum_{u\in U^H} \left(K\Delta\right)_u = 
\sum_{u\in U^H} \left(\sum_{\{u,v\}\in E^H} K_{u,\{u,v\}}\Delta_{\{u,v\}}\right)_u = 0
\]
Combining this with the fact that $t_\eta+Kt_\Delta=X^A-X^B$ with $\sum_{u\in U^H} (X^A-X^B)_u=1$ we obtain Equation \eqref{eq:sumtu1}.
\qed
\end{proof}

\textbf{Lemma \ref{lm:cond.GH.connected}.}
\lemmaCondGHConnectedStatement

\begin{proof}
  Suppose $G^H$ is not connected, then there is a connected component $C\subseteq U^H\setminus \{v_1\}$.
  Let $t_C=\left(t_u\right)_{u\in C}$ and $\Delta_C=(\Delta_e)_{e\in\dirEH\cap(C\times C)}$. Let $K_C = \left(K_{u,e}\right)_{u\in C,e\in \dirEH\cap (C\times C)}$ be the incidence matrix of $G^H[C]$, the subgraph of $G^H$ induced by $C$.  Due to the properties of the incidence matrix of a graph there would hold $\sum_{u\in C} (K_C \Delta_C)_u = 0$.  As there would be no edges between vertices in $C$ and vertices outside $C$, we would have $\sum_{u\in C} (K \Delta)_u = 0$. There would follow $\sum_{u\in C}t_u=\sum_{u\in C}(X^A-X^B-K\Delta)_u = 0$ which would contradict with $t_\eta=\mathbbm{1}_{n_H}/n_H$.  In conclusion, $G^H$ must be connected.
\qed  \end{proof}

\subsection{Random $k$-out Graphs}
\label{app:random_graph}

In this section, we will study the differential privacy properties for the
case where all users select $k$ neighbors randomly, leading to a proof of
Theorem \ref{thm:diffprivacy-random}.  We will start by analyzing the
properties of $G^H$ (Section \ref{sec:dp.rnd.gh}). Section 
\ref{sec:dp.rnd.setup} consists of preparations for embedding a suitable
spanning tree in $G^H$.  Next, in Section \ref{sec:dp.rnd.embed} we will prove
a number of lemmas showing that such suitable spanning tree can be embedded
almost surely in $G^H$.  Finally, we will apply these results to proving
differential privacy guarantees for \gopa{} when communicating over such a random
$k$-out graph $G$ in Section \ref{sec:dp.rnd.gopa}, proving Theorem \ref{thm:diffprivacy-random}.

In this section, all newly introduced notations and definitions are local and
will not be used elsewhere. At the same time, to follow more closely existing
conventions in random graph theory, we may reuse in this section some variable names used elsewhere and give them a different meaning.

\subsubsection{The Random Graph $G^H$}
\label{sec:dp.rnd.gh}

Recall that the communication graph $G^H$ is generated as follows:
\begin{itemize}
\item We start with $n=|U|$ vertices where $U$ is the set of agents. 
\item All (honest) agents randomly select $k$ neighbors to obtain a $k$-out
graph $G$.
\item We consider the subgraph $G^H$ induced by the set $U^H$ of honest users
who did not drop out.
Recall that $\rndn=|U^H|$ and that a fraction $\honestfraction$ of the users
is honest and did not drop out, hence $\rndn=\honestfraction n$. \end{itemize}

Let $\rndk=\honestfraction k$. 
The graph $G^H$ is a subsample of a $k$-out-graph, which for larger $\rndn$
and $\rndk$ follows a distribution very close to that of Erd\H{o}s-R\'enyi
random graphs $G_p(\rndn,2\rndk/\rndn)$.  To simplify our argument, in the
sequel we will assume $G^H$ is such random graph as this does not affect the
obtained result.
In fact, the random $k$-out model concentrates the degree of
vertices more narrowly around the expected value than Erd\H{o}s-R\'enyi
random graphs, so any tail bound our proofs will rely on that holds for Erd\H{o}s-R\'enyi random graphs also holds for the graph $G^H$ we are considering.
In particular, for $v\in U^H$, the degree of $v$ is a random variable which we will approximate for sufficiently large $\rndn$ and $\rndk$ by a binomial $B(\rndn,2\rndk/\rndn)$ with expected value $2\rndk$ and variance $2\rndk(1-2\rndk/\rndn)\approx 2\rndk$.

\subsubsection{The Shape of the Spanning Tree}
\label{sec:dp.rnd.setup}

Remember that our general strategy to prove differential privacy results is to
find a spanning tree in $G^H$ and then to compute the norm of the vector
$t_\Delta$ that will ``spread'' the difference between $X^A$ and $X^B$ over
all vertices (so as to get a $\sigma_\eta$ of the same order as in the trusted
curator setting).  Here, we will first define the shape of a rooted tree and
then prove that with high probability this tree is isomorphic to a spanning
tree of $G^H$.  Of course, we make a crude approximation here, as in the (unlikely) case that our predefined tree cannot be embedded in $G^H$ it is still possible that other trees could be embedded in $G^H$ and would yield similarly good differentially privacy guarantees.  While our bound on the risk that our privacy guarantee does not hold will not be tight, we will focus on proving our result for reasonably-sized $U$ and $k$, and on obtaining interesting bounds on the norm of $t_\Delta$.

Let $\rndG=([\rndn],\rndEG)$ be a random graph where between every pair of vertices there is an edge with probability $2\rndk/\rndn$.
The average degree of $\rndG$ is $2\rndk$.

Let $\rndk\ge 4$. Let $\rndq\ge 3$ be an integer.
Let $\rndD_1$, $\rndD_2$ \ldots $\rndD_\rndq$ be a sequence of positive integers such that
\begin{equation}
  \label{eq:rnd.Dvsq}
\left(\sum_{i=1}^{\rndq} \prod_{j=1}^i \rndD_j\right) - (\rndD_\rndq+1) \prod_{j=1}^{\rndq-2} \rndD_j
< \rndn \le \sum_{i=1}^{\rndq} \prod_{j=1}^i \rndD_j.
\end{equation}

Let $\rndT=([\rndn],\rndET)$ be a balanced rooted tree with $\rndn$ vertices, constructed as follows.
First, we define for each level $l$ a variable $\rndlev_l$ representing the number of vertices at that level, and a variable $\rndclev_l$ representing the total number of vertices in that and previous levels.  In particular:
at the root $\rndclev_{-1}=0$, $\rndclev_0=\rndlev_0=1$ and for $l\in
[\rndq-2]$ by induction $\rndlev_l = \rndlev_{l-1}\rndD_l$ and
$\rndclev_l=\rndclev_{l-1}+\rndlev_l$.
Then, $\rndlev_{\rndq-1} = \lceil(\rndn-\rndclev_{\rndq-2})/(\rndD_\rndq+1)\rceil$,
$\rndclev_{\rndq-1}=\rndclev_{\rndq-2}+\rndlev_{\rndq-1}$, $\rndlev_\rndq=\rndn-\rndclev_{q-1}$ and $\rndclev_q = \rndn$.
Next, we define the set of edges of $\rndT$:
\[
\rndET=\{\{\rndclev_{l-2}+i,\rndclev_{l-1}+\rndlev_{l-1}j+i\} \mid l\in[\rndq]
\wedge i\in[\rndlev_{l-1}] \wedge \rndlev_{l-1}j+i\in[\rndlev_l] \}.
\]
So the tree consists of three parts: in the first $\rndq-2$ levels, every vertex has a fixed, level-dependent number of children, the last level is organized such that a maximum of parents has $\rndD_\rndq$ children, and in level $\rndq-1$ parent nodes have in general $\rndD_{\rndq-1}-1$ or $\rndD_{\rndq-1}$ children.
Moreover, for $0\le l\le \rndq-2$, the difference between the number of vertices in the subtrees rooted by two vertices in level $l$ is at most $\rndD_q+2$.
We also define the set of children of a vertex, i.e., for $l\in[\rndq]$ and $i\in[\rndlev_{l-1}]$, 
\[
ch(\rndclev_{l-2}+i) = \{\rndclev_{l-1}+\rndlev_{l-1}j+i \mid \rndlev_{l-1}j+i\in[\rndlev_l] \}.
\]

In Section \ref{sec:dp.rnd.embed}, we will show conditions on $\rndn$, $\rndD_1 \ldots \rndD_q$, $\rndk$ and $\rnddelta$ such that for a random graph $\rndG$ on $\rndn$ vertices and a vertex $\rndroot$ of $\rndG$, with high probability (at least $1-\rnddelta$) $\rndG$ contains a subgraph isomorphic to $\rndT$ whose root is at $\rndroot$.

\subsubsection{Random Graphs Almost Surely Embed a Balanced Spanning Tree}
\label{sec:dp.rnd.embed}

The results below are inspired by \cite{Krivelevich2010}.
We specialize this result to our specific problem, obtaining proofs which are also valid for graphs smaller than $10^{10}$ vertices, even if the bounds get slightly weaker when we drop terms of order $O(\log(\log(\rndn)))$ for the simplicity of our derivation.

Let $\rndF$ be the subgraph of $\rndT$ induced by all its non-leafs, i.e.,
$\rndF=([\rndclev_{\rndq-1}],\rndEF)$  with $\rndEF=\{\{i,j\}\in\rndET\mid i,j\le\rndclev_{\rndq-1} \}$.

\begin{lemma}
  \label{lm:rnd.22}
  Let $\rndG$ and $\rndF$ be defined as above.  Let $\rndroot$ be a vertex of $\rndG$.
  Let $\rndn\ge \rndk \ge \rndD_i\ge 3$ for $i\in [l]$. 
  Let $\rndgamma = \max_{l=1}^{q-1} \rndD_l/\rndk$  and let $\rndgamma +4(\rndD_\rndq+2)^{-1}+2\rndn^{-1} \le 1$. 
  Let $\rndk\ge 4\log(2\rndn/\rnddeltaa(\rndD_\rndq+2))$.  
  Then, with probability at least $1-\rnddeltaa$, there is an isomorphism $\rndphi$ from $\rndF$ to a subgraph of $\rndG$, mapping the root $1$ of $\rndF$ on $\rndroot$.
\end{lemma}

\begin{proof}
  We will construct $\rndphi$ by selecting images for the children of vertices of $\rndF$ in increasing order, i.e., we first select $\rndphi(1)=\rndroot$, then map children $\{2 \ldots \rndD_1+1\}$ of $1$ to vertices adjacent to $\rndroot$, then map all children of $2$ to vertices adjacent to $\phi(2)$, etc.
  Suppose we are processing level $l\in[\rndq-1]$ and have selected $\rndphi(j)$ for all $j\in ch(i')$ for all $i'<i$ for some $\rndclev_{l-1}<i\le\rndclev_l$.  We now need to select images for the $\rndD_l$ children $j\in ch(i)$ of vertex $i$ (or in case $l=q-1$ possibly only $\rndD_l-1$ children).  This means we need to find $\rndD_l$ neighbors of $i$ not already assigned as image to another vertex (i.e., not belonging to $\cup_{0\le i'<i} \rndphi(ch(i'))$).  We compute the probability that this fails.  For any vertex $j\in[\rndn]$ with $i\neq j$, the probability that there is an edge between $i$ and $j$ in $\rndG$ is $2\rndk/\rndn$.  
  Therefore, the probability that we fail to find $\rndD_l$ free neighbors of $i$ can be upper bounded as
  \begin{eqnarray}
    Pr\left[\rndfailF(i)\right]=
    Pr\left[\text{Bin}\left(\rndn-\rndclev_l, \frac{2\rndk}{\rndn}\right)<\rndD_l\right]
    \le\exp\left(\frac{-\left((\rndn-\rndclev_l)\frac{2\rndk}{\rndn}-\rndD_l\right)^2}{2(\rndn-\rndclev_l)\frac{2\rndk}{\rndn}}\right).\label{eq:rnd.22.1}
  \end{eqnarray}
  We know that $\rndn-\rndclev_l \ge \rndlev_\rndq$.
  Moreover, $(\rndlev_\rndq + \rndD_\rndq-1)/\rndD_\rndq \ge \rndlev_{\rndq-1}$
  and $\rndclev_{\rndq-2} + 1 \le \rndlev_{\rndq-1}$,
  hence $2(\rndlev_\rndq + \rndD_\rndq -1)/\rndD_\rndq  \ge \rndclev_{\rndq-2}+\rndlev_{\rndq-1} +1 = \rndclev_{\rndq-1}+1$  and $2(\rndlev_\rndq + \rndD_\rndq -1)/\rndD_\rndq +\rndlev_\rndq \ge \rndn+1$.
  There follows \[
  \rndlev_\rndq(2+\rndD_\rndq)
  \ge \rndn+1 - 2(\rndD_\rndq -1)/\rndD_\rndq \ge \rndn -1.
  \]
  Therefore,
  \begin{equation}
    \label{eq:rnd.22.levq}
  \rndn-\rndclev_l \ge \rndlev_\rndq \ge \rndn(1-2(\rndD_\rndq+2)^{-1}-\rndn^{-1}).
  \end{equation}
  Substituting this and $\rndD_l \ge \gamma \rndk$ in Equation 
  \eqref{eq:rnd.22.1},
  we
  get
  \begin{eqnarray*}
    Pr\left[\rndfailF(i)\right]&\le& \exp\left(\frac{-\left(\rndn(1-2(\rndD_\rndq+2)^{-1}-\rndn^{-1})\frac{2\rndk}{\rndn}-\rndk\rndgamma\right)^2}{2\rndn(1-2(\rndD_\rndq+2)^{-1}-\rndn^{-1})\frac{2\rndk}{\rndn}}\right)\\
    &\le& \exp\left(\frac{-\rndk^2\left(2(1-2(\rndD_\rndq+2)^{-1}-\rndn^{-1})-\rndgamma\right)^2}{4\rndk}\right)\\
    &\le& \exp\left(\frac{-\rndk^2}{4\rndk}\right)
    = \exp\left(\frac{-\rndk}{4}\right),
  \end{eqnarray*}
  where the latter inequality holds as $\gamma + 4(\rndD_\rndq+2)^{-1}+2\rndn^{-1}\le 1$.
  As $\rndk\ge 4\log(2\rndn/\rnddeltaa(\rndD_\rndq+2))$ 
  we can conclude that   
  \[
  Pr\left[\rndfailF(i)\right] \le \frac{\rnddeltaa(\rndD_\rndq+2)}{2\rndn}. \]
  The total probability of failure to embed $\rndF$ in $\rndG$ is therefore
  given by
  \begin{eqnarray*}
    \sum_{i=2}^{\rndclev_{\rndq-1}} \rndfailF(i) &\le& 
    (\rndclev_{\rndq-1}-1) \frac{\rnddeltaa(\rndD_\rndq+2)}{2\rndn} \\
    &\le& \left(\rndn(2(\rndD_\rndq+2)^{-1}+\rndn^{-1})-1\right) 
    \frac{\rnddeltaa(\rndD_\rndq+2)}{2\rndn}
    = \frac{2\rndn}{\rndD_\rndq+2} \frac{\rnddeltaa(\rndD_\rndq+2)}{2\rndn}
    =\rnddeltaa,
  \end{eqnarray*}
  where we again applied (\ref{eq:rnd.22.levq}).\qed
\end{proof}

Now that we can embed $\rndF$ in $\rndG$, we still need to embed the leaves of
$\rndT$.  Before doing so, we review a result on matchings in random graphs. 
The next lemma mostly follows \cite{Bollobas2001a} (Theorem 7.11 therein),
we mainly adapt to our notation, introduce a confidence parameter and make a
few less crude approximations\footnote{In particular, even though Bollobas's
proof is asymptotically tight, its last line uses the fact that $(e\log n)^
{3a} n^{1-a+a^2/n}=o(1)$ for all $a\le n/2$.  This expression is only lower
than $1$ for $n\ge 5.6\cdot10^{10}$, and as the sum of this expression over
all possible values of $a$ needs to be smaller than $\rnddeltaa$, we do not
expect this proof applies to graphs representing current real-life datasets.}.

\begin{lemma}
  \label{lm:rndB}
  Let $\rndm\ge 27$ (in our construction, $\rndm=\rndlev_\rndq$)   
  and $\rndzeta \ge 4$.
  Consider a random bipartite graph with vertex partitions $A=\{a_1 \ldots a_\rndm\}$ and $B=\{b_1 \ldots b_\rndm\}$, where for every $i,j\in[\rndm]$ the probability of an edge $\{a_i,b_j\}$ is $\rndp=\rndzeta(\log(\rndm)\rndlll{-\log \log \log(\rndm)}{})/\rndm$.  Then, the probability of a complete matching between $A$ and $B$ is higher than \[
1-\frac{e\rndm^{-2(\rndzeta-1)/3}}{1-\rndm^{-(\rndzeta-1)/3}}.
  \]
\end{lemma}

\begin{proof}
  For a set $X$ of vertices, let $\Gamma(X)$ be the set of all vertices adjacent to at least one member of $X$.
  Then, if there is no complete matching between $A$ and $B$, there is some set $X$, with either $X\subset A$ or $X\subset B$, which violates Hall's condition, i.e., $|\Gamma(X)|<|X|$.  Let $X$ be the smallest set satisfying this property (so the subgraph induced by $X\cup \Gamma(X)$ is connected).
  The probability that such sets $X$ and $\Gamma(X)$ of respective sizes $i$ and $j=|\Gamma(X)|$ exist is upper bounded by
 appropriately combining:
  \begin{itemize}
  \item the number of choices for $X$, i.e., $2$ (for selecting $A$ or $B$) times $\comb{\rndm}{i}$,
  \item the number of choices for $\Gamma(X)$, i.e., $\comb{\rndm}{i-1}$ (considering that $j\le i-1$),
  \item an upper bound for the probability that under these choices of $X$ and $\Gamma(X)$ there are at least $2i-2$ edges (as the subgraph induced by $X\cup\Gamma(X)$ is connected), i.e., $\comb{ij}{i+j-1}$ possible choices of the vertex pairs and $\rndp^{i+j-1}$ the probability that these vertex pairs all form edges, and
  \item the probability that there is no edge between any of $X\cup\Gamma(X)$ and the other vertices, i.e., $(1-\rndp)^{i(\rndm-j)+j(\rndm-i)}=(1-\rndp)^{\rndm(i+j)-2ij} $.
  \end{itemize}
  
  Thus, we upper bound the probability of observing such sets $X$ and $\Gamma(X)$ of sizes $i$ and $j$ as follows:
  \begin{eqnarray}
    \rndfailB(i,j) &\le& 
      \comb{\rndm}{i} \comb{\rndm}{j} \comb{ij}{i+j-1} \rndp^{i+j-1} (1-\rndp)^{\rndm(i+j)-2ij} \nonumber\\
      &\le& \left(\frac{\rndm e}{i}\right)^i \left(\frac{\rndm e}{j}\right)^{j}   \left(\frac{ij e}{i+j-1}\right)^{i+j-1} \rndp^{i+j-1} (1-\rndp)^{\rndm(i+j)-2ij}. \nonumber
  \end{eqnarray}
  Here, in the second line the classic upper bound for combinations is used: $\comb{\rndm}{i}<\left(\frac{\rndm e}{i}\right)^i$.
  As $2j\le i+j-1$, we get 
  \begin{eqnarray}
    \rndfailB(i,j) 
    &\le& \left(\frac{\rndm e}{i}\right)^i \left(\frac{\rndm e}{j}\right)^{j}   \left(\frac{i e}{2}\right)^{i+j-1} \rndp^{i+j-1} (1-\rndp)^{\rndm(i+j)-2ij} \nonumber\\
  &\le&  \frac{\rndm^{i+j}e^{2i+2j-1}i^{j-1}}{j^j 2^{i+j-1}} \rndp^{i+j-1} (1-\rndp)^{\rndm(i+j)-2ij}.  \label{eq:rnd.73.1}
  \end{eqnarray}
  As $0<\rndp<1$, there also holds
  \[
  (1-\rndp)^{1/\rndp} < 1/e,
  \]
  and therefore
  \[
  (1-\rndp)^{\rndm(i+j)-2ij} = (1-\rndp)^{\frac{1}{\rndp}\rndp (\rndm(i+j)-2ij)}
  < (1/e)^{\rndp (\rndm(i+j)-2ij)}.
  \]
  We can substitute $\rndp=\rndzeta(\log(\rndm)\rndlll{-\log \log \log(\rndm)}{})/\rndm$ to obtain
  \begin{eqnarray*}
    (1-\rndp)^{\rndm(i+j)-2ij} &<& (1/e)^{(\rndzeta(\log(\rndm)\rndlll{-\log \log \log(\rndm)}{})/\rndm) (\rndm(i+j)-2ij)}
      = \left(\frac{\rndlll{\log\log(\rndm)}{1}}{\rndm}\right)^{\rndzeta (\rndm(i+j)-2ij)/\rndm}.
  \end{eqnarray*}
  Substituting this into Equation
  \eqref{eq:rnd.73.1}, we get
  \begin{eqnarray}
    \rndfailB(i,j) &\le&
    \frac{\rndm^{i+j}e^{2i+2j-1}i^{j-1}}{j^j 2^{i+j-1}}
    \left(\frac{\log(\rndm)\rndzeta}{\rndm}\right)^{i+j-1}
    \left(\frac{\rndlll{\log\log(\rndm)}{1}}{\rndm}\right)^{\rndzeta (\rndm(i+j)-2ij)/ \rndm} \nonumber\\
    &=&
  \frac{\rndm e i^{j-1}}{j^j}
    \left(\frac{\log(\rndm)\rndzeta e^2}{2}\right)^{i+j-1}
    \left(\frac{\rndlll{\log\log(\rndm)}{1}}{\rndm}\right)^{\rndzeta ((i+j)-2ij/\rndm)}.  \nonumber
  \end{eqnarray}
  Given that $\rndm^{\rndzeta/3}\ge \rndzeta\log(\rndm)\rndlll{\log\log(\rndm)}{} e^2/2$ holds for $\rndm\ge 27$ 
  and $\rndzeta \ge 4$,  
  we get 
    \begin{eqnarray}
    \rndfailB(i,j)
    &\le&
  \frac{\rndm e i^{j-1}}{j^j}
    \left(\frac{\log(\rndm)\rndlll{\log\log(\rndm)}{}\rndzeta e^2}{2\rndm^{\rndzeta/3}}\right)^{i+j-1}
    \rndm^{-\rndzeta ((i+j)-2ij/\rndm)+\rndzeta(i+j-1)/3} \nonumber\\
    &\le&
  \frac{e i^{j-1}}{j^j}
  \rndm^{-\rndzeta (\frac{2}{3}(i+j)+1/3-2ij/\rndm)+1} \nonumber\\
    &\le&
  \frac{e i^{j-1}}{j^j}
  \rndm^{-\rndzeta (\frac{1}{3}i+\frac{1}{3})+1}. \nonumber
    \end{eqnarray}
    As $\rndzeta \ge 4$, 
    this implies    
    \begin{eqnarray}
    \rndfailB(i,j)
  &\le&
  \frac{e}{i}\left(\frac{i}{j}\right)^j
  \rndm^{-\frac{\rndzeta i}{3}-\frac{1}{3}}. \label{eq:rnd.73.3}
  \end{eqnarray}
There holds:
    \begin{eqnarray*}
      \sum_{j=1}^{i-1} \left(\frac{i}{j}\right)^j
      &=& \sum_{j=1}^{\lfloor i/3\rfloor} \left(\frac{i}{j}\right)^j
      + \sum_{j=\lfloor i/3\rfloor+1}^i \left(\frac{i}{j}\right)^j \\
      &\le & \sum_{j=1}^{\lfloor i/3\rfloor} \left(\frac{i}{j}\right)^j
      + \sum_{j=\lfloor i/3\rfloor+1}^i 3^j 
      = \sum_{j=1}^{\lfloor i/3\rfloor} \left(\frac{i}{j}\right)^j
      + 3^{\lfloor i/3 \rfloor+1}
      \frac{3^{i-\lfloor i/3\rfloor}-1}{3-1} \\
     &<& \sum_{j=1}^{\lfloor i/3\rfloor} \left(\frac{i}{j}\right)^j
      + \frac{3^{i+1}}{2} 
     < \sum_{j=1}^{\lfloor i/3\rfloor} i^{i/3}
      + \frac{3^{i+1}}{2} \\
      & \le & \frac{i}{3} i^{i/3}
      + \frac{3^{i+1}}{2}. \\
      \end{eqnarray*}
        Substituting in Equation (\ref{eq:rnd.73.3}) gives
    \begin{eqnarray*}      
      \rndfailB&=&\sum_{i=2}^{m/2} \sum_{j=1}^{i-1}\rndfailB(i,j)\\
      &<& \sum_{i=2}^{m/2} \sum_{j=1}^{i-1} \frac{e}{i}\left(\frac{i}{j}\right)^j
      \rndm^{-\frac{\rndzeta i}{3}-\frac{1}{3}} 
      < \sum_{i=2}^{m/2} \left(\frac{i^{i/3+1}}{3}
      + \frac{3^{i+1}}{2}\right)\frac{e}{i}
      \rndm^{-\frac{\rndzeta i}{3}-\frac{1}{3}} \\
      &<& \sum_{i=2}^{m/2} \left(i^{i/3}
      + 3^{i+1}\right) \frac{e}{2}
      \rndm^{-\frac{\rndzeta i}{3}-\frac{1}{3}}
      < \sum_{i=2}^{m/2} i^{i/3}\frac{e}{2} \rndm^{-\frac{\rndzeta i}{3}-\frac{1}{3}} 
      + 3^{i+1} \frac{e}{2} \rndm^{-\frac{\rndzeta i}{3}-\frac{1}{3}}.
    \end{eqnarray*}
    As $\rndm\ge 27 = 3^3$ we can now write 
    \begin{eqnarray*}
       \rndfailB &<& \frac{e}{2}\sum_{i=2}^{m/2}  \rndm^{-\frac{(\rndzeta-1) i}{3}-\frac{1}{3}}
      + \rndm^{(i+1)/3} 
      \rndm^{-\frac{\rndzeta i}{3}-\frac{1}{3}}
      < \frac{e}{2}\sum_{i=2}^{m/2}  \rndm^{-\frac{(\rndzeta-1) i}{3}-\frac{1}
      {3}}
      +\rndm^{-\frac{(\rndzeta-1) i}{3}}\\
      &<& e\sum_{i=2}^{m/2} \rndm^{-\frac{(\rndzeta-1) i}{3}}
      = e \rndm^{-\frac{2(\rndzeta-1) }{3}} 
      \sum_{i=0}^{m/2-2} \left(\rndm^{-\frac{\rndzeta-1}{3}}\right)^i 
      \\ &=& e \rndm^{-\frac{2(\rndzeta-1) }{3}} 
      \frac{1-\left(\rndm^{-\frac{\rndzeta-1}{3}}\right)^{m/2-1}}{1-\left(\rndm^{-\frac{\rndzeta-1}{3}}\right)} 
                    < e \rndm^{-\frac{2(\rndzeta-1) }{3}} 
      \frac{1}{1-\left(\rndm^{-\frac{\rndzeta-1}{3}}\right)}.
    \end{eqnarray*}
    This concludes the proof.\qed
\end{proof}
    
\begin{lemma}
  \label{lm:Bollo.73.delta}
  Let $\rndm\ge 27$    
  and $\rnddeltab>0$.
  Let \[
  \rndzeta = \max\left(4, 1+\frac{3\log(2e/\rnddeltab)}{2\log(\rndm)}\right).
  \] 
  Consider a random bipartite graph as described in Lemma \ref{lm:rndB} above.
Then, with probability at least $1-\rnddeltab$ there is a complete matching between $A$ and $B$.
\end{lemma}
\begin{proof}
  From the given $\rndzeta$, we can infer that
  \begin{eqnarray*}
    &\rndzeta-1 \ge \frac{3\log(2e/\rnddeltab)}{2\log(\rndm)}&\\
    &\rndzeta-1 \ge \frac{-3\log\left(\rnddeltab/2e\right)}{2\log(\rndm)}&\\
    &-\frac{2}{3}(\rndzeta-1)\log(\rndm) \le \log\left(\rnddeltab/2e\right) & \\
    &\rndm^{-2(\rndzeta-1)/3} \le \rnddeltab/2e & 
  \end{eqnarray*}
  We also know that $\rndzeta\ge 4$ and $\rndm\ge 27$, hence
  $(\rndzeta-1)/3\ge 1$ and $1-\rndm^{-(\rndzeta-1)/3}\ge 26/27 \ge 1/2$.
  We know from Lemma \ref{lm:rndB} that the probability of having a complete matching is at least
  \[
  1-\frac{e\rndm^{-2(\rndzeta-1)/3}}{1-\rndm^{-(\rndzeta-1)/3}}\ge 1-\frac{e
  (\rnddeltab/2e)}{1/2} = 1-\rnddeltab.
\]
\end{proof}

  \begin{lemma}
    Let $\rnddeltab>0$, $\rndD\ge 1$ (in our construction, $\rndD=\rndD_\rndq$) and $\rnddsum\ge 27$.
    Let $\rndd_1 \ldots \rndd_\rnddcnt$ be positive numbers, with $\rndd_i=\rndD$ for $i\in [\rnddcnt-1]$, 
    $\rndd_\rnddcnt \in [\rndD]$ and $\sum_{i=1}^\rnddcnt \rndd_i=\rnddsum$.  Let $A=\{a_1 \ldots a_\rnddcnt \}$ and $B=\{b_1 \ldots b_\rnddsum\}$ be disjoint sets of vertices in a random graph $\rndG$ where the probability to have an edge $\{a_i,b_j\}$ is $p=2\rndk/\rndn$ for any $i$ and $j$.
  Let
  \begin{equation}
    \label{eq:lm.23.rndp}
  \rndp \ge 
  1-\left(1-\max\left(4, 1+\frac{3\log(2e/\rnddeltab)}{2\log(\rnddsum)}\right)\frac{\log(\rnddsum)}{\rnddsum}\right)^\rndD.
  \end{equation}
  Then with probability at least $1-\rnddeltab$, $\rndG$ contains a collection
  of disjoint $d_i$-stars with centers $a_i$ and leaves in $B$. 
\end{lemma}
\begin{proof}
  Define an auxiliary random bipartite graph $\rndGbis$ with sides $A^\prime=\{a^\prime_1 \ldots a^\prime_{\rnddsum} \}$ and $B=\{b_1 \ldots b_\rnddsum \}$.  For every $i,j\in[\rnddsum]$, the probability of having an edge between $a_i$ and $b_j$ in $\rndGbis$ is $\rndpbis=1-(1-\rndp)^{1/\rndD}$.  We relate the distributions on the edges of $\rndG$ and $\rndGbis$ by requiring there is an edge between $a_i$ and $b_j$ if and only if there is an edge between $a^\prime_{\rndD (i-1) + i^\prime}$ and $b_j$ for all $i^\prime \in [\rndD]$.

  From Equation (\ref{eq:lm.23.rndp}) we can derive
  \begin{equation}
    \label{eq:lm.24.rndp}
  \rndpbis \ge 
  \max\left(4, 1+\frac{3\log(2e/\rnddeltab)}{2\log(\rnddsum)}\right)\frac{\log(\rnddsum)}{\rnddsum}.
  \end{equation}
  
  Setting  \[
   \rndzeta = \max\left(4, 1+\frac{3\log(2e/\rnddeltab)}{2\log(\rnddsum)}\right),
   \]
   this ensures $\rndpbis$ satisfies the constraints of Lemma \ref{lm:Bollo.73.delta}:
   \[\rndpbis = \rndzeta \log(\rnddsum)/\rnddsum.\]
  
   As a result, there is a complete matching in $\rndGbis$ with probability at least $1-\rnddeltab$, and hence the required stars can be found in $\rndG$ with probability at least $1-\rnddeltab$.
   \qed
\end{proof}

\begin{lemma}
  \label{lm:rnd.embedT}
  Let $\rnddeltaa>0$ and $\rnddeltab>0$.
  Let $\rndG$ and $\rndT$ and their associated variables be as defined above.
  Assume that the following conditions are satisfied:
  \begin{eqnarray*}
  &(a)& \rndn\ge 27(\rndD_\rndq+2)/\rndD_\rndq, \\
  &(b)& \rndgamma +2(\rndD_\rndq+2)^{-1}+\rndn^{-1} \le 1, \\    
  &(c)& \rndk\ge 4\log(2\rndn/\rnddeltaa(\rndD_\rndq+2)), \\
    &(d)& \rndk \ge  \max\left(4, 1+\frac{3\log(2e/\rnddeltab)}{2\log(\rndn\rndD_\rndq/(\rndD_\rndq+2))}\right)\frac{\rndD_\rndq+2}{2} \log\left(\frac{\rndn\rndD_\rndq}{\rndD_\rndq+2}\right), \\
    &(e)& \rndgamma = \max{}_{l=1}^{\rndq-1} \rndD_l/\rndk.
    \end{eqnarray*}

  Let $\rndG$ be a random graph where there is an edge between any two vertices with probability $\rndp$.  Let $\rndroot$ be a vertex of $\rndG$.
  Then, with probability at least $1-\rnddeltaa-\rnddeltab$, there is a subgraph isomorphism between the tree $\rndT$ defined above and $\rndG$ such that the root of $\rndT$ is mapped on $\rndroot$.
\end{lemma}

\begin{proof}

  The conditions of Lemma \ref{lm:rnd.22} are clearly satisfied, so with
  probability $1-\rnddeltaa$ there is a tree isomorphic to $\rndF$ in $\rndG$.
  Then, from condition (d) above and knowing that the edge probability is $\rndp=2\rndk/\rndn$, we obtain
  \[
  \rndp \ge  \max\left(4, 1+\frac{3\log(2e/\rnddeltab)}{2\log(\rndn\rndD_\rndq/(\rndD_\rndq+2))}\right)\frac{1}{\rndn\rndD_\rndq/(\rndD_\rndq+2)} \log\left(\frac{\rndn\rndD_\rndq}{\rndD_\rndq+2}\right)\rndD_\rndq.\]
  Taking into account that $\rndm=\rndn\rndD_\rndq/(\rndD_\rndq+2)$, we get
  \[
  \rndp \ge  \max\left(4, 1+\frac{3\log(2e/\rnddeltab)}{2\log(\rndm)}\right)\frac{1}{\rndm} \log\left(\rndm\right)\rndD_\rndq,\]
which implies the condition on $\rndp$ in Lemma \ref{lm:Bollo.73.delta}.  The other conditions of that lemma can be easily verified.  
As a result, with probability at least $1-\rnddeltab$ there is a set of stars in $\rndG$ linking the leaves of $\rndF$ to the leaves of $\rndT$, so we can embed $\rndT$ completely in $\rndG$.
\qed
\end{proof}

\subsubsection{Running \gopa{} on Random Graphs}
\label{sec:dp.rnd.gopa}

Assume we run \gopa{} on a random graph satisfying the properties above, what can
we say about the differential privacy guarantees?  According to Theorem~
\ref{thm:diffpriv.etadelta}, it is sufficient that there exists a spanning tree
and vectors $t_\eta$ and $t_\Delta$ such that $t_\eta + K t_\Delta = X^A-X^B$.
We fix $t_\eta$ in the same way as for the other discussed topologies (see
sections \ref{sec:priv.worst-dp}  and \ref{sec:priv.completeGraph}) in order to achieve the desired $\sigma_\eta$ and focus our
attention on $t_\Delta$.
According to Lemma~\ref{lm:rnd.embedT}, with high probability there exists in
$G^H$ a spanning tree rooted at the vertex where $X^A$ and $X^B$ differ and a branching factor $\rndD_l$ specified per level.  So given a random graph on $\rndn$ vertices with edge density $2\rndk/\rndn$, if the conditions of Lemma~\ref{lm:rnd.embedT} are satisfied we can find such a tree isomorphic to $\rndT$ in the communication graph between honest users $G^H$.  
In many cases (reasonably large $\rndn$ and $\rndk$), this means that the lemma guarantees a spanning tree with branching factor as high as $O(\rndk)$, even though it may be desirable to select a small value for the branching factor of the last level in order to more easily satisfy condition (d) of Lemma \ref{lm:rnd.embedT}, e.g., $\rndD_\rndq=2$ or even $\rndD_\rndq=1$.

\begin{lemma}
  \label{lm:rnd.tdeltaBound}
  Under the conditions described above,
  \begin{eqnarray}
    t_\Delta^\top t_\Delta &\le& \frac{1}{\rndD_1}\left(1+\frac{1}{\rndD_2}\left(1+\frac{1}{\rndD_3}\left(\ldots \frac{1}{\rndD_\rndq}\right)\right)\right) +\frac{(\rndD_\rndq+2)(\rndD_\rndq+2+2\rndq)}{\rndn}
    \label{eq:rnd.dp.1}\\
    &\le& \frac{1}{\rndD_1}\left(1+\frac{2}{\rndD_2}\right)+O(\rndn^{-1}).\nonumber
  \end{eqnarray}
\end{lemma}
\begin{proof}
  Let $q$ be the depth of the tree $\rndT$.
  The tree is balanced, so in every node the number of vertices in the subtrees of its children differs in at most $\rndD_\rndq+2$.
  For edges $e$ incident with the root (at level 0 node), $|t_e-\rndD_1^{-1}|\le \rndn^{-1}(\rndD_\rndq+2)$.
  In general, for a  node at level $l$ (except leaves or parents of leaves),
  there are $\prod_{i=1}^l \rndD_i$ vertices, each of which have $\rndD_{l+1}$ children, and for every edge $e$ connecting such a node with a child, \[
  \left|t_e-\prod_{i=1}^{l+1} \rndD_i^{-1}\right|\le (\rndD_\rndq+2)/\rndn.
  \]
  For a complete tree (of $1+\rndD+\ldots +\rndD^q$ vertices), we would have
  \begin{eqnarray*}
    t_\Delta^\top t_\Delta &=& \sum_{l=1}^{\rndq} \prod_{i=1}^l \rndD_i \left(\prod_{i=1}^l \rndD_i^{-1}\right)^2 
    = \sum_{l=1}^{q} \left(\prod_{i=1}^l \rndD_i\right)^{-1},
  \end{eqnarray*}
  which corresponds to the first term in Equation (\ref{eq:rnd.dp.1}).
  As the tree may not be complete, i.e., there may be less than $\prod_{i=1}^q \rndD_i$ leaves, we analyze how much off the above estimate is.
  For an edge $e$ connecting a vertex of level $l$ with one of its children,
  \[\left|t_e -\prod_{i=1}^{l+1}\rndD_i\right|\le (\rndD_\rndq+2)/\rndn,\] and hence

  \begin{eqnarray*}
    \left|t_e^2 - \left(\prod_{i=1}^{l+1}\rndD_i\right)^2\right|
    &\le& \left|t_e - \prod_{i=1}^{l+1}\rndD_i\right|\left(t_e + \prod_{i=1}^{l+1}\rndD_i\right) \\
    &\le& \frac{\rndD_\rndq+2}{\rndn} \left(t_e - \prod_{i=1}^{l+1}\rndD_i+ 2\prod_{i=1}^{l+1}\rndD_i\right) \\
    &\le & \left(\frac{\rndD_\rndq+2}{\rndn}\right)^2 + 2 \frac{\rndD_\rndq+2}{\rndn} \prod_{i=1}^{l+1}\rndD_i.
  \end{eqnarray*}
  Summing over all edges gives
  \begin{eqnarray*}
    t_\Delta^\top t_\Delta -\sum_{l=1}^{q} \left(\prod_{i=1}^l \rndD_i\right)^{-1}  &\le& \sum_{l=1}^\rndq \rndlev_l\left( (\rndD_\rndq+2)/\rndn^2 + 2\left(\prod_{i=1}^{l+1}\rndD_i\right)^{-1}(\rndD_\rndq+2)/\rndn\right) \\
    &= & (\rndD_\rndq+2)^2/\rndn + \sum_{l=1}^\rndq 2\rndlev_l\left(\prod_{i=1}^{l+1}\rndD_i\right)^{-1}(\rndD_\rndq+2)/\rndn \\
    &\le & (\rndD_\rndq+2)^2/\rndn + \sum_{l=1}^\rndq 2(\rndD_\rndq+2)/\rndn \\
    &=& \frac{(\rndD_\rndq+2)(\rndD_\rndq+2+2\rndq)}{\rndn}.\smartqed
  \end{eqnarray*}
  \qed
\end{proof}

So if we choose parameters $\rndD$ for the tree $\rndT$, the above lemmas provide values $\rnddeltaa$ and $\rnddeltab$  such that $\rndT$ can be embedded in $G^H$ with probability at least $1-\rnddeltaa-\rnddeltab$ and an upper bound for $t_\Delta^\top t_\Delta$ that can be obtained with the resulting spanning tree in $G^H$.

Theorem~\ref{thm:diffprivacy-random} in the main text summarizes these results, simplifying the conditions by assuming that $\rndD_i=\lfloor(k-1)\honestfraction/2\rfloor$ for $i\le \rndq-1$ and $\rndD_\rndq = 2$.

\begin{proof}[Proof of Theorem~\ref{thm:diffprivacy-random}]
  Let us choose $\rndD_i=\lfloor(k-1)\honestfraction/3\rfloor$ for $i\in
  [\rndq-1]$ and $\rndD_\rndq=1$ for some appropriate $\rndq$ such that
  Equation (\ref{eq:rnd.Dvsq}) is satisfied. We also set $\delta=\rnddeltaa=\rnddeltab$.

  Then, the conditions of Lemma~\ref{lm:rnd.embedT} are satisfied.  In
  particular, condition (a) holds as $\rndn=\honestfraction n \ge 81 = 27
  (\rndD_\rndq + 2) /\rndD_\rndq$.  Condition (e) implies that
  \begin{eqnarray*}
    \rndgamma &=& \max_{i=1}^{\rndq-1} \rndD_i/\rndk 
    = \frac{1}{\rndk} \left\lfloor\frac{(k-1)\honestfraction}{3}\right\rfloor.
    \end{eqnarray*}
  Condition (b) holds as
  \begin{eqnarray*}
    \rndgamma + 2(\rndD_\rndq + 2)^{-1} + \rndn^{-1}
    &=&  \frac{1}{\rndk} \left\lfloor\frac{(k-1)\honestfraction}{3}\right\rfloor + \frac{2}{3} + \rndn^{-1} \\
    & \le & \frac{1}{\rndk} \frac{(k-1)\honestfraction}{3} + \frac{2}{3} + \rndn^{-1} 
    \le  \frac{1}{3} - \frac{\honestfraction}{3\rndk} + \frac{2}{3} + \rndn^{-1} 
     =   \frac{1}{3} - \frac{1}{3k} + \frac{2}{3} + \rndn^{-1} \\
    & \le &  \frac{1}{3} - \frac{1}{\rndn} + \frac{2}{3} + \rndn^{-1} 
    = 1.
  \end{eqnarray*}
  Condition (d) holds because
  we know that $\honestfraction k \ge 6\log(\honestfraction n/3)$,
  which is equivalent to 
  \[
  \rndk \ge 4 \frac{\rndD+2}{\rndD} \log\left(\frac{\rndn\rndD_\rndq}{\rndD_\rndq + 2}\right),
    \]
    and we know that  $\honestfraction k \ge \frac{3}{2} + \frac{9}{4}\log
    (2e/\delta)$, which is equivalent to 
    \[
    \rndk\ge \left(1+\frac{3\log(2e/\rnddeltab)}{2\log(\rndn\rndD_\rndq/(\rndD_\rndq+2))}
    \right)\frac{\rndD+2}{\rndD} \log\left(\frac{\rndn\rndD_\rndq}{\rndD_\rndq + 2}\right).
      \]
      Finally, condition (c) is satisfied as we know that $\honestfraction k
      \ge 4\log(\honestfraction n/3\delta)$.
      Therefore, applying the lemma, we can with probability at least $1-
      2\delta$ find a spanning tree isomorphic to $\rndT$. If we find one, Lemma 
      \ref{lm:rnd.tdeltaBound} implies that
      \begin{eqnarray*}
        t_\Delta^\top t_\Delta &\le&
        \sum_{l=1}^\rndq \left(\prod_{i=1}^l \rndD_i\right)^{-1} + \frac{(\rndD_q+2)(\rndD_q+2+2\rndq)}{\rndn} \\
        &=& \sum_{l=1}^{\rndq-1} \rndD_1^{-l}   + \rndD_1^{1-\rndq}\rndD_\rndq^{-1} + \frac{3(3+2\rndq)}{\rndn}     
        = \rndD_1^{-1}\frac{1-\rndD_1^{1-\rndq}}{1-\rndD_1^{-1}} + \rndD_1^{1-\rndq}\rndD_\rndq^{-1} + \frac{9+6\rndq}{\rndn}  \\
        &\le& \frac{1}{\rndD_1-1} + \frac{3}{\rndn} + \frac{9+6\rndq}{\rndn}  
        = \frac{1}{\lfloor (k-1)\honestfraction/3 \rfloor-1} + 
        \frac{12+6\rndq}{\rndn}   \\
        &=& \frac{1}{\lfloor (k-1)\honestfraction/3 \rfloor-1} + \frac{12+6\log(\rndn)}{\rndn}  
      \end{eqnarray*}
      This implies the conditions related to $\sigma_\Delta$ and $t_\Delta$
      are satisfied.  From Theorem~\ref{thm:diffpriv.etadelta}, it follows that
      with probability $1-2\delta$ \gopa{} is $
      (\varepsilon,\delta)$-differentially private, or in short \gopa{} is $
      (\varepsilon,3\delta)$-differentially private.      
\end{proof}

\subsection{Matching the Utility of the Centralized Gaussian Mechanism} 
\label{app:dp-cor}

From the above theorems, we can now obtain a simple corollary which
precisely quantifies the amount of independent and pairwise noise needed
to achieve a desired privacy guarantee depending on the topology. 

\begin{proof}[Proof of Corollary~\ref{thm:dp_corollary}]
In the centralized (trusted curator) setting, the standard
centralized Gaussian
mechanism (\cite{Dwork2014a} Theorem A.1 therein) states that in order for the noisy average $
(\frac{1}{n}\sum_{u\in U}
 X_u) +
\eta$ to
be $
(\varepsilon',\delta')$-DP for some $\varepsilon',\delta'\in(0,1)$, the
variance of
$\eta$ needs to be:
\begin{eqnarray} 
\sigma_{gm}^2  = \frac{c^2}{(\varepsilon'
n )^2}.
\label{eq:var_cdp}
\end{eqnarray}
where $c^2 >  2 \log(1.25/\delta')$.

Based on this, we let the independent noise $\eta_u$ added by each user in
\gopa{} to have variance
\begin{eqnarray} 
\sigma_\eta^2 &=& \frac{n^2}{\rndn}\sigma_{gm}^2 = \frac{c^2}{
(\varepsilon')^2\rndn},
\label{eq:var_cdp_ind2}
\end{eqnarray}
which, for the approximate average $\hat{X}^{avg}$, gives a total variance of:
\begin{equation} 
\label{eq:var_cdp_ind1}
Var\Big(\frac{1}{\rndn}\sum_{u\in U^H}\eta_u\Big) =
\frac{1}{\rndn^2}
\rndn\sigma_\eta^2 = \frac{c^2}{(\varepsilon'\rndn)^2}.
\end{equation}
We can see that when $\rndn=n$ (no malicious user, no dropout), Equation 
\eqref{eq:var_cdp_ind1}
exactly corresponds to the variance required
by the centralized Gaussian mechanism in Equation \eqref{eq:var_cdp}, hence
\gopa
will
achieve the same utility. When there are malicious users and/or dropouts, each
honest user
needs to add a factor
$n/\rndn$ more noise to compensate.
for the fact that drop out users do not participate and malicious users can
subtract their own inputs and independent noise terms from $\hat{X}^{avg}$.
This is consistent with
previous work on distributed noise
generation under malicious parties \cite{Shi2011}.

Now, given some $\kappa >
0$, let $\sigma_\Delta^2= \kappa\sigma_\eta^2$ if $G$ is the complete graph,
$\sigma_\Delta^2 = \kappa\sigma_\eta^2\rndn(\frac{1}{\lfloor (k-1)\honestfraction/3 \rfloor-1} + 
  (12+6\log(\rndn))/\rndn)$
  for the random $k$-out graph, and $\sigma_\Delta^2 = \kappa
  \rndn^2\sigma_\eta^2/3$ for an arbitrary connected $G^H$.
In all cases, the value of $\theta$ in Theorems \ref{thm:diffprivacy},
\ref{thm:diffprivacy-complete} and \ref{thm:diffprivacy-random} after plugging
$\sigma_\Delta^2$ gives
  \begin{align*}
  \theta &= \frac{\varepsilon^2}{c^2} + \frac{\varepsilon^2}{\kappa c^2} 
  = \frac{ (\kappa+1) \varepsilon^2}{ \kappa c^2}.
  \end{align*}
  We set $\varepsilon=\varepsilon'$ and require that  $\theta \leq \thetamax(\varepsilon, \delta)$ as in conditions of Theorems \ref{thm:diffprivacy},
  \ref{thm:diffprivacy-complete} and \ref{thm:diffprivacy-random}. Then, by Equation \eqref{eq:etadeltabound4} we have 
  \begin{align*}
  \varepsilon &\geq \frac{ (\kappa+1) \varepsilon^2}{2 \kappa c^2} + \sqrt{\frac{ 
  (\kappa+1)}{\kappa}}\frac{\varepsilon}{c}.
  \end{align*} 
  For $d^2 = \frac{\kappa}{\kappa+1}c^2$ we can rewrite the above as
 $\varepsilon \geq \frac{\varepsilon^2}{2d^2} + \frac{\varepsilon}{d}$.
  Since $\varepsilon\leq 1$, this is satisfied if $d - \frac{\varepsilon}
  {2d} \geq 1$ and in turn when $d \geq 3/2$, or equivalently
  when $c
  \geq \frac{3}{2} \sqrt{\frac{\kappa+1}{\kappa}}$. Now analyzing the
  inequality in Equation
  \eqref{eq:etadeltabound5}
  we have:
  \begin{align*}
  \Big(\varepsilon - \frac{ (k+1) \varepsilon^2}{2 \kappa c^2} \Big)^2 & \geq 2\log(2/\delta\sqrt{2\pi})
  \Big(\frac{\varepsilon^2}{c^2} + \frac{\varepsilon^2}{\kappa c^2}\Big)\\
  \varepsilon^2 + \frac{(\kappa+1)^2\varepsilon^4}{4\kappa^2c^4} - \frac{(\kappa+1)\varepsilon^3}{\kappa c^2} &
  \geq 2\log(2/\delta\sqrt{2\pi})
  \Big(\frac{(\kappa+1)\varepsilon^2}{\kappa c^2}\Big)\\
  \frac{1}{2}\Big(\frac{\kappa c^2}{\kappa+1} + \frac{(\kappa+1)\varepsilon^2}
  {4\kappa c^2} -
  \varepsilon\Big) & \geq \log(2/\delta\sqrt{2\pi}).
  \end{align*}
  Again denoting $d^2=\frac{\kappa}{\kappa+1}c^2$ we can rewrite the above as
  \[\frac{1}{2}\Big(d^2 + \frac{\varepsilon^2}{4d^2} - \varepsilon \Big) \geq \log(2/\delta\sqrt{2\pi}).\]
  For $d \geq 3/2$ and $\varepsilon \leq 1$, the derivative of ${d^2 + \frac{\varepsilon^2}{4d^2} - \varepsilon}$ is positive, so ${d^2 + \frac{\varepsilon^2}{4d^2} - \varepsilon > d^2 - 8/9}$. 
  Thus, we only require $d^2\geq 2\log(1.25/\delta)$.
  Therefore Equation \eqref{eq:etadeltabound5} is satisfied when:
  \begin{align*}
  \frac{\kappa}{\kappa+1}\log(1.25/\delta') & \geq \log(1.25/\delta),
  \end{align*}
  which is equivalent to
  $\delta \geq 1.25\Big( \frac{\delta'}{1.25} \Big)^{\frac{\kappa}
  {\kappa+1}}$.
  The constant $3.75$ instead of $1.25$ for the random $k$-out graph case is
  because Theorem~\ref{thm:diffprivacy-random} guarantees $
  (\varepsilon,3\delta)$-DP instead of $(\varepsilon,\delta)$ in
  Theorems~\ref{thm:diffprivacy} and \ref{thm:diffprivacy-complete}.
  \qed
\end{proof}

\subsection{Smaller $k$ and $\sigma_\Delta^2$ via Numerical Simulation}
\label{app:dp_numerical_sim}

\begin{table}[t]
\caption{Examples of admissible values for $k$ and $\sigma_\Delta$, obtained
by numerical simulation, to ensure $(\varepsilon,\delta)$-DP with trusted curator
utility for $\varepsilon=0.1$, $\delta'=1/\rndn^2$, $\delta=10\delta'$.}
\label{tab:cor2}
\vskip 0.1in
\centering
\begin{tabular}{|l|l|c|c|}
\hline
\multirow{4}{2cm}{$n=100$} & \multirow{2}{1.5cm}{$\rho=1$} 
& $k=3$ & $\sigma_\Delta=55.2$ \\\cline{3-4}
&& $k=5$ & $\sigma_\Delta=38.2$\\\cline{2-4}
& \multirow{2}{1.5cm}{$\rho=0.5$} 
& $k=20$ & $\sigma_\Delta=23.6$ \\\cline{3-4}
&& $k=30$ & $\sigma_\Delta=19.6$\\\hline
\multirow{4}{2cm}{$n=1000$} & \multirow{2}{1.5cm}{$\rho=1$} 
& $k=5$ & $\sigma_\Delta=59.9$ \\\cline{3-4}
&& $k=10$ & $\sigma_\Delta=37.8$\\\cline{2-4}
& \multirow{2}{1.5cm}{$\rho=0.5$} 
& $k=20$ & $\sigma_\Delta=42$ \\\cline{3-4}
&& $k=30$ & $\sigma_\Delta=28.5$\\\hline
\multirow{4}{2cm}{$n=10000$} & \multirow{2}{1.5cm}{$\rho=1$} 
& $k=10$ & $\sigma_\Delta=51.1$ \\\cline{3-4}
&& $k=20$ & $\sigma_\Delta=33.8$\\\cline{2-4}
& \multirow{2}{1.5cm}{$\rho=0.5$} 
& $k=20$ & $\sigma_\Delta=59.3$ \\\cline{3-4}
&& $k=40$ & $\sigma_\Delta=33.4$\\\hline
\end{tabular}
\vskip -0.1in
\end{table}

For random $k$-out graphs, the conditions on $k$ and $\sigma_\Delta^2$ given
by Theorem~\ref{thm:diffprivacy-random} are quite conservative. While we are
confident that they can be refined by resorting to tighter approximations in
our analysis in Section~\ref{app:random_graph}, an alternative option to find
smaller, yet admissible values for $k$ and $\sigma_\Delta^2$ is to resort to
numerical simulation.

Given the number of users $n$, the proportion $\rho$ of nodes who are honest
and do not drop out, and a
value for $k$, we implemented a program that generates a random $k$-out graph,
checks if the subgraph $G^H$ is connected, and if so finds a
suitable spanning tree for $G^H$ and computes
the corresponding value for $t_\Delta^\top t_\Delta$ needed by our
differential privacy analysis (see for instance
sections \ref{sec:priv.worst-dp} and \ref{sec:priv.completeGraph}).
From this, we can in turn deduce a sufficient value for $\sigma_\Delta^2$
using Corollary~\ref{thm:dp_corollary}.

Table~\ref{tab:cor2} gives examples of values obtained by simulations for
various values of $n$, $\rho$ and several choices for $k$. In each case, the
reported $\sigma_\Delta$
corresponds to the \emph{worst-case} value required across $10^5$ random runs,
and the chosen value of $k$ was large enough for $G^H$ to be connected in 
\emph{all} runs. This was the case even for slightly smaller values of $k$.
Therefore, the values reported in Table~\ref{tab:cor2} can be considered safe to use in practice.


\section{Details of Security and Cryptographic Aspects}
\label{app:crypto} 

We detail in this appendix some of components of the Verification Protocol of Section \ref{sec:verif}. We describe
the Phase \ref{enum:verif.phase.setup} Setup  in Appendix \ref{app:crypto.setup},
robustness after dropouts of the Phase \ref{enum:verif.phase.mitigate} Mitigation in Appendix  \ref{app:crypto.drop-out}, 
on measures against attacks on efficiency in Appendix \ref{app:crypto.efficiency} and
of issues that could generate the use finite precision representations in Appendix \ref{app:crypto.finite_precision}.

\subsection{Setup Phase}
\label{app:crypto.setup} 

Our verification protocol requires public unbiased randomness to generate Pedersen commitment parameters $\comschpp$ and private random seeds to generate Gaussian samples for Property \eqref{eq:verif3}. We describe below how to perform these tasks in a Setup phase.

\paragraph{Public randomness}
To generate a public random seed, a simple procedure is the following.  First, all users draw uniformly a random number and publish a commit to it.  When all users have done so, they all reveal their random number.  Then, they sum the random numbers (modulo the order $q$ of the cyclic group) and use the result as public random seed.  If at least one user was honest and drew a random number, this sum is random too, so no user can both claim to be honest and claim that the obtained seed is not random. Finally, the amount of randomness of the seed is expanded by the use of a cryptographic hash function. 
Appendix  \ref{app:supp.rand} provides in more detail a folklore method which evenly distributes the work over users.

\paragraph{Private Seeds}
\label{app.crypto.setup.sds} 
In a second part of Setup, users collaboratively generate samples $\pseed_1, \dots, \pseed_n$ such that, for all $u \in \userset$,  $\pseed_u$ is private to $u$ and has uniform distribution in the interval $[0 ,  M-1]$ for some public integer $M < q/2$, the number of bins to generate Gaussian samples (in Appendix \ref{app:supp.gauss}).  In particular,

\begin{enumerate} 
	\item \textbf{For all $u \in U$:} $u$ draws uniformly $z_u \in [0, M-1]$, and publishes $\mathbf{c_{z_u}} = Com_\comschpp(z)$. 
        \item The users draw a public, uniformly distributed random number $z$ (as above).
	\item \textbf{For all $u \in U$:} $u$ computes  $\pseed_u  \gets z + t_u \mod M$ and publishes $\mathbf{c_{\pseed_u}} \gets Com_\comschpp(\pseed_u)$ together with the proof of the modular sum (see the $\Sigma$-protocol for modular sum in \cite{camenisch_proving_1999}). 
\end{enumerate}
It is important that the $c_{z_u}$ are published before generating the public random $z$ to avoid that users would try to generate several $r_u$ and check which one is most convenient for them.

\subsection{Dealing with Dropout}
\label{app:crypto.drop-out}

In this section, we give additional details on the strategies for dealing with
dropout outlined in Section~\ref{sec:gopa}.
We consider that a
user drops out of the computation if they is off-line for a period which is
too long for the community to wait until their return. This can
happen accidentally to honest users, e.g. due to lost network
connection. Malicious users may
also intentionally drop out to affect the progress
of the computation.  Finally, a user detected as cheater by the verification procedure of Section \ref{sec:verif} and banned from the system may also be seen as a dropout.

Unaddressed drop outs affect the outcome of the computation as we rely on the fact that pairwise noise terms $\Delta_{u,v}$ and $\Delta_{v,u}$ cancel out for the correctness and utility of the computation.

We propose a three-step approach for handling dropout:
\begin{enumerate}
	\item First, as a preventive measure, users should exchange pairwise noise
	with enough users so that the desired privacy guarantees hold even if some
	neighbors drop out. This is what we proposed and theoretically analyzed in
	Section~\ref{sec:graphs}, where the
	number of neighbors $k$ in Theorem~\ref{thm:diffprivacy-random} depends on (a lower bound on) the proportion $\rho$ of honest users who do not
	drop out. It is important to use a safe lower bound on $\rho$ to make sure
	users will have enough safety margin to handle actual dropouts.
	\item \label{step:dropout.rollback} Second, as long as there is time, users attempt to repair as much as possible
	the problems incurred by dropouts. A user $u$ who did
	not publish ${\hat{X}}_u$ yet can just remove the corresponding pairwise noise
	(and possibly exchange noise with another active user instead).
	Second, a user
	$u$ who did publish ${\hat{X}}_u$ already but has still some safety margin
	thanks to step 1 can simply reveal the noise exchanged with the user who
	dropped out, and subtract it from his published ${\hat{X}}_u$.

      \item Third, it is possible that a user $u$ did publish ${\hat{X}}$ and afterwards still so many neighbors drop out that revealing all exchanged pairwise noise would affect the privacy guarantees for $u$.  If that happens it means a significant fraction of the neighbors of $u$ dropped out, while the neighbors of $u$ form a random sample of all users.  In such case, it is likely that also globally many users dropped out.  If caused by a large-scale network failure the best strategy could be to just wait longer than initially planned.  Else, given that $u$ is unable to reveal more pairwise noise without risking their privacy, the only options are either that $u$ discards all his pairwise noise and restarts with a new set of neighbors, or that $u$ doesn't address the problem and his pairwise noise is not compensated by the noise of another active user.
To avoid such problems, in addition to step 1, it can be useful to check which
users went off-line just before publishing ${\hat{X}}$ and to have penalties
for users who (repeatedly) drop out at the most inconvenient times. 
Note that for an adversary
	who wants to remain undetected while performing this attack, the behavior 
	of the involved colluding users would need to be statistically
   indistinguishable from that of incidental dropouts. This would result
	in an attack with no extra impact than the one caused by such dropouts.  
	\item Finally, when circumstances require and allow it, we can ignore the
	remaining problems and proceed with the algorithm,
	which will then output an estimated average with slightly larger error.
	This can be the case for instance when only a few drop outs have not yet been
	resolved, there is not much time available, and the
	corresponding pairwise terms $\Delta_{u,v}$ are known to be not too large 
	(e.g., by proving that they were drawn from $\mathcal{N}(0,\sigma_\Delta^2)$
	where $\sigma_\Delta^2$ is small enough, or by the use of range proofs).
\end{enumerate}

\begin{figure}[htb] 
	\centering
	\includegraphics[width=.45\linewidth]{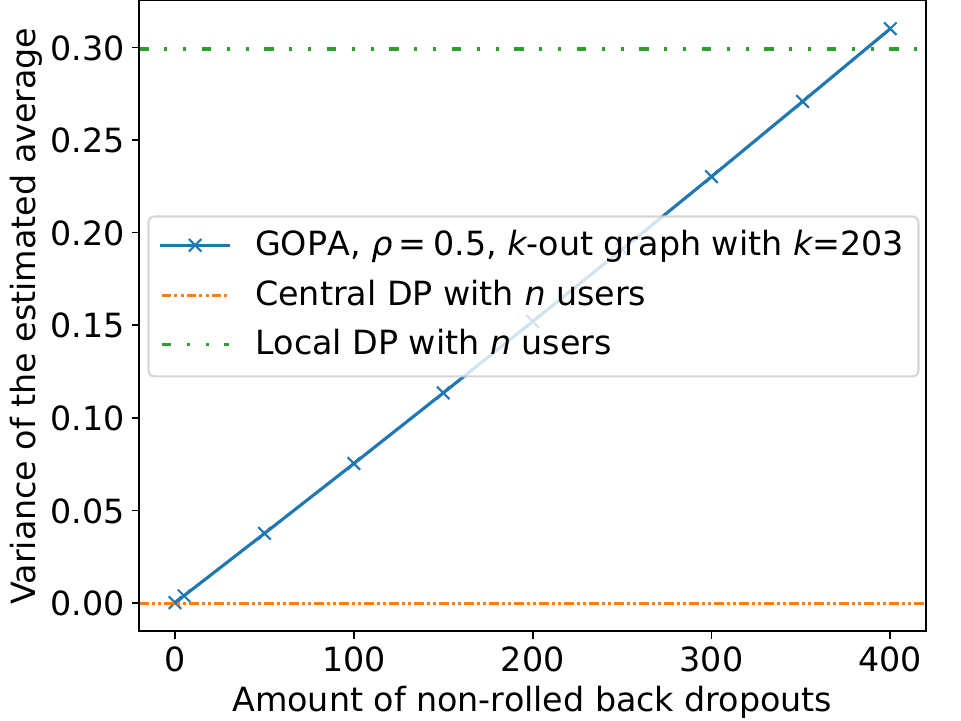}
	\caption{Impact of non-rolled back dropouts on the utility of \gopa. See
		text for details.}
	\label{fig:nonroll_dropout}
\end{figure}

Figure~\ref{fig:nonroll_dropout} illustrates with a simple simulation the impact of a small number of residual pairwise noise terms of variance $\sigma_\Delta^2$ in the final result, which may happen in the rare circumstances that the pairwise noise terms of some users who dropped out are not ``rolled
back'' by their neighbors (step \ref{step:dropout.rollback} above).
We consider $n=10000$, $\rho=0.5$, $\varepsilon=0.1$, $\delta=10/(\rho n)^2$, $\kappa = 0.3$ and set the values of $k$, $\sigma^2_\eta$ and $\sigma^2_\Delta$ using Corollary~\ref{thm:dp_corollary} so that \gopa satisfies $(\varepsilon,\delta)$-DP (in particular we set them in the same way as in Table~\ref{tab:cor}).
We simulate this by drawing a random $k$-out graph, selecting a certain number of dropout users at random, marking all their exchanged noise as not rolled back (in practice its is also possible that part of their noise gets rolled back) and computing the variance of the estimated average. The simulation is averaged over $100$ runs (even though the standard deviation across random runs is negligible).
We see that \gopa can tolerate a number of ``catastrophic drop-outs''
while remaining more accurate than local DP.
This ability to retain a useful estimate despite
residual pairwise noise terms is rather unique to \gopa, and not possible with
secure
aggregation methods which typically use uniformly random pairwise masks 
\cite{Bonawitz2017a}. 
We note that this robustness can be optimized by choosing smaller
$\sigma^2_\Delta$ (i.e., smaller $\kappa$)
and compensating by adding a bit more independent noise according to Corollary~\ref{thm:dp_corollary}.

\subsection{Robustness Against Attacks on Efficiency}
\label{app:crypto.efficiency} 

In this section, we study several attacks, their impact and \gopa{}'s defense against them.

\paragraph{Dropout}

A malicious user could drop out of the computation intentionally, with the impact described in Section~\ref{app:crypto.drop-out}.
However, dropping out is bad for the reputation of a user, and users dropping out more often than agreed could be banned from future participation. One can use techniques from identity management (a separate branch in the field of security) to ensure that creating new accounts is not free, and that possibilities to create new accounts are limited, e.g., by requiring to link accounts to unique persons, unique bank account or similar. This also ensures that the risk of being banned outweighs the incentive of the (bounded) delay of the system one could cause by intentionally dropping out.

\paragraph{Flooding with Neighbor Requests}

In Section \ref{sec:privacy}, we discuss privacy guarantees for complete graphs, path
graphs and random communication graphs.
In the case of a complete communication graph, all users exchange noise with all other users.  This is slow, but there is no opportunity for malicious users to interfere with the selection of neighbors as the set of neighbors is fixed.
In other cases, e.g., when the number of users is too large and every user selects $k$ neighbors randomly as in Section \ref{sec:priv.random}, one could wonder
whether colluding malicious users could select neighbors in such a way that
the good working of the algorithm is disturbed, e.g., a single honest user is
flooded with neighbor requests and ends up exchanging noise with $O(n)$ others.

We first stress the fact that detecting such attacks is easy.  If
all agents randomly select neighbors uniformly at random as they should, then
every agent expects to receive $k$ neighbor invitations, perhaps plus a few
standard deviations of order $O(\sqrt{k})$.
As soon as a user receives a sufficiently unlikely number of neighbor invitations, we know that with overwhelming probability the user is targeted by malicious users.

To avoid this, we can let all users select neighbors in a deterministic way starting from a public random seed (e.g., take the public randomness generated in Section \ref{app:crypto.setup}, add the ID of the user to it, apply a hash function to the sum, and use the result to start selecting neighbors).  In this way, neighbor selection is public and can't be tampered with.  It is possible some neighbors of a user $u$ were off-line and $u$ skipped them, but unless so many users are off-line that the community should have noticed severe problems $u$ should be able to find enough neighbors among the first $ck$ in his random sequence for a small constant $c$.

\paragraph{Other Common Attacks}

We assume the algorithm is implemented on a system which is secure according to classic network-related attacks, such as denial-of-service (DoS) attacks. 
Such attacks are located at the network level rather than at the algorithm level. As such, they apply similarly to any distributed algorithm requiring communication over a network.  To the extent such attacks can be mitigated, the solutions are on the network level, including (among others) a correct organization of the network and its routers.

Similarly, we assume that all (honest) communication is secure and properly
encrypted, referring the reader to the state-of-the-art literature for details
on possible implementations.

\subsection{Further Discussion on the Impact of Finite Precision }
\label{app:crypto.finite_precision}

In practice, we cannot work with real numbers but only with finite precision
approximations, see Section \ref{sec:verif.crypto}. We provide
a brief discussion of the impact of this on the guarantees offered by the protocol.
There is already a large body of work addressing issues which could arise
because of finite precision. Here are the main points:
\begin{enumerate} 
	\item Finite precision can be an issue for differential privacy in general,
	(see e.g. \cite{balcer_et_al:LIPIcs:2018:8353} for a study of the effect
	of floating point representations). Issues can be
	overcome with some care, and our additional encryption
	does not make the problem worse (in fact we can argue that encryption typically uses more bits and in our setting this may help).
      \item
        The issue of finite precision has been studied in cryptography. While some operations such as multiplication can cause difficulties in the context of homomorphic encryption, in our work we use a partially homomorphic scheme with only addition.   As a result, we can just represent our numbers with as many bits (after the decimal dot) as in plaintext memory. 
      \item If the number of users is high (and hence also the sum of the $X_u$ and the $\Delta_{u,v}$ variables), working up to the needed precision doesn't cause a cost which is high compared to the cost of the digits before the decimal dot.
\end{enumerate}


\section{Complexity of \gopa{}}
\label{app:complexity}

The cost of the protocol in terms of computation and communication was
summarized in Theorem~\ref{thm:cryptocost}.
This section summarizes each component of this cost, relying in
particular on the cost of proving computations in  \ref{sec:verif}.

Dominant computations are exponentiations in the cryptographic group $\G$ defined in Section \ref{sec:verif.crypto}. The cost of signing messages is negligible.
We describe costs centered on any user $u \in \userset$ . For simplicity, when we say a task costs $c$, it means that costs at most $c$ computations 
for proving a computation, $c$ computations for another party verifying this computation and it requires the exchange of 
$c\groupelsize$ bits, where $\groupelsize$ is the size in bits of an element of $\G$.
Some computations depend on fixed-precision parameter $\gdprec$ to represent numbers in $\Z_q$ (see Section \ref{sec:verif.crypto}) and the amount $1/B$ of equiprobable bins used to sample independent noise $\eta_u$ (see Appendix \ref{app:supp.gauss}).

The costs break down as follows:
\begin{itemize}
\item The Phase \ref{enum:verif.phase.setup} Setup requires generating public randomness which has constant cost (except for $O(1)$ parties that perform some extra computations, see Algorithm \ref{alg:genrand}) and private seeds that require 2 range proofs in the interval $[0, 1/B\gdprec]$, with a final cost of $20\log_2 (1/B) + 20\log_2 (1/\gdprec)$.

\item Validity of input at  Phase \ref{enum:verif.phase.verif}  Verification requires a range proof in the 
interval $[0,1/\gdprec]$, with a cost of $10\log_2(1/\gdprec)$. Extra computations of consistency cannot be accounted here
as they depend on the nature of external computations.

\item Correctness of computations of properties \eqref{eq:verif1} and \eqref{eq:verif2} at  Phase \ref{enum:verif.phase.verif} Verification cost at most $5|N(u)| + 4$, accounting the computations over commitments and proofs of knowledge 
of terms of each property.

\item The verification of Property \eqref{eq:verif3} at  Phase \ref{enum:verif.phase.verif} Verification costs $\ln(1/B) \cdot (34.1\ln(1/\gdprec) + 7.22\ln(q))$ (see Appendix \ref{app:supp.gauss}). 
\end{itemize} 

The overall cost of a protocol for user $u$ is at most of 

\[ 5|N_u| + 20\log_2 (1/B) + 30\log_2 (1/\gdprec) + \ln(1/B) \cdot (34.1\ln(1/\gdprec) + 7.22\ln(q)) + 4.\]

\newpage	

\section*{Supplementary Material}

This supplementary materials file contains background, often a summary of what can be found elsewhere in literature, and further details, which are not essential for the elaboration of the contribution.


\section{Further Details on Cryptography and Security}
\label{app:supplementary}

\subsection{Commitment Schemes }
\label{app:supp.cmt}

We start by defining formally  a commitment and its properties.

\begin{definition}[Commitment Scheme] 
	\label{def:commitment_scheme}
	A commitment scheme consists of a pair of (computationally efficient) algorithms $(Setup, Com)$. The setup algorithm $Setup$ is executed once, with randomness $t$ as input, and outputs a tuple $\comschpp \gets Setup(t)$, which is called the set of parameters of the scheme.
	The algorithm $Com$ with parameters $\comschpp$, denoted $Com_{\comschpp}$, is a function ${Com_{\comschpp}:\mathcal{M}_{\comschpp} \times \mathcal{R}_{\comschpp} \rightarrow \mathcal{C}_{\comschpp}}$, where $\mathcal{M}_{\comschpp}$ is called the message space, $\mathcal{R}_{\comschpp}$ the randomness space, and $\mathcal{C}_{\comschpp}$ the commitment space. For a message $m \in \mathcal{M}_{\comschpp}$, the algorithm draws $r \in \mathcal{R}_{\comschpp}$ uniformly at random and computes commitment $\mathbf{c} \gets Com_{\comschpp}(m,r)$. 
\end{definition}
The security of a commitment scheme typically depends on $t$ not being biased, in particular, it must be hard to guess non-trivial information about $t$.

We now define some key properties of commitments.
\begin{property}[Hiding Property]
	A commitment scheme is \emph{hiding} if, for all secrets  $x \in 
	\mathcal{M}_{\comschpp}$ and given that $r$ is chosen uniformly at random from $\mathcal{R}_{\comschpp}$, the commitment $\mathbf{c_x} = Com_{\comschpp}(x,r)$ does not reveal any information about $x$.
\end{property}

\begin{property}[Binding Property]
	A commitment is \emph{binding} if there exists no computationally efficient
	algorithm $\mathcal{A}$ that can find $x_1,x_2 \in \mathcal{M}_{\comschpp}$, $r_1, r_2 \in \mathcal{R}_{\comschpp}$ such that $x_1 \neq x_2$ and ${Com_{\comschpp}(x_1,r_1) = Com_{\comschpp}(x_2,r_2)}$. 
\end{property}

\begin{property}[Homomorphic Property] 
	\label{prop:homomorphic} 
	A \emph{homomorphic} commitment scheme is a commitment scheme such that $
	\mathcal{M}_{\comschpp}$, $\mathcal{R}_{\comschpp}$ and $\mathcal{C}_{\comschpp}$ are abelian groups, and for all $x_1, x_2 \in \mathcal{M}_{\comschpp}$, $r_1, r_2 \in \mathcal{R}_{\comschpp}$ we have
	\[
	Com_{\comschpp}(x_1,r_1) + Com_{\comschpp}(x_2,r_2) = Com_{\comschpp}(x_1 + x_2, r_1 + r_2).
	\] 
\end{property}
Please note that the three occurrences of the '$+$' sign in the above
definition are operations in three difference spaces, and hence may have
different definitions and do not necessarily correspond to normal addition of
numbers.

\paragraph{Pedersen Commitments} 

Let $p$ and $q$ be two large primes such that $q$ divides $p-1$,  and a let  $
\mathbb{G}$ be cyclic subgroup of order $q$ of the multiplicative group $\Z^*_p$. For such group, we have that  $\mathbb{G} = \{ a^i \mod p \mid 0 \leq i < q \}$ for any $a \in \mathbb{G}$ distinct to $1$. In the Pedersen scheme, $Setup$ is a function that samples at random parameters  $\comschpp = (\mathbb{G}, g,h)$ where  $g, h$ are two generators of $\G$  
Additionally, $
\mathcal{M}_\comschpp = \mathcal{R}_\comschpp = \Z_q$, and $\mathcal{C}_\comschpp = \mathbb{G}$.
We will refer to $g$ and $h$ as \emph{bases}.
The commitment function
$Com_{\comschpp}$  is defined as 
\begin{eqnarray} 
	Com_{\comschpp} &:& \Z_q \times \Z_q \rightarrow \mathbb{G} \nonumber \\
	Com_{\comschpp} (x, r) &=& g^x \cdot h^r, 
\end{eqnarray}
where $(\cdot)$ is the  product modulo $p$ of group $\mathbb{G}$, $x$ is the
secret and $r$ is the randomness. For simplification and when $r$ is not relevant, we relax the notation $Com_\comschpp(x,r)$ to $Com_\comschpp(x)$ and assume $r$ is drawn appropriately.  Pedersen commitments are homomorphic (in
particular, $Com_{\comschpp}(x+y,r+s)=Com_{\comschpp}(x,r)\cdot Com_
{\comschpp}(y,s)$), unconditionally hiding, and  computationally binding under
the Discrete Logarithm Assumption, which we succinctly describe below.  For more details, see Chapter 7 of \cite{katz_introduction_2014}.

\begin{assumption}[Discrete Logarithm Assumption]
	\label{assump:dla}
	Let $\mathbb{G}$ be a cyclic multiplicative group of large prime order $q$, and $g$ and $h$ two elements chosen independently and uniformly at random from $\mathbb{G}$. 
	Then, there exists no probabilistic polynomial time  algorithm $\mathcal{A}$
	that takes as input the tuple $(\mathbb{G}, q, g, h)$ and outputs a value
	$b$ such that $P(g^b = h)$ is significant on the size of
	$q$.
\end{assumption} 
An optimized implementation of Pedersen commitments can be found in  \cite{franck_efficient_2017}.
To generate a common Pedersen scheme in our adversary model, the  generation of the unbiased random input $t$ for $Setup$ can be done as described in Algorithm \ref{alg:genrand} of Appendix \ref{app:supp.rand}. 

\subsection{Zero Knowledge Proofs} 
\label{app:supp.zkp}

\emph{Zero Knowledge Proofs} (ZKP) \cite{goldwasser_knowledge_1989}
allow  to prove statements about private committed values.  
We use  $\Sigma$-protocols \cite{cramer_modular_1997}, a type of ZKP which can prove arithmetic relations, range membership and certain types of logical formulas. We explain basic concepts about ZKP in Section \ref{ch:background.zkp.def} and $\Sigma$-protocols in Section \ref{ch:background.zkp.sigmap}.

\subsubsection{Zero Knowledge Proofs and Arguments} 
\label{ch:background.zkp.def} 
ZKP are a special case of \emph{interactive Proofs of Knowledge} (PoK). We first define PoK.  
For a NP relation $\rel$, a PoK is a
protocol between a prover $\Prov$ and a verifier $\Verif$
in which $\Prov$ tries to prove to $\Verif$ that he knows a \emph{witness} $\wit$
such that $(\stat, \wit) \in  \rel$ for a public statement $\stat$. At the end of the protocol, $\Verif$ either \emph{accepts} or \emph{rejects} the proof.  We denote by  $(\stat;\wit)$ to a member of a relation or an input of a protocol, using a semicolon to separate the public statement $\stat$ from the private witness $\wit$.  
The tuple of all messages in a proof is called the \emph{conversation} or  \emph{transcript}.
A ZKP is a PoK that satisfy the following properties: 
\begin{itemize} 
	\item \textbf{Completeness:} a proof is \emph{complete} if $\Verif$ always accepts the proof when $(\stat;  \wit) \in \rel$ and $\Prov$ knows $\wit$.  
	\item \textbf{Soundness:} a proof is \emph{sound} if any prover 
	whose proof for statement $\stat$ is accepted by the verifier knows a valid witness $\wit$ such that $(\stat; \wit) \in \rel$  with overwhelming probability. The notion of soundness we use is called \emph{witness extended emulation}  \cite{lindell_parallel_2003}. Proofs that are sound only if the prover is computationally
	bounded are also called \emph{arguments}. 
	\item \textbf{Zero Knowledge:} a proof is \emph{zero knowledge} if its transcript reveals no or negligible information about the witness other than its validity.
\end{itemize}

\subsubsection{$\Sigma$-Protocols} 

\label{ch:background.zkp.sigmap}

Here we explain the basics to understand classic $\Sigma$-protocols  proposed by  \cite{cramer_modular_1997,cramer_zero-knowledge_1998}. While the proofs defined therein are applicable for a wide family of commitments, 
we instantiate them for the Pedersen scheme. 

A $\Sigma$-protocol is a 3-message PoK with a transcript of the form $(m_1, t, m_2)$,
where $m_1$ and $m_2$ are messages of $\Prov$, and $t$ is a message of $\Verif$ called \emph{challenge},
which is uniformly distributed in some challenge space $S$. The size of $S$ is linear in $q$ (see Appendix \ref{app:supp.cmt}).

The security of these protocols rely on $\Verif$ behaving honestly in the draw of challenges. 
To circumvent this in settings where no trusted verifier is available, $\Sigma$-protocols
can be transformed to non-interactive using the \emph{Fiat-Shamir heuristic}~\cite{fiat_how_1987,bernhard_how_2012}. This  consists on replacing the
challenge of $\Verif$ by a hash function that computes $t$ from $m_1$ and the statement of the proof. 
Then, a prover can generate a transcript by itself and send it to an untrusted verifier. This 
transformation is proven secure when the hash function is modeled as a random oracle \cite{bellare_random_1993}.
We will use this transformation for both classic and compressed $\Sigma$-protocols described in the following sections. 
However, for simplicity, we will describe protocols in their interactive form. 

Below, we describe some existent protocols, explain their complexity and refer to the literature for other aspects such as security proofs. The protocols are secure under the DLA and instantiated for Pedersen commitments, described in Appendix \ref{app:supp.cmt}.  Recall parameters  $\Z^*_p$, $\Z_\pOrd$, $\mathbb{G}$, $g$ and $h$ defined therein.
For some building blocks, we use
extra  pairs of elements of $\G$ as bases of for $\comschpp$, where bases belonging to the same pair
are always assumed
to be chosen
independently at random from $\mathbb{G}$.
Recall that $p$ is the order of $\Z^*_p$, of which $\mathbb{G}$ is subgroup of. 
Products and exponentiations of members of $\mathbb{G}$ are group operations, which means that they are 
implicitly modulo $p$. 
To analyze the computational cost of the building blocks, we count group exponentiations (\gexp{}), which are the dominant operations.
For a ZKP, we call the \textit{size of a proof} to the size of the transcript. To describe sizes of proofs, we use $\groupelsize$ for the size in bits of an element of $\mathbb{G}$.
In that context, $\groupelsize$ is equal to $\lfloor \log_2 p \rfloor + 1$, but in some 
implementations the size can be different so we abstract from details. $a \getsR S$ means that $a$ is sampled uniformly at random from elements of $S$.

\paragraph{Basic proof of knowledge} This proof is also known as the basic protocol. It allows  $\Prov$ to prove that he knows an opening of a commitment $P$. That is, a value $x$ and randomness $r$ such that $P = g^x  h^r$. This is a proof  for the relation \[\{ (P \in \G; x, r \in \Z_\pOrd): P = g^x h^r \}. \] 
The protocol goes as follows: 
\begin{enumerate} 
	\item $\Prov$ draws $a,b \getsR \Z_q$, computes $A \gets g^a h^b$ and sends $A$ to $\Verif$. 
	\item $\Verif$ draws a challenge $t \gets_R \Z_q$ and sends it to $\Prov$. 
	\item $\Prov$ computes $d \gets a + xt \pmod q$ and $e \gets b + rt \pmod q$ and sends $(d,e)$ to $\Verif$. 
	\item $\Verif$: if $g^d h^e = A P^t$ then \textbf{accept},  else \textbf{reject}
\end{enumerate}

The computational cost of the proof is $2$ \gexp{}
for $\Prov$, and $3$ for $\Verif$. The proof size is of $3\groupelsize$ bits. This protocol was first proposed by \cite{schnorr_efficient_1991} for a single base,  and adapted to many bases (in our case, $g$ and $h$) in \cite{chaum_improved_1988}.

We provide an intuition of how ZKP properties are obtained for the basic proof of knowledge.
Completeness follows directly from the homomorphic property. 
For soundness, consider a prover $\Prov^*$ that, by following the protocol, can produce an accepting 
transcript   $(A, t_1, (d_1,e_1))$ with significant probability. 
Then it is shown that also with non-negligible probability, by using $\Prov^*$'s strategy many times, 
another accepting transcript of the form $(A, t_2, (d_2,e_2))$ can be produced,  where the first message is equal in both transcripts and such that $t_1 \neq t_2$.  
Now, since both transcripts are accepting, we have 
$g^{d_1} h^{e_1} = AP^{t_1}$ and 
$g^{d_2} h^{e_2} = AP^{t_2}$
so $\Prov^*$ can efficiently compute the witness  $x = \frac{d_1 - d_2}{t_1 - t_2}$ and $r = \frac{e_1 - e_2}{t_1 - t_2}$ which is valid as $g^x h^r = P$.  Either $\Prov^*$ already knew it or he can efficiently compute the 
discrete logarithms base $g$ of $h$, which is a contradiction due to the DLA. 

Zero knowledge is obtained by showing that all information
seen in the proof is randomness that does not depend on the secrets. 
By only knowing the statement $P$ one can sample 
$d'\getsR \Z_\pOrd$, $e'\getsR \Z_\pOrd$ and $t' \getsR \Z_\pOrd$, compute $A' = P^{-t'} g^{d'} h^{e'}$
and generate the transcript $(A', t', (d',e'))$
has the same distribution as a conversation between an honest prover and 
an honest verifier. Note that, as it is computed in reverse,  $(A', t', (d',e'))$
cannot be efficiently produced by a dishonest prover in the actual protocol, as the order of messages
cannot be changed.

In the following, we will present $\Sigma$-protocols that use similar arguments as the ones described for the basic proof of knowledge to prove their properties. 
In particular, they can all generate accepting transcripts ``in reverse'' even if witnesses are not known, as shown
to prove the zero knowledge property above. 
We will not describe their security proofs, but we will point to the literature where these can be found.

\begin{paragraph}{Proof of equality} This  proof allows  $\Prov$ to
	prove that he committed to the same private value $x$ in  two different
	commitments $P = g_1^{x}  h_1^{r}$
	and $P' = g_2^{x}  h_2^{r'}$. That is, the relation 
	\[ \{ (P, P' \in \G; x, r,r' \in \Z_\pOrd): P = g_1^{x}  h_1^{r} \land P' = g_2^{x}  h_2^{r'} \}.\] 
	The pairs of bases $(g_1, h_1)$ and $(g_2, h_2)$ might be different. This proof is described by \cite{chaum_wallet_1993} for a single base and a generalization for many bases can be found in \cite{chaum_improved_1988}. 
	The protocol goes as follows: 
	\begin{enumerate}
		\item $\Prov$ generates $a, b,c \getsR \Z_q$, computes $A \gets g_1^a h_1^b$ and $B \gets g_2^a h_2^c$ and sends $(A,B)$ to $\Verif$.
		\item $\Verif$ draws a challenge $t \getsR \Z_q$ and sends it to $\Prov$. 
		\item $\Prov$ computes $d \gets a + xt \pmod q$, $e \gets b + rt \pmod q$ and $f \gets c + r't \pmod q$ and sends $(d,e,f)$ to $\Verif$. 
		\item $\Verif$: if $g_1^d h_1^e = A P^t$ and  $g_2^d h_2^f = B (P')^t$, then \textbf{accept} else \textbf{reject}.
	\end{enumerate} 
	The proof requires the computation of $4$ \gexp{}
	for $\Prov$ and $6$
	for $\Verif$. The proof size is of $5\groupelsize$ bits. 
\end{paragraph}

\paragraph{Composition}
Let $\prot_1$ and $\prot_2$ be two $\Sigma$-protocols for ZKP of relations $\rel_1$ and $\rel_2$ respectively. 
A single $\Sigma$-protocol to prove the relation 
\[ 
\rel_\land = \{ (\stat, \stat'; \wit, \wit'): (\stat;\wit) \in \rel_1 \land (\stat'; \wit') \in \rel_2 \}
\]
can be easily constructed. Basically, we just make both proofs to share the same challenge. The protocol is as follows.  The first message is $(m_1, m_1')$, where $m_1$ and $m_1'$ are exactly as generated for the first messages of  $\prot_1$ and $\prot_2$ respectively. Then $\Verif$ generates the challenge $t$ and sends it to $\Prov$. Next, $\Prov$ generates the third
message $(m_2,m_2')$ where $m_2$ and $m_2'$ are generated from $(m_1, t)$ and from $(m_1', t)$ as they will be generated
for the third messages of $\prot_1$ and $\prot_2$ respectively. Finally $\Verif$ verifies $(m_1, t, m_2)$ and $(m_1', t, m_2')$ are accepting 
as in $\prot_1$ and $\prot_2$. 
Using this composition, if we wish to prove statements over many relations, we can do it with a single  $\Sigma$-protocol. 

Another technique allows to construct ZKPs of disjunctions of statements. 
The work by \cite{cramer_proofs_1994} proposed a $\Sigma$-protocol for the relation 
\[
\rel_\lor = \{  (\stat, \stat'; \wit, \wit'): (\stat;\wit) \in \rel_1 \lor (\stat'; \wit') \in \rel_2 \},
\]
where $\prot_1, \prot_2, \rel_1$ and $\rel_2$ are defined as for conjunctions. We briefly describe the protocol below. Without loss of generality, assume that $\Prov$ knows 
$(\stat;\wit) \in \rel_1$, but he does not know $(\stat';\wit') \in \rel_2$ and that $\Verif$ does not know which of the two
witnesses is known by $\Prov$. As $\prot_2$ is zero knowledge, an accepting transcript $(m_1', t_b, m_2')$ 
can be generated in reverse as shown for the proof of knowledge above and  without $\Prov$ knowing a valid witness. 
\begin{enumerate}  
	\item $\Prov$ generates $(m_1', t_b, m_2')$ in reverse and generates $m_1$ as he would normally do for the first message of $\prot_1$.
	\item $\Prov$ sends  $(m_1, m_1')$ to $\Verif$.
	\item $\Verif$ draws a challenge $t \getsR \Z_q$ and sends it to $\Prov$.
	\item $\Prov$ computes $t_a \gets t - t_b \bmod \pOrd$ and computes $m_2$ from $m_1$ and $t_a$ as he would have done in the original protocol $\prot_1$.
	\item $\Prov$ sends $(m_2, m_2')$ to $\Verif$.
	\item $\Verif$ checks that $(m_1, t_a, m_2)$ and $(m_1', t_b, m_2')$ are accepting as in protocols $\prot_1$ and $\prot_2$, and that $t_a + t_b \bmod \pOrd$ is equal to $t$.  If all checks pass, then \textbf{accept}. Otherwise \textbf{reject}. 
\end{enumerate}  
Here, $\Prov$ can generate  $(m_1, t_a, m_2)$ as he knows a  witness for $\rel_1$ and he can simulate 
an accepting transcript for $\rel_2$. As $t_a + t_b \bmod \pOrd= t$, $\Verif$ is ensured that at least
either $t_a$ or $t_b$ is unpredictable for $\Prov$ when generating the first message of the proof, so he would not be able to pass the proof with significant
probability unless he knows at least one witness. We refer to the paper for the security proofs and further details. 

The composition techniques outlined above allow one to prove logical formulas of conjunctions and disjunctions of statements. 
The computational and communication cost for the resulting protocol is similar to the sum of the costs
of the proof of each composed statement. 


%
\begin{paragraph}{Proof of linear relation}
	Let $\vv{x} = (x_i)_{i = 1}^k$ for some $k > 0$ be a vector of private values,  $\vv{P} = (P_i)_{i=1}^k$ a public vector of commitments such that $P_i$ is a commitment of $x_i$, $\vv{a} = (a_i)_{i=1}^k$ a public vector of integer  coefficients, and $b$ a public scalar. The goal of $\Prov$ is to prove the equality $\langle \vv{x}, \vv{a} \rangle = b \pmod q$, where $\langle \cdot, \cdot \rangle$ is the inner product.  The proof follows almost directly from the homomorphic property of commitments. We have that $P_L = \prod_{i = 1}^k P_i^{a_i}$ is a valid commitment of $\langle \vv{x} , \vv{a} \rangle$. Then, $\Prov$ and $\Verif$ prove that values committed with $P_L$ and $P_b = g^b$ are equal. Note that it is not necessary to use the base $h$ to compute the commitment  $P_b$  as $b$ is publicly known and hence we can use randomness equal to $0$.  The proven relation is 
	\begin{align*}
		\Big\{(\vv{P} \in \G^k, \vv{a} \in \Z_q^k, b \in \Z_q ; \vv{x} , \vv{r} = (r_i)_{i=1}^k \in \G^k )   \\
		: \bigwedge_{i = 1}^k P_i = g^{x_i} h^{r_i} \land \langle \vv{a}, \vv{x} \rangle = b \pmod q \Big\}.
	\end{align*}
	where $\vv{r}$ is the randomness for each commitment. 
	For $\Prov$ and $\Verif$, the proof requires $k + 1$ \gexp{} to compute  $P_L$ and $P_b$ plus one proof of equality. This leads to a final cost  of $k+5$ \gexp{} 
	for $\Prov$ and $k+7$ 
	for $\Verif$. Considering that initial commitments $\vv{P}$ are already shared with $\Verif$ and that $\Prov$ has already proven knowledge of them, the proof size is equal to 
	that of the proof of equality, which is $5\groupelsize$ bits.
\end{paragraph}

\begin{paragraph}{Proof of a committed bit}
	A handy special case of the composition using disjunction is a proof that a secret $b$  is a bit, i.e.
	that $\{b=0 \lor b=1 \}$. This instantiation is described by (\cite{mao_guaranteed_1998} Section 3.2
	therein). It requires $3$ \gexp{} 
	for $\Prov$ and $4$ 
	for $\Verif$. The size of the proof is of $4\groupelsize$ bits. 
\end{paragraph} 

\begin{paragraph}{Range proof}
	Here, the statement to prove is that, for a positive integer $M < q-1$ of $B_M$ bits and a
	commitment $P$, $\Prov$ knows an opening $x$ of  $P$ 
	that lies in the range $[0, M]$. This is a folklore proof that can be done by committing to the bits $b_1,  \ldots, b_{B_M }$ 
	of the bit representation of $x$, i.e. that $x = \sum_{i=1}^{B_M} 2^{i-1} b_i$. For $i \in  \{1, \dots,B_M\}$ it is proven that  $b_i$ is a bit with the protocol mentioned above. The validity of the decomposition of $x$ can be proven by a linear proof. This implies that $x \in [0, 2^{B_M}-1]$. For ranges whose upper bound is not of the form $2^{B_M}-1$, $\Prov$ commits to $M - x$ and proves that it also lies in $[0, 2^{B_M}-1]$. This finishes the proof for arbitrary ranges.  

	Committing to every bit $b_i$ has a cost of $B_M$ \gexp{}  
	for $\Prov$. The cost of proofs of bits is dominated by $3B_M$ \gexp{}
	for $\Prov$ and $4B_M$ 
	for $\Verif$. The proof that 
	these bits are the decomposition of $x$ is dominated by $B_M$ \gexp{}  
	for both $\Prov$ and $\Verif$. For arbitrary ranges, the cost duplicates as described above.
	The total cost of the proof is dominated by $10 \log_2(M)$ \gexp{} 
	for $\Prov$ and the same cost  
	for $\Verif$. The size of the proof is dominated by $10\log_2(M)\groupelsize$ bits. 
\end{paragraph}

\begin{paragraph}{Proof of product}
	In this proof, $\Prov$ wants to prove the knowledge of openings $a$, $b$ and $d$
	of commitments $P_a = g^{a} h^{r_1}$, $P_b = g^b h^{r_2}$ and
	$P_d = g^d h^{r_3}$ and that they satisfy $a b = d \pmod q$. This proof can
	be done in a straightforward way with the knowledge and equality proofs. We use the fact that, by
	properties of our cyclic group $\mathbb{G}$,  $P_a$ is not only a commitment but also a base, and that $P_d = (P_a)^b h^{(r_3 - br_1)}$ is a valid commitment for $b$ using bases $(P_a, h)$. For the proof, $\Prov$ first proves the knowledge of an opening of $P_a$, and then it proves the equality of two openings $P_b$ and $P_d$ with the proof of equality, where commitments are computed with the pairs of bases $(g, h)$ and $(P_a, h)$ respectively. If $\Prov$ does not know the commitment of a product of $ab$ hidden in $P_d$ it would not pass the proof. 
	The proof is the instantiation for Pedersen commitments of the protocol by  \cite{cramer_zero-knowledge_1998}, and requires  $6$ \gexp{} 
	for $\Prov$ and $9$ 
	for $\Verif$. The proof size is of $8\groupelsize$ bits. 
\end{paragraph}


\begin{paragraph}{Proof of modular sum}
	Here, $\Prov$ wants to prove that, for a public modulus $M < q/2$, a public value $t \in 
	[0, M-1]$ and two secrets $x, z \in [0,M-1]$, the relation $x = z + t \mod M$ is
	satisfied. Let $P_x$ and $P_z$ be the commitments of $x$ and
	$z$ respectively. First, $\Prov$ does two range proofs to prove that $x$ and $z$ lie in $[0,M-1]$. 
	Then it proves that an auxiliary value $b$ is a bit. Finally it proves that ${x = z + t - bM \pmod q}$ with a linear proof.
	Because of their range, we have that $z + t < q$, hence the modulus $q$ does
	not interfere and the linear equality holds if and only if the modular sum is satisfied. 
	This proof is inspired by the more general proof of modular sum proposed by 
	\cite{camenisch_proving_1999}. 
	Its complexity is given by the composition of the proofs mentioned above,
	where the dominant term comes from  the two range proofs, and amounts to $20 \log_2(M)$ \gexp{} 
	for $\Prov$ and the same
	for $\Verif$. The size of the proof is of $20\log_2(M)\groupelsize$ bits. 
\end{paragraph}

\subsection{Public Randomness}

\label{app:supp.rand}

We provide here an algorithm to generate public randomness which, compared to the sketch provided in Section \ref{app:crypto.setup}, distributes the cost more evenly over the participants. In particular, we provide a decentralized method to generate $n$ public random numbers in $\Z_q$ with a computational effort of $O(1)$ per user, and costs logarithmic in $n$ for a negligible portion of users. In \gopa{} it is used among others to generate private random seeds and to initialize the Pedersen commitment parameters $\comschpp$ shared by all users.

In our procedure, we enumerate users from $1$ to $n = |U|$ and use the cryptographic hash function $H$ described at the beginning of the Section to generate random numbers  over $\Z_{2^T}$. We also make use of a commitment function $Com$, for which a common trusted initialization is not required (see for example \cite{blum_coin_1983}). The procedure is depicted in Algorithm \ref{alg:genrand}.

\begin{algorithm*}[hb]
	\floatname{algorithm}{Algorithm}
	\caption{Generation of Public Randomness}
	\begin{algorithmic}[1]
		\STATE{\textbf{Input:} Hash function $H$ and a commitment function $Com$. }
		\FORALL{\label{algl:rnd_commitreveal.start}$u \in  \{1, \dots, n\}$ }
		
		\STATE{Draw  ${s_u \getsR \Z_q}$, compute ${\mathbf{c_u} \gets Com(s_u)}$ and publish $\mathbf{c_u}$  }  
		\ENDFOR 
		\FORALL{$u \in  \{1, \dots, n\}$ }
		\STATE{\label{algl:rnd_commitreveal.end} Set $s[u,u] \gets s_u$ and publish it }
		\ENDFOR 
		\FOR{\label{algl:rnd_distribsum.start}$j = 1$ \TO  $\lfloor \log_2(n) \rfloor  +1$ sequentially } 
		\FORALL{$i \in \{0, \dots, \lfloor n/2^j \rfloor \}$ }
		\STATE{	Let $u_{min} = 2^{j}i+1$ and $u_{max} = \min(2^j (i+1), n)$} 
		\IF{$s[u_{min} , u_{max}]$ is not already published} 
		\STATE{Let $u_{mid}=  u_{min}  + 2^{j-1}-1$}
		\STATE{A user $u \in \{ u_{min}, \dots, u_{max} \}$ wakes up and:}
		\STATE{$\bullet$ queries $(s[u_{min} , u_{mid}]$,   $s[u_{mid}+1,  u_{max} ])$, if a value is not published, set it to $0$} 
		\STATE{\label{algl:rnd_distribsum.end} $\bullet$ publishes ${s[u_{min} ,  u_{max}] \gets s[u_{min} , u_{mid} ] + [u_{mid}  + 1,   u_{max} ]} \mod q$ }
		\ENDIF 
		\ENDFOR
		\ENDFOR
		\FORALL{\label{algl:rnd_hash.start} $u \in \{1, \dots, n\}$ } 
		\STATE{\label{algl:rnd_hash.end}  Query $S \gets s[1,n]$ and publish $t_u \gets H(S + u)$}
		\ENDFOR
		\STATE{ \textbf{Output:} A set  of unbiased random numbers $t_1, \dots, t_n \in \Z_{2^T}$} 
	\end{algorithmic}
	\label{alg:genrand}
\end{algorithm*}

The ``commit-then-reveal'' protocol from in lines  \ref{algl:rnd_commitreveal.start}-\ref{algl:rnd_commitreveal.end}  is to avoid users from choosing the value $s_u$  depending
on the choice of other users, which could bias the final seed $S = \sum_{u \in U} s_u \mod q$. The value $S$, computed in lines \ref{algl:rnd_distribsum.start}-\ref{algl:rnd_distribsum.end} in a distributed way, requires at most $O(\log(n))$ queries and sums for a user in the worst case, but $O(1)$ for almost all users. Inactive users do not affect the computation, but their $s_u$ terms might be ignored.  As it is computed from modular sums, $S$ is uniformly distributed over $\Z_{q}$ if at least one active user is honest. Finally, we compute our public random numbers $t_1, \dots, t_n$ in lines \ref{algl:rnd_hash.start} and \ref{algl:rnd_hash.end}, which are  unpredictable due to the properties of $H$ and the unpredictability and uniqueness of its input $S + u$.

To verify that a user $u$ participated correctly, one needs to check
that (1) the commitment $\mathbf{c_u}$ was properly computed from $s_u$, (2) all sums $s[u_{min}, u_{max}]$ that $u$ published were computed correctly, and (3)
the challenge $t_u$ was correctly computed from $S + u$ by the application
of $H$. 

The cost of the protocol for a user is dominated in the worst case at by $\log_2(n)$ sums  
and $3\log_2(n)\groupelsize$ bits in total, corresponding to the several queries of  $s[u_{min}, u_{mid}]$ and $s[u_{mid}+1, u_{max}]$ and the publication of its sum.  However, this is cost applies to only $O(1)$ users, while the rest of the users computes one commitment, one evaluation of $H$ and $O(1)$ sums and transfers $O(1)$ bits during the execution. 

\subsection{Private Gaussian Sampling}
\label{app:supp.gauss} 

In our algorithm, every user $u$ needs to generate a Gaussian distributed number $\eta_u$, which he does not publish, but for which we need to verify that it is generated correctly, as otherwise a malicious user could bias the result of the algorithm.

Proving the generation of a Gaussian distributed random number is more
involved than ZKPs such as linear relationships and range proofs we need to verify in the other parts of the computation.

Recall that $\psi$ is the desired fixed precision and $B$ is a precision parameter such that $B^2/\psi$ is an integer.  We want to draw $\eta_u$ from the Gaussian distribution approximated with $1/B$ equiprobable bins and we want to prove the correct drawing up to a precision $B$.

We will start from the private seed $\pseed_u$, generated in the  Phase \ref{enum:verif.phase.setup} Setup (Section \ref{app:crypto.setup}), which is only known to user $u$, for which $u$
has published a commitment $Com_{\comschpp}(\pseed_u)$ and for which the other
users know it has been generated uniformly randomly from an interval $[0,M-1]$ for $M \approx 1/B$.
There are many ways now to exploit a uniformly distributed number to generate a Gaussian distributed one, but studying and comparing their efficiency and numerical quality is out of the scope of the current paper.  Here, we only describe one strategy.

User $u$ will compute $x'$ such that $((2\pseed_u+1)/M)-1=\hbox{erf}(x'/\sqrt{2})$.  We know that $x'$ is normally distributed.  The main task is then to provide a ZKP that $y=\hbox{erf}(x)$ for $y=((2\pseed_u+1)/M)-1$ and $x=x'/\sqrt{2}$. As $\erf$ is symmetric, we do our analysis for positive values of $x$ and $y$,  while the extension for the negative case is straightforward. We want to achieve an approximation where the error on $y$ as a function of $x$ is at most $B$.

\paragraph{Approximating the error function.}
The error function relates its input and output in a way that cannot be
expressed with additive, multiplicative or exponential equations.  We
therefore approximate $\hbox{erf}$ using a converging series.  In particular, we will rely on the series
\begin{equation}
	\label{eq:erf.series}
	\hbox{erf}(x) = \frac{2}{\sqrt{\pi}} \sum_{l=0}^\infty \frac{(-1)^l x^{2l+1}}{l!(2l+1)}.
\end{equation}
As argued by \cite{chevillard:inria-00261360}, this series has two major
advantages. First, it only involves additions and multiplications, while other
known series converging to $\hbox{erf}(x)$ often include multiple $\exp
(-x^2/2)$ factors which would require additional evaluations and proofs.
Second, it is an alternating series, which means we can determine more easily
in advance how many terms we need to evaluate to achieve a given precision.  

Nevertheless, Equation \eqref{eq:erf.series} converges slowly for large $x$.
It is more efficient to prove either that 
\[
y=\hbox{erf}(x) \quad or \quad 1-y = \hbox{erfc}(x),
\]
as for $\hbox{erfc}(x)=1-\hbox{erf}(x)$ there exist good approximations
requiring only a few terms for large $x$. An example is the asymptotic
expansion
\begin{equation}
	\label{eq:erfc.series1}
	\hbox{erfc}(x) = \frac{e^{-x^2}}{x\sqrt{\pi}} S_{\mathrm{erfc}}(x) + R_L(x),
\end{equation}
where 
\begin{equation}
	\label{eq:erfc.series2}
	S_{\mathrm{erfc}}(x) =  \sum_{l=0}^{L-1} \frac{(-1)^l (2l-1)!!}{(2x^2)^l}
\end{equation}
with $l!!=1$ for $l<1$ and $(2l-1)!! = \prod_{i=1}^{l} (2i-1)$.
This series diverges, but if $x$ is sufficiently large then the remainder
\begin{equation}
	\label{eq:erfc.remainder}
	R_L(x) \leq \frac{e^{-x^2}}{x\sqrt{\pi}} \frac{(2L-1)!!}{(2x^2)^L} 
\end{equation}
after the first $L$ terms is sufficiently small to be neglected. 
So $u$ could prove either part of the disjunction 
depending on whether the $\hbox{erf}$ or $\hbox{erfc}$ approximations achieve sufficient precision.

\paragraph{Zero Knowledge Proof of $\erf(x)$.} 
We consider use fixed-precision rounding operations. The implied rounding
does not cause major problems for several reasons.  First, the Gaussian
distribution is symmetric, and hence the probability of rounding up and
rounding down is exactly the same, making the rounding error a zero-mean
random variable. Second, discrete approximations of the Gaussian mechanism
such as binomial mechanisms have been studied and found to give similar guarantees as the Gaussian mechanism \cite{cp-sgd}. Third, we can require the cumulated rounding error to be an order of magnitude
smaller than the standard deviation of the noise we are generating, so that
any deviation due to rounding has negligible impact.
\newcommand{\gdenc}[1]{\langle #1 \rangle}
We will use a fixed precision for all numbers (except for small integer constants), and we will represent numbers as multiples of $\gdprec$.

Let $t_l = \frac{x^{2l+1}}{l!}$, and $\Lerf+1$ be the amount of terms of  the series in Equation \eqref{eq:erf.series} we evaluate. We provide a ZKP of the evaluation of $\erf$ by proving that $t_0 = x$, 
\begin{align} 
	t_{l} &= \frac{t_{l-1}  x^2}{l} \text{\quad} \text{for all } l \in \{1, \dots, \Lerf \}, \label{eq:erf.zkp}
\end{align} 
and $y = \frac{2}{\sqrt{\pi}} \sum_{l=0}^{\Lerf} (-1)^l  \frac{t_l}{2l+1}$. 
The bulk of the ZKP is in Equation \eqref{eq:erf.zkp} and in the divisions by $2l+1$ of the latter relation. For any fixed-precision value $d$, let $\gdenc{d} = \frac{d}{\gdprec}$ be its integer encoding.
We can achieve a ZKP of the fixed precision product $c = ab$ for private $a$, $b$ and $c$  by proving that $\frac{1}{\gdprec}\gdenc{c} - \gdenc{a}\gdenc{b} \in [- 1/2\gdprec, 1/2\gdprec ] $. 
For round-off division, a ZKP that $a/b = c$ for private $a, c$  and a public positive integer $b$ is possible by proving that 
$\gdenc{a} - b\gdenc{c} \in [ -b/2,b/2]$. The above proofs can be achieved using circuit  and range proofs in $\Z_q$. 

Similarly for $\erfc$, let  $m_l = (2l-1)!!/(2x)^l$ and $\Lerfc$ be the amount of terms we compute of the series defined in Equation \eqref{eq:erfc.series2}, we can construct a ZKP of its evaluation by proving $m_0 = 0$, 
\begin{equation} 
	m_l = m_{l-1} \frac{2l-1}{2x^2} \text{\qquad, for all } l \in \{1, \dots, \Lerfc-1 \}
	\label{eq:erfc.zkp} 
\end{equation}  
and $y =  \frac{e^{-x^2}}{x\sqrt{\pi}}\sum_{l=0}^{\Lerfc-1} (-1)^l m_{l}.$  Proving that $\enc{m_{l}}\enc{x^2} - \frac{2l-1}{\gdprec}\enc{m_{l-1}}  \in \left[  -\enc{x^2} , \enc{x^2} \right]$ is equivalent to prove the relation in  Equation \eqref{eq:erfc.zkp}. 
This can be done with $10 B_{x^2}$ cost, where $B_{x^2}$ is maximum number of bits of $\enc{x^2}$. We can assume that $B_{x^2} \leq \log_2(q)$ (where $q$ is the order of the Pedersen group $\G$) which makes the cost for each term not bigger than $10\log_2(q)$. All other non-dominant computations can similarly be proven with circuit proofs. 

To approximate the $\exp(-x^2)$ term, the
common series $\exp(z)=\sum_{i=0}^\infty z^i/i!$ is known to converge quickly.
Even if its terms first go up until $z<i$, for larger $z$ where this could
slow down convergence one can simply divide $z$ by a constant $a$ (maybe
conveniently a power of $2$), approximate $\exp(z/a)$ and then compute $(\exp
(z/a))^a$, which can be done efficiently.  For simplicity, we omit the details here.

\newcommand{\Eerfc}{E_{\mathrm{erfc}}} 

\paragraph{Amount of approximation terms.} 

Now we determine the magnitudes of $\Lerf$ and $\Lerfc$ needed to achieve an error smaller than $B$ in our computation. We then require the approximation and rounding error to be smaller than  $B/2$. 
The amount of terms of the series must not depend on $x$ as the latter is a private value, but we must be able to achieve the expected precision for all possible $x$. 

The $\erf$ series requires more approximation terms as $x$ gets larger. On the other hand, the $\erfc$ series is optimal when $\Lerfc = \lfloor x^2 + 1/2 \rfloor$ terms are evaluated, and the error gets smaller as $x$ increases. Then, we use $\erfc$ only when $x$ is large enough to satisfy the required precision, and use $\erf$ for smaller values. We first compute the lower bound $\xminerfc$ 
for the application of $\erfc$. Then the domain of $\erf$ is restricted to $\left[ 0, \xminerfc \right)$ so can obtain the an upper bound of $\Lerf$. Similarly, the restricted domain of $\erfc$ allows us to upper bound $\Lerfc$. 

Developing Equation \eqref{eq:erfc.remainder}, we can achieve an 
approximation error of $\erfc$ of 
\begin{eqnarray*} 
	\Eerfc(x) &\leq& \frac{1}{e^{x^2}x\sqrt{\pi}} \frac{(2L-1)!!}{(2x^2)^L} 
	= \frac{1}{e^{x^2}x\sqrt{\pi}} \frac{(2L)!}{L!2^L(2x^2)^L} \\
	&\approx& \frac{1}{e^{x^2}x\sqrt{\pi}} \frac{\sqrt{4\pi L}(2L/e)^{2L}}{\sqrt{2\pi L} (L/e)^L 2^L (2x^2)^L} 
	= \frac{\sqrt{2}}{e^{x^2}x\sqrt{\pi}} \frac{(2^{2L}) (L^{2L}) (e^L)}{(e^{2L}) (L^L) (2^L) (2x^2)^L} \\
	&=& \frac{\sqrt{2}}{e^{x^2}x\sqrt{\pi}} \frac{2^{L} L^{L}}{e^{L} (2x^2)^L} 
	= \frac{\sqrt{2}}{e^{x^2}x\sqrt{\pi}} \frac{L^{L}}{e^{L} (x^2)^L}. 
\end{eqnarray*} 
This is minimal in $\lfloor x^2 + 1/2 \rfloor$ 
so given $x$ we can achieve an error of 
\begin{eqnarray*} 
	\frac{\sqrt{2}}{e^{x^2}x\sqrt{\pi}} \frac{L^{L}}{e^{L} (x^2)^L} &\le& \frac{\sqrt{2}}{e^{x^2}x\sqrt{\pi}} \frac{{\lfloor x^2 + 1/2 \rfloor}^{\lfloor x^2 + 1/2 \rfloor}}{ e^{\lfloor x^2 + 1/2 \rfloor} (x^2)^{\lfloor x^2 + 1/2 \rfloor}} 
	\leq \frac{\sqrt{2}}{\sqrt{\pi}} \frac{(1 + 1/2x^2)^{\lfloor x^2 + 1/2 \rfloor}}{x e^{(2x^2 - 1/2 )}}.
\end{eqnarray*}
As we use $\erfc$ for large values of $x$, we assume $x \geq 1$ which will not affect our reasoning,  then and we can simplify the above by
\begin{eqnarray} 
	\Eerfc(x) &\leq&	\frac{\sqrt{2}}{\sqrt{\pi}} \frac{(1 + 1/2x^2)^{\lfloor x^2 + 1/2 \rfloor}}{x e^{(2x^2 - 1/2 )}} 
	\leq \frac{\sqrt{2}}{\sqrt{\pi}} \frac{3\sqrt{6}/4}{x e^{(2x^2 + 3/2 )}} 
	\leq \frac{1}{2} \sqrt{\frac{27}{\pi}} \frac{1}{ e^{(2x^2 - 1/2 )}}. 
	\label{eq:erfc.err}
\end{eqnarray}
Then, by Equation \eqref{eq:erfc.err} and if 
\[
\frac{B}{2} \ge 
\frac{1}{2} \sqrt{\frac{27}{\pi}} \frac{1}{ e^{(2x^2 - 1/2 )}}  \ge E_{\mathrm{erfc}}(x),
\]
the prover will use the $\hbox{erfc}$ series, else the $\hbox{erf}$ series.  In  the latter case, we require that 
\[
\sqrt{\frac{27}{\pi}} \frac{1}{ 2e^{(2x^2 - 1/2 )}} \geq \frac{B}{2},
\] 
which implies 
\[ 
\sqrt{\frac{27}{\pi}} \frac{1}{B} \geq e^{(2x^2 - 1/2 )}. 
\] 
The above is equivalent to 
\[
\ln(27/\pi)/2 +  \ln(1/B) \geq 2x^2 - 1/2, 
\] 
which implies
\[  
1.076 +  \ln(1/B)/0.434 \geq (2x^2 - 1/2 ) 
\] 
and
\begin{equation} 
	x_{\mathrm{erfc}}^{min}  = \sqrt{ \frac{\ln(1/B)}{2} + 0.788 } \geq x.
	\label{eq:erf.maxx}
\end{equation}

Now that we bounded the domains of $\erf$ and $\erfc$, we can obtain the number of terms to evaluate for the series. 
As shown in Equation \eqref{eq:erf.series}, this is an alternating series too, so
to reach an error smaller than $B/2$, it is sufficient to truncate the series
when terms get smaller than $B/2$ in absolute value.  In particular, we need $L$ terms with
\[
\frac{2x^{2L+1}}{\sqrt{\pi}L!(2L+1)} \le \frac{B}{2}.
\]
Using again Stirling's approximation, this means 
\[
\frac{2x^{2L+1}}{\sqrt{\pi}\sqrt{2\pi L}(L/e)^L (2L+1)} \le \frac{B}{2}.
\]
Taking logarithms, we get
\[
\ln(2)+(2L+1)\ln(x) - \ln(\pi\sqrt{2}) - (L+1/2)\ln(L)+L-\ln(2L+1) \le \ln(B)-\ln(2). 
\]
We have that ${2\ln(2) - \ln(\pi\sqrt{2})-\ln(2L+1) \le 0 }$, so the above inequality is satisfied if
\[
(2L+1)\ln(x) - (L+1/2)\ln(L) + L \leq \ln(B) 
\]
which is equivalent to 
\[
(L+1/2) \ln(ex^2/L) \leq \ln(B),
\]
or written differently: 
\[
(L+1/2) \ln\Bigg(1 - \frac{L - ex^2}{L}\Bigg) \leq \ln(B)
\]
As $\ln(1+\alpha) \leq \alpha$, this is satisfied if
\[
-(L+1/2) \frac{L - ex^2}{L} \leq \ln(B),
\]
which is equivalent to 
\[ 
(L+1/2)\frac{L - ex^2}{L} \geq \ln(1/B). 
\]
The above is satisfied if
\[L - ex^2 \geq \ln(1/B). \]
It follows that we need
\[ 
L \geq \ln(1/B) + ex^2. 
\]
Substituting the worst case value $x_{\mathrm{erfc}}^{min}$ of $x$ from Equation \eqref{eq:erf.maxx},
we get

\[
L \geq \Big\lceil  \Big(1+ \frac{e}{2}\Big) \ln(1/B) + 0.79e \Big\rceil 
\]
which, approximating, is implied if 
\begin{equation}
	\label{eq:erf.minL}
	L \geq  \lceil 2.36\ln(1/B) +  2.15 \rceil = L_{\mathrm{erf}}. 
\end{equation}
Now, to compute $\Lerfc$,  we just observe in Equation \eqref{eq:erfc.err} that the error gets smaller as $x$ increases, 
so the biggest error for a fixed number of terms of the series is in $\xminerfc$. Therefore, plugging Equation \eqref{eq:erf.maxx} we have 
\begin{align} 
	\Lerfc &= \lfloor  (\xminerfc)^2 +1/2\rfloor \notag \\
	&= \left\lfloor  \frac{\ln(1/B)}{2} + 0.788 +\frac{1}{2} \right\rfloor\notag \\
	&= \left\lfloor  \frac{\ln(1/B)}{2} + 1.288 \right\rfloor. \label{eq:erfc.minL}
\end{align} 

\paragraph{Required precision $\gdprec$.} Now we determine the storage needed in our fixed precision representation such that the rounding error is smaller than $B/2$. We have to consider error propagation and the errors $\err{div}$ and $\err{mul}$ due to division and product ZKPs respectively, which are equal to $\gdprec/2$. Let $\err{l}$ be the error to compute the $t_l$ terms of the $\erf$ series. The total $\erf$ rounding error is then 
\[
\Eerfround =    \frac{2}{\sqrt{\pi}} \left(\sum_{l=0}^{\Lerf} \frac{\err{l}}{2l+1} + \err{div} \right) + \err{mul}.
\]
We require $1/\gdprec M^2$ to be an integer so as the variable $x$ is a multiple of $1/M$, we can represent $x$ and $x^2$ without rounding. As $t_0 = x$ we have that $\err{0} = 0$, and for other terms we have 
\begin{align*}
	t_{l+1} \pm \err{l+1} &= \frac{(t_l \pm \err{l}) x^2  \pm \err{mul}}{l+1} \pm \err{div}  \\ 
	&= \frac{t_l x^2 \pm \err{l} x^2  \pm \gdprec/2 }{l+1} \pm \frac{\gdprec}{2} \\ 
	&= \frac{t_l x^2}{l+1} \pm \left( \frac{\err{l} x^2  + \gdprec/2 }{l+1} + \frac{\gdprec}{2} \right).
\end{align*}  
Then the absolute error at term $l+1$ is 
\[ \err{l+1} = \frac{\err{l} x^2  + \gdprec/2 }{l+1} + \frac{\gdprec}{2} \leq \frac{\err{l}x^2}{l+1} + \gdprec. \] For  $1 \le l \le x^2 -1$ we have $\err{l} \geq \gdprec$ and 
\[ 
\err{l+1} = \err{l} \frac{x^2}{l+1} + \gdprec  \leq 2 \err{l} \frac{x^2}{l+1}. 
\] 
If $l+1 \le x^2$, we can easily see that
\begin{equation}  
	\err{l} \leq \gdprec 2^{l-1} \frac{x^{2(l-1)}}{l!}.
	\label{eq:erf.round.El} 
\end{equation} 
If $l+1 > x^2$, we have that 
\begin{align*} 
	\err{l+1} &= \err{l} \frac{x^2}{l+1} + \gdprec \leq E_l + \gdprec.
\end{align*} 
From the above and plugging Equation \eqref{eq:erf.round.El} we have 
\begin{align*} 
	E_l &\leq E_{\lfloor x^2 \rfloor} +  (l - \lfloor x^2 \rfloor) \gdprec \\
	&\leq \gdprec 2^{\lfloor x^2 \rfloor} \frac{x^{2\lfloor x^2 \rfloor}}{\lfloor x^2 \rfloor !}  + (l -\lfloor x^2 \rfloor) \gdprec. 
\end{align*} 
From the above, independently of $x^2$ and $l$ we have that 
\begin{align*} 
	\err{l} &\leq  \gdprec 2^{\lfloor x^2 \rfloor } \frac{x^{2\lfloor x^2 \rfloor }}{\lfloor x^2 \rfloor !}  +  \sum_{ k= \lfloor x^2 \rfloor +1 }^l\gdprec \\
	&=  \gdprec \frac{(2x^2)^{\lfloor x^2 \rfloor}}{\lfloor x^2 \rfloor !} +   \sum_{ k= \lfloor x^2 \rfloor +1 }^l\gdprec. 
\end{align*} 
We assume $x \geq 1$ as $\Eerfround$ is smaller when $x \in [0, 1)$. Using the Stirling approximation we can bound $\err{l}$ to 
\begin{align*} 
	\err{l}  &\leq \gdprec \frac{(2x^2)^{\lfloor x^2 \rfloor }e^{\lfloor x^2 \rfloor}}{\sqrt{2\pi \lfloor x^2 \rfloor} \lfloor x^2 \rfloor^{\lfloor x^2 \rfloor} }  + \sum_{k= \lfloor x^2 \rfloor +1 }^l \gdprec \\
	&\leq \gdprec \frac{(2e)^{ x^2 }}{\sqrt{2\pi}  x}  + \sum_{ k=\lfloor x^2 \rfloor +1 }^l \gdprec \\
	&\leq \gdprec \frac{(2e)^{ x^2 }}{\sqrt{2\pi} }  + \sum_{ k=\lfloor x^2 \rfloor +1 }^l \gdprec .
\end{align*} 

Then 
\begin{align} 
	\Eerfround &= \err{mul} + \frac{2}{\sqrt{\pi}} \sum_{l=0}^{\Lerf}  \err{div} + \frac{\err{l}}{2l+1}   \notag \\
	&\leq  \frac{\gdprec}{2} +  \frac{2}{\sqrt{\pi}} \sum_{l=0}^{\Lerf} \frac{\gdprec}{2} +  \frac{1}{2l+1} \left( \gdprec \frac{(2e)^{ x^2 }}{\sqrt{2\pi}  }  + \sum_{ k=\lfloor x^2 \rfloor +1 }^l  \gdprec \right) \notag \\
	&= \frac{\gdprec}{2} + \frac{2}{\sqrt{\pi}} \gdprec \left( \frac{\Lerf+1}{2} +  \sum_{l=0}^{\Lerf} \frac{1}{2l+1}  \frac{(2e)^{ x^2 }}{\sqrt{2\pi} }  +  \sum_{l=0}^{\Lerf} \frac{1}{2l+1}  \sum_{k=\lfloor x^2 \rfloor +1}^{l} 1  \right)   \notag \\
	&= \frac{\gdprec}{2} + \frac{2}{\sqrt{\pi}} \gdprec \left( \frac{\Lerf+1}{2} +  \sum_{l=0}^{\Lerf} \frac{1}{2l+1}  \frac{(2e)^{ x^2 }}{\sqrt{2\pi} }  +  \sum_{l=\lfloor x^2 \rfloor +1}^{\Lerf} \frac{l - \lfloor x^2 \rfloor}{2l+1}   \right)  \notag \\
	&\leq \frac{\gdprec}{2} +  \frac{2}{\sqrt{\pi}} \gdprec \left(\frac{\Lerf+1}{2} + \sum_{l=0}^{\Lerf}  \frac{(2e)^{ x^2 }}{\sqrt{2\pi} }  +  \sum_{l=\lfloor x^2 \rfloor +1}^{\Lerf} \frac{l - \lfloor x^2 \rfloor}{2l} \right)  \notag \\
	&= \frac{\gdprec}{2} +  \frac{2}{\sqrt{\pi}} \gdprec \left(\frac{\Lerf+1}{2} +  (\Lerf+1) \frac{(2e)^{ x^2 }}{\sqrt{2\pi}  }  +  \sum_{l=\lfloor x^2 \rfloor +1}^{\Lerf} 1- \frac{ \lfloor x^2 \rfloor}{2l} \right)  \notag \\
	&\leq \frac{\gdprec}{2} +  \frac{2}{\sqrt{\pi}} \gdprec \left(\frac{\Lerf+1}{2} +(\Lerf+1) \frac{(2e)^{ x^2 }}{\sqrt{2\pi}  }  +  \sum_{l= \lfloor x^2 \rfloor +1 }^{\Lerf} 1 \right)  \notag \\
	&\leq \frac{\gdprec}{2} +  \frac{2}{\sqrt{\pi}} \gdprec \left(\frac{\Lerf+1}{2} + (\Lerf+1) \frac{(2e)^{ x^2 }}{\sqrt{2\pi}  }  +  \Lerf+1 \right) \notag \\
	&\leq  \frac{\gdprec}{2} + \frac{2(\Lerf+1)}{\sqrt{\pi}} \gdprec \left(\frac{1}{2} + \frac{(2e)^{ x^2 }}{\sqrt{2\pi}}  +  1  \right)  \notag \\ 
	&\leq  \frac{\gdprec}{2} +  \frac{2(\Lerf+1)}{\sqrt{\pi}} \gdprec \frac{\sqrt{2}(2e)^{ x^2 }}{\sqrt{\pi}  }  \notag \\ 
	&=  \left( \frac{1}{2} +  \frac{\sqrt{8}(\Lerf+1)(2e)^{x^2}}{\pi } \right) \gdprec  \notag \\ 
	&\leq  \frac{2\sqrt{8}(\Lerf+1)(2e)^{x^2}}{\pi}  \gdprec. \notag
\end{align}
By plugging Equation \eqref{eq:erf.minL} we have 
\begin{align} 
	\Eerfround &\leq   \frac{2\sqrt{8}(2.36\ln(1/B) +  2.15)(2e)^{x^2}}{\pi }  \gdprec \notag \\
	&\leq     \frac{(13.36\ln(1/B) +  12.17)(2e)^{(\xminerfc)^2}}{\pi }  \gdprec \notag \\
	&\leq \frac{(13.36\ln(1/B) +  12.17)(2e)^{\ln(1/B)/2 + 0.788}}{\pi  }  \gdprec  \notag \\
	&\leq \frac{(13.36\ln(1/B) +  12.17) (2e)^{\log_{2e}(1/B)\ln(2e)/2}(2e)^{0.788}}{\pi  }  \gdprec  \notag \\
	&\leq \frac{3.8(13.36\ln(1/B) +  12.17) (1/B)^{\ln(2e)/2}}{\pi  }  \gdprec  \notag \\
	&\leq \frac{50.8\ln(1/B) +  48.3}{\pi   B^{0.85}} \gdprec 
	\label{eq:erf.round}. 
\end{align}

The $\erfc$ case requires a similar analysis. To compute error of terms $m_0, \dots, m_{\Lerfc}$ defined in the ZKP we take into account the error propagation and the error of our finite precision.  Let now $\errerfc{i}$ be the absolute error of the term $m_i$. We have that $\errerfc{0} = 0$  and 
\begin{align*} 
	m_{i+1} \pm \errerfc{i+1} &= \frac{(m_i \pm \errerfc{i} ) (2i+1)}{2x^2} \pm \err{div}  \\
	&= \frac{m_i (2i+1)}{2x^2} \pm \left( \frac{\errerfc{i}  (2i+1)}{2x^2} + \frac{\gdprec}{2} \right) \\
	&= m_{i+1} \pm \left( \frac{\errerfc{i}  (2i+1)}{2x^2} +\frac{\gdprec}{2} \right). 
\end{align*} 
Hence, 
\begin{align*}  
	\errerfc{i+1} &=  \frac{\errerfc{i}  (2i+1)}{2x^2} + \frac{\gdprec}{2}.
\end{align*} 
In $\erfc$,  $i  < \Lerfc =  \lfloor (\xminerfc)^2  + 1/2 \rfloor$ therefore
\[
\frac{\errerfc{i}  (2i+1)}{2x^2} + \frac{\gdprec}{2} \leq \errerfc{i} + \frac{\gdprec}{2}. 
\]
Then we have that $\errerfc{i} \leq i \frac{\gdprec}{2}$. The total error is 
\begin{align}
	\Eerfcround &= \frac{e^{-x^2}}{x\sqrt{\pi}} \sum_{i=0}^{\Lerfc-1} \errerfc{i} \notag \\
	&=  \frac{e^{-x^2}}{x\sqrt{\pi}} \sum_{i=0}^{\Lerfc-1} i  \frac{\gdprec}{2}  \notag \\
	&=  \frac{e^{-x^2}}{x\sqrt{\pi}}  \frac{(\Lerfc-1)\Lerfc}{2}  \frac{\gdprec}{2}.  \notag
\end{align}

Plugging  Equation \eqref{eq:erfc.minL} to the above we have 
\begin{align} 
	\Eerfcround &\leq   \frac{(0.5\ln(1/B) + 1.288)^2}{4 e^{x^2}x\sqrt{\pi}}\gdprec \notag \\
	&\leq    \frac{0.25\ln^2(1/B) + 1.288\ln(1/B) + 1.659 }{4 e^{(\xminerfc)^2} \sqrt{\pi}}\gdprec \notag \\
	&\leq   \frac{0.25\ln^2(1/B) + 1.288\ln(1/B) + 1.659 }{4 e^{\ln(1/B)/2 + 0.788} \sqrt{\pi}}\gdprec \notag \\
	&\leq   \frac{0.25\ln^2(1/B) + 1.288\ln(1/B) + 1.659 }{8 \sqrt{1/B}  \sqrt{\pi}}\gdprec .
	\label{eq:erfc.round}
\end{align}

Now we determine what value of $\gdprec$ is needed. Equations \eqref{eq:erf.round} and \eqref{eq:erfc.round} fix our requirements to 
\[
\Eerfround  \leq   \frac{50.8\ln(1/B) +  48.3}{\pi   B^{0.85}} \gdprec  \le \frac{B}{2}
\]
and to 
\[ 
\Eerfcround \leq \frac{0.25\ln^2(1/B) + 1.288\ln(1/B) + 1.659 }{8 \sqrt{1/B}  \sqrt{\pi}}\gdprec  \leq \frac{B}{2}.
\] 
$\Eerfround$ imposes the biggest constraint to $\gdprec$, 
as it requires that 
\[
\gdprec  \le \frac{\pi B^{1.85}}{ 101.6\ln(1/B) +  96.6} = O\left(\frac{B^{1.85}}{\ln(1/B)}\right).
\] 
Typically, one would like the total error (due to approximation and
rounding) to be negligible with respect to the standard deviation $\sigma_\eta$,
so one could choose $B=\sigma_\eta / 10^6 |U^H|$.

\paragraph{Computation Cost.}  We now evaluate the computational cost of the proof. 
When we say a task has ``cost $c$'', it means that  requires at most $c$ cryptographic 
computations for generating the proof, $c$ for another party to verify it, and the exchange of  $c\groupelsize$ bits, where 
$\groupelsize$ is the size in bits of an element of $\G$. 

The main statement of the ZKP is
\begin{equation} \\
	\label{eq:zkp.erf.erfc}
	\bigg\{ \bigg(x \in \bigg[0, \bigg\lfloor \frac{ y^{\mathrm{erfc}}_{min}}{\gdprec} \bigg\rfloor \bigg] \land y = \erf(x)\bigg) \lor \bigg(y \in \left[\bigg\lfloor \frac{ y^{\mathrm{erfc}}_{min}}{\gdprec} \bigg\rfloor +1,  \frac{1}{\gdprec} \right]   \land 1-y = \erfc(x)\bigg) \bigg\}
\end{equation}
where $y^{\mathrm{erfc}}_{min} = \erf\left(x_{\mathrm{erfc}}^{min}\right)$ is a public constant. 
The main costs are in the proofs of $\erf$ and $\erfc$.

Proving computations of $\erf$ in Equation \eqref{eq:erf.zkp} requires $3$ range proofs of cost 
of at most $10\log_2(1/\gdprec)$, 
$10\log_2(l)$ and $10\log_2(2l+1)$.  As 
$l \leq \Lerf = 2.36\ln(1/B)$, the cost of evaluating a term is 
$10\log_2(1/\gdprec) + 20\log_2(\ln(1/B))$. We evaluate $\Lerf$ terms. The total cost is 
\begin{align*} 
	& \Lerf (10\log_2(1/\gdprec) + 20\log_2(\ln(1/B)))  = 10\Lerf \log_2(1/\gdprec)\Lerf + 20\log_2(\ln(1/B)).
\end{align*} 
The dominating term above is 
\[ 10\log_2(1/\gdprec)\Lerf = 23.6\log_2(1/\gdprec) \ln(1/B) < 34.1\ln(1/\gdprec)\ln(1/B). \]
The for $\erfc$, we require $\Lerfc$ proofs of the computation in Equation \eqref{eq:erfc.zkp}, one for each term, and its cost is dominated by
\begin{align*} 
	10\log_2( q ) \Lerfc = 10\log_2( q)  \left\lfloor  0.5\ln(1/B) + 1.288 \right\rfloor. 
\end{align*}   
Neglecting lower order constants, the above is dominated by
\begin{align*}
	10\log_2(q) 0.5\ln(1/B) &= 5\log_2(e)\ln(q) \ln(1/B) < 7.22\ln(q) \ln(1/B).
\end{align*} 
 The total cost is the sum of the costs of the $\erf$ and $\erfc$ ZKPs, which is dominated by
\[
\ln(1/B) \cdot \left( 34.1\ln(1/\gdprec) + 7.22\ln(q) \right).
\]

\end{document}